\documentclass[12pt]{article}

\usepackage[left=1in,top=1in,right=1in,bottom=1in,head=.1in,nofoot]{geometry}

\usepackage{setspace,url,bm,amsmath} 

\newcommand{\Gu}[1]{{\color{blue}{#1}}}
\usepackage{titlesec} 
\titlelabel{\thetitle.\quad} 
\usepackage{mathtools}
\usepackage{enumitem}
\usepackage[margin=20pt]{subcaption}
\usepackage{minitoc}
\usepackage{amsmath,amsfonts,amsthm,amssymb}
\usepackage{graphicx}
\newcommand{\cmark}{\ding{51}}%
\newcommand{\xmark}{\ding{55}}%
\usepackage{xcolor}
\definecolor{darkblue}{rgb}{0.0, 0.0, 0.55}
\usepackage[bookmarks, plainpages = false, colorlinks = true, citecolor = darkblue, linkcolor = chicago-maroon, anchorcolor = red, urlcolor = chicago-maroon]{hyperref}
\definecolor{chicago-maroon}{RGB}{128,0,0}
\definecolor{northwestern-purple}{RGB}{82,0,99}
\usepackage{threeparttable}
\usepackage{manyfoot}
\usepackage{makecell}
\usepackage{booktabs}%
\usepackage{algorithm}%
\usepackage{algorithmicx}%
\usepackage{algpseudocode}%
\usepackage{listings}%
\usepackage{dsfont}%
\usepackage{bm, bbm}
\usepackage{etoolbox}
\usepackage{natbib}
\usepackage{adjustbox}
\usepackage{threeparttable}
\usepackage{dcolumn}

\theoremstyle{definition}

\DeclareFontFamily{U}{mathx}{\hyphenchar\font45}
\DeclareFontShape{U}{mathx}{m}{n}{
      <5> <6> <7> <8> <9> <10>
      <10.95> <12> <14.4> <17.28> <20.74> <24.88>
      mathx10
      }{}
\DeclareSymbolFont{mathx}{U}{mathx}{m}{n}
\DeclareFontSubstitution{U}{mathx}{m}{n}
\DeclareMathAccent{\widecheck}{0}{mathx}{"71}
\DeclareMathAccent{\wideparen}{0}{mathx}{"75}

\usepackage{bibunits}
\def\Levy{L\'{e}vy}
\def\diff{\mathrm{d}}

\allowdisplaybreaks[4]

\newtheorem{assumption}{Assumption}
\newtheorem{theorem}{Theorem}
\newtheorem{lemma}{Lemma}
\newtheorem{remark}{Remark}

\newtheorem{proposition}{Proposition}

\newcommand{\op}{\mathrm{op}}
\newcommand{\tr}{\mathrm{tr}}
\newcommand{\norm}[2]{\lVert#1\rVert_#2}
\newcommand{\B}[1]{#1_{[k]}}
\newcommand{\ST}[1]{#1_{[k]1}}
\newcommand{\SC}[1]{#1_{[k]0}}
\newcommand{\ProjBeta}[1]{\beta_{[k]}(#1)}

\newcommand{\var}{\mathrm{var}}

\newcommand{\Sigmak}[1]{\Sigma_{[k]#1}}
\newcommand{\HSigmak}[1]{\tilde{\Sigma}_{[k]#1}}
\newcommand{\Sigmakx}{\Sigma_{[k]}}
\newcommand{\HSigmakx}{\hat{\Sigma}_{[k]}}

\newcommand{\tilx}{\tilde{X}}

\newcommand{\sgij}{\overline{G}_{ij}}
\newcommand{\Res}{\mathrm{Res}}
\newcommand{\SigmaX}{\Sigma}
\newcommand{\HSigmaX}{\hat{\Sigma}}
\newcommand{\zetaSec}[1]{\zeta^2_{\mathrm{II},#1}}
\newcommand{\HzetaSec}[1]{\hat{\zeta}^2_{\mathrm{II},#1}}
\def\Holder{H\"{o}lder}

\def\hat{\widehat}
\def\tilde{\widetilde}
\def\Ek{E_{[k]}}
\def\Eka{E_{[k]1}}

\def\unadj{\mathrm{unadj}}
\def\oracle{\mathrm{ora}}
\usepackage{pifont}
\usepackage{orcidlink}

\title{Assumption-lean covariate adjustment under covariate adaptive randomization when $p = o (n)$}

\author{Yujia Gu$^1$\thanks{email: \href{guyj@mail.tsinghua.edu.cn}{guyj@mail.tsinghua.edu.cn} The authors are listed in alphabetical order.} \ \ Lin Liu\orcidlink{0000-0002-9883-7962}$^{2}$\thanks{email: \href{linliu@sjtu.edu.cn}{linliu@sjtu.edu.cn} Corresponding author. The authors would like to thank Xin Zhang (Pfizer), \href{https://math.sjtu.edu.cn/faculty/chengwang/index.html}{Cheng Wang} and \href{https://github.com/Cinbo-Wang}{Xinbo Wang} for helpful discussions.} \ \ Wei Ma\orcidlink{0000-0002-2952-7944}$^3$\thanks{email: \href{mawei@ruc.edu.cn}{mawei@ruc.edu.cn}}}

\date{\small
$^1$Department of Statistics and Data Science, Tsinghua University, Beijing, China \\
$^2$Institute of Natural Sciences, MOE-LSC, School of Mathematical Sciences, SJTU-Yale Joint Center for Biostatistics and Data Science, Shanghai Jiao Tong University, Shanghai, China \\
$^3$Institute of Statistics and Big Data, Renmin University of China, Beijing, China \\
\vspace{1em}
\today
}

\begin{document}

\begin{bibunit}[apalike]

\maketitle

\begin{abstract}
\small Adjusting for (baseline) covariates with working regression models becomes standard practice in the analysis of randomized clinical trials (RCT). When the dimension $p$ of the covariates is large relative to the sample size $n$, specifically $p = o (n)$, adjusting for covariates even in a linear working model by ordinary least squares can yield overly large bias, defeating the purpose of improving efficiency. This issue arises when no structural assumptions are imposed on the outcome model, a scenario that we refer to as the \emph{assumption-lean setting}. Several new estimators have been proposed to address this issue. However, they focus mainly on simple randomization under the finite-population model, not covering covariate adaptive randomization (CAR) schemes under the superpopulation model. Due to improved covariate balance between treatment groups, CAR is more widely adopted in RCT; and the superpopulation model fits better when subjects are enrolled sequentially or when generalizing to a larger population is of interest. Thus, there is an urgent need to develop procedures in these settings, as the current regulatory guidance provides little concrete direction. In this paper, we fill this gap by demonstrating that an adjusted estimator based on second-order $U$-statistics can almost unbiasedly estimate the average treatment effect and enjoy a guaranteed efficiency gain if $p = o (n)$. In our analysis, we generalize the coupling technique commonly used in the CAR literature to $U$-statistics and also obtain several useful results for analyzing inverse sample Gram matrices by a delicate leave-$m$-out analysis, which may be of independent interest. Both synthetic and semi-synthetic experiments are conducted to demonstrate the superior finite-sample performance of our new estimator compared to popular benchmarks.
\end{abstract}


\doublespacing

\newpage

\section{Introduction}
\label{sec:intro}

Randomization is widely regarded as the gold standard for evaluating treatment effects from clinical trials and interventional studies in other disciplines. Although simple randomization is straightforward to implement in practice and is optimal in theory under certain criteria \citep{bai2023randomize}, it can nevertheless lead to imbalances in the distributions of (baseline) covariates in different treatment groups and, in turn, to accidental bias and loss of efficiency in estimating treatment effects \citep{student1938comparison, efron1971forcing, kasy2016experimenters, banerjee2020theory}. 

As a remedy, randomization schemes that leverage the information in covariates, in particular covariate adaptive randomization (CAR), are now routinely adopted in practice, especially in clinical settings \citep{Lin2015}. When balancing stratified covariates, stratified block randomization \citep{zelen1974randomization} and minimization \citep{Taves1974,Pocock1975} are commonly deployed in randomized clinical trials (RCT). Stratified block randomization defines sets of strata of a subset of covariates and allocates units within each stratum using block randomization. This method is widely used in the design of RCT, accounting for approximately 70\% of all cases \citep{ciolino2019ideal, Lin2015}. Minimization was initially designed to balance covariates over their marginal distributions \citep{Taves1974, Pocock1975}, but it has also been generalized to control various types of imbalance measures, including marginal and/or within-stratum imbalances \citep{hu2012asymptotic}. Other randomization strategies that fall under the purview of CAR include stratified biased coin design \citep{efron1971forcing} and various model-based approaches \citep{Atkinson1982,Begg1980}. A rather comprehensive discussion of these randomization schemes in RCT can be found in \citet{Rosenberger2015}. We also refer those who are interested in the use of CAR in disciplines beyond medicine to \citet{duflo2007using}.

Beyond balancing covariates in the design stage, it is now well accepted that adjusting for covariates in the analysis stage can also improve the estimation efficiency \citep{leon2003semiparametric, zhang2008improving, tsiatis2008covariate, ma2024new, van2024covariate}. In fact, such practice has been encouraged by two quite impactful regulatory guidelines \citep{EMA2015, FDA2023}. The robustness of \emph{(covariate-)adjusted estimators} is highly valued and critical for their acceptance into routine practice because the data-generating process in RCT is typically unknown beyond the treatment assignment mechanism. In general, it is required that the bias of any adjusted estimator, regardless of its source (e.g., model misspecification or curse of dimensionality), must be dominated by sampling variability of order $n^{-1 / 2}$. This is due to the availability of a $\sqrt{n}$-consistent and asymptotically normal ($\sqrt{n}$-CAN) estimator that adjusts for no covariates in the analysis stage, thus imposing no structural assumptions on the outcome model. 

We now review several existing covariate adjustment methods that are relevant to, but not restricted to, CAR. Consider a trial enrolling $n$ subjects. For simple randomization (without balancing any covariates in the design stage), several studies have explored the properties of regression-adjusted estimators \citep{freedman2008regressiona, Freedman2008, lin2013}. However, under covariate-adaptive randomization, the robustness of these estimators becomes more challenging to probe due to dependencies in treatment assignments. \citet{Bugni2018,Bugni2019} proposed a model-assisted approach that adjusts for stratification covariates. The method proposed there ensures valid statistical inference that is robust to potential misspecification of the outcome model. To adjust for additional covariates and improve efficiency, recent work developed stratum-common and stratum-specific estimators for settings where the covariate dimension $p$ is fixed \citep{ma2022regression, ye2023toward, gu2023regression}. For high-dimensional regimes ($p \gg n$), \citet{liu2022lasso} introduced lasso-adjusted estimators that ensure both robustness and efficiency gains.

Despite the above progress, there remains an important gap in the theoretical underpinnings of how to construct adjusted estimators of treatment effects, when the dimension $p$ of adjusted covariates is moderately high relative to $n$, or more precisely $p = o (n)$ and $p < n$, without imposing any structural assumptions on the outcome model, such as sparsity as in \citet{Bloniarz2016} or low complexity (quantified by the metric entropy integral) as in \citet{guo2023generalized}. We coin the setting in which no structural assumptions are imposed on the outcome model as the \emph{assumption-lean setting}, borrowing a terminology recently popularized in the causal inference literature \citep{vansteelandt2022assumption}. Under simple randomization and the design-based framework (or equivalently, the finite-population model), various bias-corrected adjusted estimators with linear working models have been proposed to reduce bias due to large $p$ \citep{lei2021regression, chang2024exact, lu2025debiased, zhao2024hoif, chiang2025regression}. As delineated in \citet{zhao2024hoif}, all the above bias-corrected adjusted estimators can essentially be interpreted as estimators based on a particular type of $U$-statistics motivated by the theory of higher-order influence functions \citep{liu2017semiparametric}.

When it comes to the more general CAR designs that are the main focus of this paper, the relevant literature is scant. We also deviate from the above literature by adopting the superpopulation model to better fit RCT, in which subjects are often enrolled sequentially, and one is often interested in generalizing the conclusions drawn from an RCT to a larger population. To our knowledge, \citet{jiang2025adjustments} is the only work investigating the statistical properties of adjusting for high-dimensional covariates up to $p = O (n)$. In particular, they demonstrate that the standard OLS adjusted estimator with linear working models \citep{ma2022regression, ye2023toward} is $\sqrt{n}$-CAN and has a guaranteed efficiency gain compared to the unadjusted estimator as long as $p \ll \sqrt{n}$. However, when $p \gtrsim \sqrt{n}$, \citet{jiang2025adjustments} assumed a correct linear relationship between the outcome and the covariates to preserve the same statistical guarantee using the OLS adjusted estimator. As will be shown in the sequel, the OLS estimator is not robust against model misspecification once $p \gtrsim \sqrt{n}$. A more detailed comparison of our work with the articles mentioned above is provided in Remark~\ref{rem:comparison}. In summary, it remains elusive whether one can construct a $\sqrt{n}$-CAN adjusted estimator of the average treatment effect (ATE) under the assumption-lean setting and the superpopulation model for CAR, and this research gap is reflected in the following statement in the most recent FDA guideline on covariate adjustment \citep{FDA2023}:
\begingroup
\begin{quote}
\emph{The statistical properties of covariate adjustment are best understood when the number of covariates adjusted for in the study is small relative to the sample size (Tsiatis et al., 2008). Therefore, sponsors should discuss their proposal with the relevant review division if the number of covariates is large relative to the sample size or if proposing to adjust for a covariate with many levels (e.g., study site in a trial with many sites).}
\end{quote}
\label{quotation}
\endgroup


\paragraph{Main contributions and organization}

Our main objective is to address the gap targeted by the above quote; we refer the interested readers to Section~\ref{sec:conclusion} for our concrete recommendation. In particular, our paper advances the methodological and theoretical understanding of covariate adjustment on several fronts, while providing insights and guidance on its implementation in practice. Our contributions are summarized as follows.

\begin{enumerate}[label = \arabic*)]
\item Methodologically, we propose a new adjusted estimator of the ATE under CAR and the assumption-lean setting. Our new estimator is essentially a second-order $U$-statistic. In particular, this new estimator is $\sqrt{n}$-CAN, and has guaranteed efficiency gain compared to the unadjusted estimator (Theorem~\ref{thm:main}), as long as $p = o (n)$ in the assumption-lean setting, only under some additional mild tail assumptions on the covariates $X$ and outcome $Y (a)$; see Assumptions~\ref{ap1}--\ref{ap3}. Thus, our result resolves an important gap in the literature on covariate adjustment in randomized experiments. Furthermore, a consistent variance estimator of our new adjusted estimator of the ATE is also constructed to facilitate downstream hypothesis testing and statistical inference (Theorem~\ref{thm:var}).

\item Theoretically, new technical results on the coupling technique adopted in \citet{Bugni2018, Bugni2019} for CAR are established to analyze $U$-statistics to prove the statistical properties of our new estimator; see Lemma~\ref{lem:BG}. Furthermore, compared to the randomization-based framework \citep{zhao2024hoif}, the proof complexity under the superpopulation model is dramatically escalated due to the randomness in the covariates, demanding a fairly involved \emph{leave-two-out} analysis of inverse sample Gram matrices; see Lemma~\ref{lem:repeated Sherman-Morrison}. These new technical results could be of independent interest.

\item Empirically, through extensive simulation studies and a semi-synthetic data analysis, we demonstrate that our new estimator exhibits highly competitive performance in a variety of data generating processes, compared to popular benchmark methods. Based on both theoretical and empirical results, we provide parallel practical guidance that complements the most recent FDA guideline, in particular targeting the paragraph quoted on page~\pageref{quotation}. We have incorporated our new adjusted estimator under CAR into an R package available to download from \href{https://cran.r-project.org/web/packages/HOIFCar/index.html}{CRAN} to make the new method ready to be used in practice.
\end{enumerate}

The remainder of the paper is structured as follows. Section~\ref{sec:assumptions} sets the stage by introducing our statistical problem and the mathematical framework. Both intuitive explanation and numerical evidence are presented in Section~\ref{sec:intuition} to elucidate why the existing adjusted estimators with linear working models under CAR may have an overly large bias. We also formally introduce our new adjusted estimator at the end of Section~\ref{sec:intuition}. Section~\ref{sec:theory} then presents the statistical properties of our new estimator, which are the main theoretical contribution of this paper. 
In Section~\ref{sec:inference}, we construct a consistent variance estimator for our new adjusted estimator to facilitate downstream hypothesis testing and valid inference. Synthetic and semi-synthetic data analyses are conducted in Section~\ref{sec:experiments}, to empirically demonstrate the competitive performance of our new adjusted estimator compared to several benchmarks. 
Section~\ref{sec:conclusion} concludes the paper with a discussion of practical recommendations and some extensions. 
Technical details are deferred to the Appendix.

\section{Framework and Setup}
\label{sec:assumptions}

Suppose that we conduct a randomized experiment with $n$ units (enrolled subjects). For each unit $i$, we use $A_i$ to denote the treatment assignment with $A_i=1$ for the treatment group and $A_i=0$ for the control group. In this paper, we only consider two treatment groups, but it is straightforward to extend all our results to more than two groups. We let $n_1 = \sum_{i=1}^n A_i$ and $n_0 = \sum_{i=1}^n (1-A_i)$ be the number of units in the treated group and the control group, respectively. Since we are mainly interested in CAR, we further introduce $\{B_i\}_{i = 1}^{n}$ to denote the strata labels that are balanced in the design stage and each $B_{i}$ takes values in $\mathcal{K} = \{1,2,\dots,K\}$, where $K$ represents the total number of strata, assumed to be fixed throughout. We also collect $p$-dimensional covariates $\{X_i\}_{i = 1}^{n}$ that are not used in the design stage. We let $p_{[k]} = \Pr(B_i=k)$ be the target proportion of stratum $k$ and $\pi_{[k]} = \Pr(A_i=1|B_i=k)$ be the target proportion of treatment group in stratum $k$. We use subscript $[k]$ to indicate statistics calculated in stratum $k$, e.g. $n_{[k]} = |[k]|$, $n_{[k]1} = \sum_{i\in [k]} A_i$ and $n_{[k]0} = \sum_{i\in [k]} (1-A_i)$ representing the numbers of all units, treated units and control units in stratum $k$, respectively. Here, $i \in [k]$ indexes units in stratum $k$. Let $p_{n[k]} = n_{[k]}/n$ and $\pi_{n[k]} = \ST{n}/\B{n}$ be the proportion of stratum size and the treated group in stratum $k$. 

We adopt the potential outcomes notation to define treatment effects. Let $Y_i(0)$ and $Y_i(1)$ be the potential outcomes. By the standard consistency assumption, the observed outcome is $Y_i = A_i Y_i(1) + (1-A_i)Y_i(0)$. For any random variable $V$, let $\sigma^2_{[k]V} = \var(V|B = k)$ denote the variance of $V$ in stratum $k$. Let $\Sigmakx = \Ek\{X X^{\top}\}$ be the population-level Gram matrix of the covariates $X$ in stratum $k$, where $\Ek (\cdot) = E(\cdot|B=k)$ denotes the conditional mean in stratum $k$. We use $\eta_{[k]}(a) = \Ek\{X Y (a)\}$ to denote the inner product between $X$ and $Y (a)$, for $a\in\{0,1\}$. For a generic random variable $r (a)$ related to the treatment group $A = a$, let $\bar{r}_{[k]a} = n_{[k]a}^{-1} \sum_{i\in[k]} \mathbbm{1} \{A_i = a\} r_i(a)$, for $a \in \{0, 1\}$, be the corresponding sample mean in stratum $k$. For example, $r (a)$ may correspond to the potential outcomes $Y (a)$, the covariate $X$, or their product $X \cdot Y (a)$. Also, let $\HSigmakx = n_{[k]}^{-1} \sum_{i\in[k]}X_i X_i^{\top}$ be the sample Gram matrix of the covariates $X$ in stratum $k$.

Let $W_i = \{Y_i(0),Y_i(1),B_i,X_i\}_{i=1}^n$ be independent and identically distributed (i.i.d.) samples from the population $W = \{Y(0),Y(1),B,X\}$. Under CAR, the ATE is identified as
\begin{align*}
\tau = E \{Y (1) - Y (0)\} = \sum\limits_{k =1}^K p_{[k]} \Ek\{Y (1) - Y (0)\} = \sum\limits_{k=1}^K p_{[k]}\tau_{[k]},
\end{align*}
where $\tau_{[k]}$ is the ATE in stratum $k$. Our goal is to construct point and interval estimators of $\tau$ based on the non-i.i.d.\ observed data $\{Y_i,B_i,X_i,A_i\}_{i=1}^n$, where the dependency is induced by CAR. Throughout, we impose the following on the treatment assignment mechanism.

\begin{assumption}\label{ap1}\ 
\begin{enumerate}
    \item Conditional on $\{B_1,\dots,B_n\}$, $\{A_1,\dots,A_n\}$ is independent of $\{Y_i(0),Y_i(1),X_i\}_{i=1}^n$.
    \item $\pi_{n[k]} \stackrel{P}\to \pi_{[k]}$ as $n \to \infty$, for all $k \in \mathcal{K}$.
\end{enumerate}
\end{assumption}

Assumption~\ref{ap1} is also imposed in \citet{Bugni2019}, \citet{ma2022regression} and \citet{liu2022lasso}, which is satisfied by several commonly used CAR procedures, including simple randomization, stratified block randomization \citep{zelen1974randomization} and stratified biased coin randomization \citet{efron1971forcing}. When $\B{\pi}$'s are the same across different strata, Pocock and Simon's minimization method \citep{Pocock1975} also satisfies this assumption \citep{Ma2015}. 

Before formally presenting our new covariate-adjusted ATE estimator (abbreviated as adjusted estimator for short), we first introduce the stratified difference-in-means estimator (referred to as the \emph{unadjusted estimator} here, as it does not adjust for additional covariates $X_i$), which is, to the best of our knowledge, first proposed in \citet{Bugni2019}:
\begin{equation}
\label{unadj}
\hat\tau_{\rm unadj} = \sum\limits_{k=1}^K p_{n[k]} \big(\ST{\bar Y} -\SC{\bar Y} \big).
\end{equation}
This unadjusted estimator $\hat{\tau}_{\unadj}$ serves as a benchmark against which both existing adjusted estimators and our newly proposed adjusted estimator are compared. Importantly, \citet{Bugni2019} showed that $\hat{\tau}_{\unadj}$ $\sqrt{n}$-CAN, without any assumptions on the outcome model.

Despite the advantages mentioned above, $\hat{\tau}_{\unadj}$ does not fully incorporate all available information from the data, therefore resulting in a loss of efficiency. We next introduce the OLS estimator $\hat\tau_{\rm OLS}$ proposed by \citet{ma2022regression}, which is commonly used in practice to improve efficiency compared to the unadjusted estimator when the covariate dimension $p$ is fixed, defined as
\begin{equation}
\label{OLS}
\begin{split}
\hat\tau_{\rm OLS} = & \ \sum\limits_{k=1}^K p_{n[k]} \Big( \Big\{\bar Y_{[k]1} - \dfrac{1}{n_{[k]1}}\sum\limits_{i\in[k]}(A_i-\pi_{n[k]})X_i^{\top}\hat\beta_{[k]}(1) \Big\} \\
& - \Big\{\bar Y_{[k]0} - \dfrac{1}{n_{[k]0}}\sum\limits_{i\in[k]}(\pi_{n[k]}-A_i)X_i^{\top}\hat\beta_{[k]}(0) \Big\} \Big),
\end{split}
\end{equation}
where
\[\hat\beta_{[k]}(a) = \HSigmakx^{-1} \cdot \frac{1}{n_{[k]a}}\sum\limits_{i\in[k]} \mathbbm{1}\{A_i=a\}X_iY_i,\quad a\in\{0,1\}.\]
A practical advantage of $\hat{\tau}_{\rm OLS}$ is that it can be computed by straightforward post-processing of the output of standard, off-the-shelf statistical software packages. To see this, the OLS estimator can in fact be obtained by first regressing $Y$ against $A$, $X$ and their interactions within each stratum and then taking a weighted sum of the regression coefficients of $A$ across strata, with weights $p_{n[k]}$. Specifically, $\hat\beta_{[k]}(a)$ corresponds to the coefficient of the interaction term between $A$ and $X$ in each stratum.

We now recall the statistical properties of $\hat\tau_{\rm OLS}$ in Proposition~\ref{prop1} below, which is adapted from \citet{ma2022regression}.
\begin{proposition}\label{prop1}
Under Assumption \ref{ap1} and $E\{Y^2_i(a)\}<\infty$, let $\ProjBeta{a} = \Sigmak{}^{-1} \eta_{[k]}(a)$,  $\beta_{[k]} = (1-\pi_{n[k]})\beta_{[k]}(1) + \pi_{n[k]}\beta_{[k]}(0)$ and $r_i(a) = Y_i(a) - X_i^{\top}\beta_{[k]}$. Define $\sigma^2_{\rm OLS} = \zeta^2_{H} + \zeta^2_{\mathrm{I},r}(\pi_{[k]})$. When $p$ is fixed, we have
\[\sqrt{n} (\hat\tau_{\rm OLS} - \tau)/\sigma_{\rm OLS} \stackrel{d}\to \operatorname{N}(0,1),\]
where
\begin{equation}
\label{notation_I}
\zeta^2_{\mathrm{I},r}(\pi_{[k]}) = \sum\limits_{k=1}^K p_{[k]}\Big(\frac{\sigma^2_{[k]r(1)}}{\pi_{[k]}} + \frac{\sigma^2_{[k]r(0)}}{1-\pi_{[k]}}\Big),\ \zeta^2_{H} = E\{\check{r}_i(1) - \check{r}_i(0)\}^2,
\end{equation}
and $\check{r}_i(a) = \Ek\{r_i(a)\} - E\{r_i(a)\}$ is the difference between the mean of $r (a)$ in stratum $k$ and the mean of $r (a)$ in the population.
\end{proposition}

Before moving forward, we unpack the seemingly cumbersome notation $\zeta_{\rm I, r}^{2} (\pi_{[k]})$ that appeared in Proposition~\ref{prop1} above. The subscript $\mathrm{I}$ is used to represent sample means, distinguished from second-order $U$-statistics (demarcated by subscript $\mathrm{II}$) to appear in the next section. The other subscript $r$ indicates that $\zeta_{\rm I, r}^{2} (\pi_{[k]})$ corresponds to the asymptotic variance of the transformed outcome $r$. 
At this stage, it may seem that we are overloading the notation, but it will facilitate the presentation in Section~\ref{sec:inference}, when we discuss how to estimate the variance of our new estimator.


\begin{remark}
\label{rem:total dimension}
In this paper, we simply take $p$ as the dimension of all the covariates. In principle, we allow the dimension of all the covariates to be different from that of the covariates being adjusted. As we will see in Section~\ref{sec:theory}, $p = o (n)$ suffices for our proposed estimator to be asymptotic normal with asymptotic variance no greater than that of the unadjusted estimator. Therefore, a natural next step of our work is to develop estimators that can incorporate data-adaptive variable selection methods \citep{van2024automated} that adaptively choose $o (n)$ many variables if the total number of covariates is greater than $n$.
\end{remark}

\section{Bias-Corrected Estimator Based on \texorpdfstring{$U$}{}-Statistics}
\label{sec:intuition}

In the previous section, in particular in Proposition~\ref{prop1}, we have seen that the OLS estimator $\hat\tau_{\rm OLS}$ can be more efficient than the unadjusted estimator $\hat{\tau}_{\unadj}$, when the dimension $p$ of adjusted covariates is fixed or, more generally, when $p \ll \sqrt{n}$. However, in the assumption-lean setting, even when $p$ is moderately high in the sense that $p = o (n)$, the conclusions of Proposition~\ref{prop1} no longer hold and $\hat{\tau}_{\rm OLS}$ can suffer from non-negligible bias. If the bias of $\hat{\tau}_{\rm OLS}$ exceeds or even just equals $n^{-1 / 2}$ in order, it defeats the purpose of improving estimation efficiency by adjusting for covariates, because downstream hypothesis testing or statistical inference based on standard Wald confidence intervals will be invalid. 

We now explain the limitations of $\hat\tau_{\rm OLS}$ when $p$ is close to $n$. Recall that the OLS estimator takes the form in \eqref{OLS}. Suppose that the population-level Gram matrix $\Sigmakx$ is known and we temporarily replace $\HSigmakx$ in $\hat{\beta}_{[k]} (a)$  by $\Sigmakx$, which appeared in $\hat\tau_{\rm OLS}$ defined in \eqref{OLS}. Using the treatment group as an example, by spelling out $\hat{\beta}_{[k]} (a)$ for $a \in \{0, 1\}$, we can rewrite the linear adjustment term (or sometimes interpreted as the augmentation term) of $\hat{\tau}_{\rm OLS}$ in stratum $k$ as a second-order $V$-statistic
\[\frac{1}{n^2_{[k]1}}\sum\limits_{i\in[k]}\sum\limits_{j\in [k]}(A_i-\pi_{n[k]})X_i^{\top}\Sigmakx^{-1} A_j X_jY_j.\]
Based on this $V$-statistic representation, it is not difficult to see that it introduces a bias of order $p / n$ through the diagonal component in the double summation:
\begin{equation}
\label{bias OLS}
\frac{1}{n_{[k]1}^2}\sum\limits_{i\in[k]}(A_i-\pi_{n[k]})X_i^{\top}\Sigmakx^{-1} A_i X_iY_i = \frac{1-\pi_{n[k]}}{n_{[k]1}^2}\sum\limits_{i\in[k]}A_iX_i^{\top}\Sigmakx^{-1}X_iY_i =O_P\left(\frac{p}{n}\right),
\end{equation}
where the last equality follows because each summand $|A_iX_i^{\top}\Sigmakx^{-1}X_iY_i| \lesssim \Vert X_i \Vert_{2}^{2} \Vert \Sigmakx \Vert_{\rm op} = O_{P} (p)$, if we assume that each dimension of $X_{i}$ is of order $O_{P} (1)$ and $\Sigmakx$ has bounded eigenvalues (see Assumptions~\ref{ap2} and \ref{ap3} in Section~\ref{sec:theory}).

In Figure \ref{fig:bias}, we numerically illustrate how the bias of $\hat\tau_{\rm OLS}$, analytically characterized in \eqref{bias OLS}, scales with $p$. For comparison, we also report the bias of an oracle bias-corrected estimator to be introduced later, based on the oracle knowledge of $\Sigma_{[k]}$. The simulated data are generated according to a similar setting as Model~1 in Section~\ref{sec:sim} later and we defer the details to Appendix~\ref{app:illustration}. In Figure \ref{fig:bias}(a), we gather both the (absolute) empirical bias of $\hat{\tau}_{\rm OLS}$ calculated by taking the average of $\hat{\tau}_{\rm OLS} - \tau$ over repeated Monte Carlo draws and the analytical bias calculated using \eqref{bias OLS}, confirming that the absolute bias of $\hat{\tau}_{\rm OLS}$ indeed increases substantially with $p$ and tracks the analytic formula \eqref{bias OLS} closely. In contrast, the oracle bias-corrected estimator is nearly unbiased, with an absolute bias near zero across all $p$'s considered. We then display in Figure \ref{fig:bias}(b) the histograms of both $\hat{\tau}_{\rm OLS}$ and the oracle bias-corrected estimator respectively at $p=5$ and $p=60$. We rescaled both estimators by subtracting the true effect and dividing their corresponding Monte Carlo standard deviations. When $p=5$, the two estimators exhibit similar performance due to the small bias. However, when $p=60$, the OLS estimator exhibits a substantial bias, whereas the oracle bias-corrected estimator remains nearly unbiased.

\begin{figure}[H]
\centering
\includegraphics[width=0.95\linewidth]{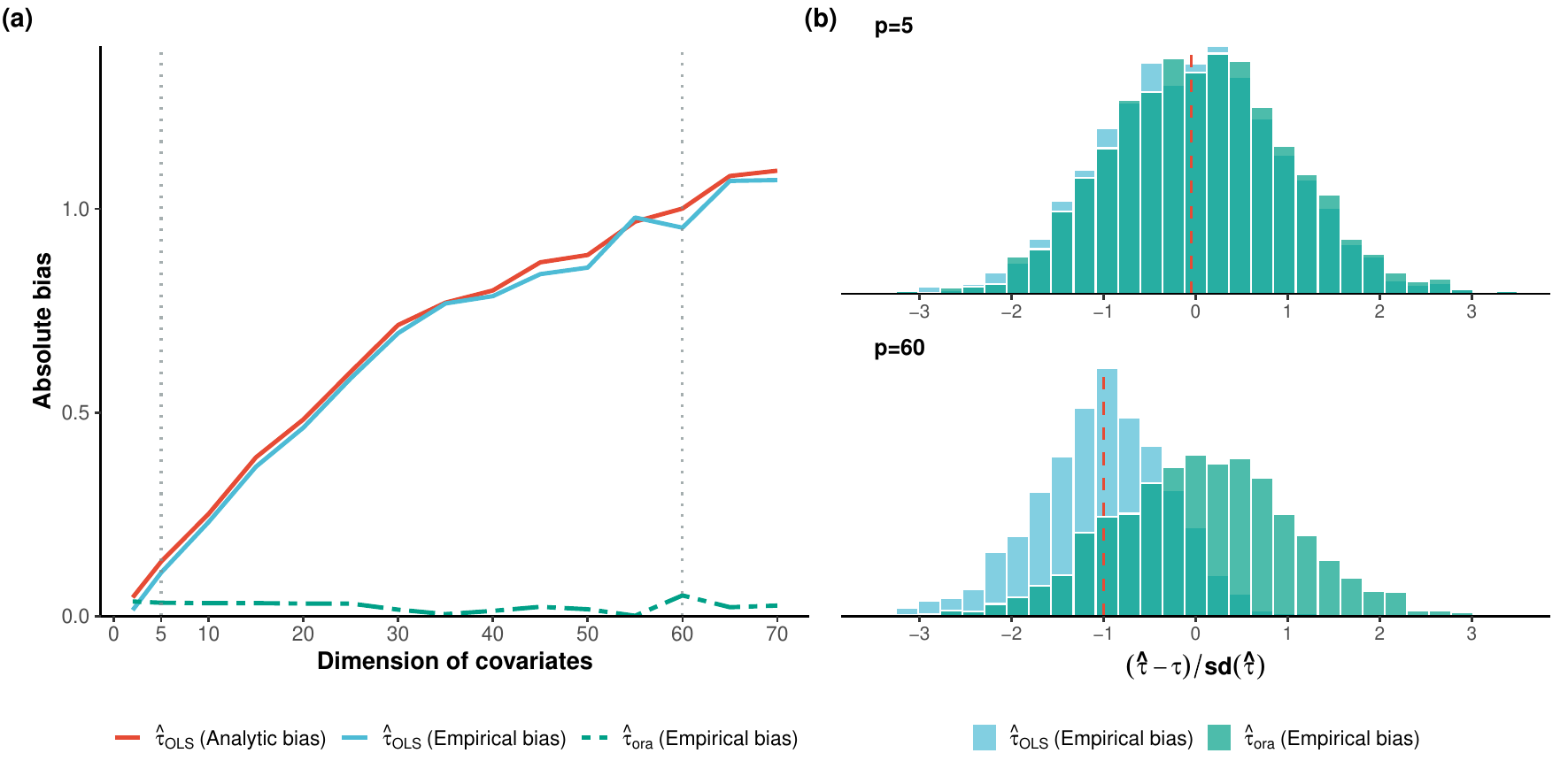}
\caption{(a) Absolute bias of the estimators relative to the true effect as the covariate dimension $p$ increases, averaged over 2000 Monte Carlo replicates. The covariate dimension refers to the set of variables adjusted in the estimators. The oracle estimator $\hat{\tau}_{\oracle}$ will be introduced later in this section. The vertical dotted lines mark the $p$'s at which the histograms of estimators are displayed in panel (b). (b) Histograms of the scaled estimators for two representative cases, $p=5$ and $p=60$, where scaling is performed by dividing each estimator by its Monte Carlo standard deviation computed from 2000 replicates.}
\label{fig:bias}
\end{figure}

The above observation prompts the need to develop new estimators that have negligible bias compared to sampling variability when $p$ is large compared to $n$ under the assumption-lean setting. To introduce our new estimator, we start by still assuming the oracle knowledge of $\Sigmakx$. We then reveal the oracle estimator reported in Figure~\ref{fig:bias}, which is a second-order $U$-statistic and removes the bias of $\hat{\tau}_{\rm OLS}$. First, for each stratum $k$, we compute two separate second-order $U$-statistics, respectively, in the treatment and control groups:
\begin{equation}
\label{U-notation}
\begin{split}
& \mathbf{U}_{n_{[k]},2}(\Sigmakx^{-1}; 1) = \frac{1}{\B{n}(\B{n}-1)}\frac{1}{\pi^2_{n[k]}} \sum_{\substack{1\leq i\neq j \leq n\\ i,j\in[k] }} (A_i - \pi_{n[k]})X_i^{\top} \Sigmakx^{-1} A_{j} X_j Y_j, \\
& \mathbf{U}_{n_{[k]},2}(\Sigmakx^{-1};0) = \frac{1}{\B{n}(\B{n}-1)}\frac{1}{(1-\pi_{n[k]})^2} \sum_{\substack{1\leq i\neq j \leq n\\ i,j\in[k] }} (\pi_{n[k]}-A_i)X_i^{\top} \Sigmakx^{-1} (1-A_j) X_j Y_j.
\end{split}
\end{equation} 
Here, we introduce the short-hand notation $\mathbf{U}_{n_{[k]},2}(\Sigmakx^{-1}; 1)$ and $\mathbf{U}_{n_{[k]},2}(\Sigmakx^{-1}; 0)$ for the ease of exposition. The oracle bias-corrected adjusted estimator of $\tau$ then reads as:
\[\hat\tau_{\oracle} = \sum\limits_{k=1}^K p_{n[k]}\Big[\Big\{\bar Y_{[k]1} - \mathbf{U}_{n_{[k]},2}(\Sigmakx;1) \Big\}-\Big\{\bar Y_{[k]0} - \mathbf{U}_{n_{[k]},2}(\Sigmakx;0) \Big\}\Big].\]
Viewing $\hat{\tau}_{\rm OLS}$ as a $V$-statistic, by removing the diagonal components using a $U$-statistic, the oracle adjusted estimator reduces the bias due to large $p$.

In reality, the population-level Gram matrix $\Sigmakx$ is in general unknown and needs to be estimated from the data. A natural approach is to use the sample Gram matrix estimator $\HSigmakx$ and plug it in. Replacing $\Sigmakx$ with $\HSigmakx$, we finally obtain the feasible adjusted estimator $\hat\tau$ defined as: 
\begin{equation}
\label{feasible}
\hat\tau = \sum\limits_{k=1}^K p_{n[k]}\Big[\Big\{\bar Y_{[k]1} - \mathbf{U}_{n_{[k]},2}(\HSigmakx;1) \Big\}-\Big\{\bar Y_{[k]0} - \mathbf{U}_{n_{[k]},2}(\HSigmakx;0) \Big\}\Big].
\end{equation}
The new adjusted estimator $\hat{\tau}$ constitutes one of the main methodological contributions of our paper. As will be demonstrated in Section~\ref{sec:theory}, the bias of $\hat{\tau}$ is negligible compared to its standard deviation, and it is guaranteed that the asymptotic variance of $\hat{\tau}$ never exceeds that of $\hat{\tau}_{\unadj}$ under the assumption-lean setting, as long as $p = o (n)$. We refer readers to Remark~\ref{rem:pseudoinverse} for further discussions on the bias of $\hat{\tau}$ when $p > n$.

\begin{remark}
\label{rem:linear models}
If treatments are independently assigned within each stratum, the oracle adjusted estimator is exactly unbiased, i.e., $E\{\hat\tau_{\oracle} - \tau\}=0$. In contrast, in CAR, treatment assignments $A_i$ and $A_j$ within the same block may be dependent, but, as we will show later in Section~\ref{sec:oracle}, $\hat\tau_{\oracle}$ is still asymptotically unbiased even after being scaled by $\sqrt{n}$:
\[\lim_{n \rightarrow \infty} E\{\sqrt{n} (\hat\tau_{\oracle}-\tau)\}\to 0,\]
It is not difficult to completely remove the remaining bias of order $o (n^{-1 / 2})$, which, however, is a less important issue and will not be further examined in this paper.
\end{remark}

\begin{remark}
\label{rem:leave-one-out}
It is noteworthy that our new estimator can also be interpreted as a modified OLS estimator simply by replacing the regression coefficients $\hat{\beta}_{[k]} (a)$ by their leave-one-out (LOO) forms. To see this, define the oracle LOO regression coefficients as: for any $i = 1, \cdots, n$, 
\begin{align*}
\tilde{\beta}^{-i}_{[k]}(1) = \Sigmakx^{-1} \cdot \frac{1}{n_{[k]}-1} \sum_{j\neq i} \frac{A_j}{\pi_{n[k]}} X_j Y_j, \, \, \, \, \text{and} \, \, \, \, \tilde{\beta}^{-i}_{[k]}(0) = \Sigmakx^{-1} \cdot \frac{1}{n_{[k]}-1} \sum_{j\neq i} \frac{1-A_j}{1-\pi_{n[k]}} X_j Y_j.
\end{align*}
We then obtain that, for each $a \in \{0, 1\}$,
\[\mathbf{U}_{n_{[k]},2}(\Sigmakx^{-1};a) = \frac{1}{n_{[k]a}} \sum\limits_{i\in [k]}(\mathbbm{1} \{A_i=a\} - \pi_{n[k]})X_i^{\top} \tilde{\beta}^{-i}_{[k]}(a).\]
We in turn define the oracle adjusted estimator $\hat{\tau}_{\rm ora}$ based on second-order $U$-statistics: 
\begin{equation*}
\hat\tau_{\oracle} =  \ \sum\limits_{k=1}^K p_{n[k]}\Big[\Big\{\bar Y_{[k]1} - \frac{1}{n_{[k]1}}\sum\limits_{i\in[k]}(A_i-\pi_{n[k]})X_i^{\top} \tilde{\beta}^{-i}_{[k]}(1) \Big\}-\Big\{\bar Y_{[k]0} - \frac{1}{n_{[k]0}}\sum\limits_{i\in[k]}(\pi_{n[k]}-A_i)X_i^{\top} \tilde{\beta}^{-i}_{[k]}(0) \Big\}\Big].
\end{equation*}
Compared to $\hat\tau_{\rm OLS}$, this LOO formulation of $\hat{\tau}_{\oracle}$ replaces the OLS regression coefficients $\hat{\beta}_{[k]}(a)$ by the oracle LOO coefficients $\tilde{\beta}^{-i}_{[k]}(a)$ for each $a\in\{0,1\}$. 
\end{remark}

\section{Statistical Properties of the Bias-Corrected Estimator}
\label{sec:theory}

In this section, for ease of our exposition, we first present the statistical properties of the oracle estimator $\hat\tau_{\oracle}$, before moving on to the properties of the feasible estimator $\hat\tau$, which are the main theoretical results of this paper. The oracle estimator $\hat\tau_{\oracle}$ serves as an ideal technical device that bridges the feasible estimator $\hat{\tau}$ and the true ATE $\tau$. Unlike $\hat{\tau}$, the $U$-statistic kernel of $\hat\tau_{\oracle}$ does not depend on the entire sample through $\HSigmakx$, so it is relatively straightforward to establish its bias, variance, and asymptotic normality. Theoretical results of $\hat{\tau}_{\oracle}$ then conceptually simplify the analysis of the feasible estimator $\hat{\tau}$: Once the results are in place in the oracle setting, deriving the properties of $\hat{\tau}$ reduces to controlling the difference between $\HSigmakx$ and $\Sigmakx$. However, it should be noted that this final step turns out to be technically challenging and involves a delicate decoupling leave-out analysis of $\HSigmakx^{-1}$.

\subsection{The oracle estimator}
\label{sec:oracle}

Before stating our theoretical results, we further introduce some mild regularity conditions required on the distributions of $X_i$ and $Y_i (a)$.

\begin{assumption}[Distributions of $X$ and $Y (a)$]\label{ap2}
There exists an absolute constant $M > 0$ such that the following hold:
\begin{enumerate}[label = (\arabic*)]
\item $X_i$ is uniformly bounded by $M$ in the sense that there exists constant $M$ such that $\max_{i,j} |X_{ij}| \leq M < \infty$, where $X_{ij}$ is the $j$-th covariate of the $i$-th unit.

\item Given $X_i$, the moments of $Y_i(a)$ up to order 4 are all bounded by $M$ almost surely: 
$\max_{i} \Ek [|Y_i(a)|^4|X_i] \leq M < \infty$ almost surely for $a \in\{0,1\},\ k\in\mathcal{K}$.
\end{enumerate}
\end{assumption}

\begin{assumption}[Eigenvalues of $\Sigmakx$]\label{ap3}
Let $\Lambda_{\min}$ and $\Lambda_{\max}$ be, respectively, the minimum and maximum eigenvalues of the stratum-specific population Gram matrix $\Sigmakx$. There exist two absolute constants $\kappa_l$ and $\kappa_{u}$ independent of $n$ such that $0 < \kappa_l \leq \Lambda_{\min} \leq \Lambda_{\max} \leq \kappa_u < \infty$.
\end{assumption}

Assumption~\ref{ap2} imposes tail conditions on the covariates $X$ and the potential outcome $Y (a)$, which are mild in the context of RCT analysis. Specifically, it is standard to rescale $X$ in practice so that $X$ can be viewed as bounded. The bounded fourth-moment condition on $Y (a)$ is imposed in \citet{Freedman2008} and \cite{lin2013}. We conjecture that it is possible to further relax Assumption~\ref{ap2}(1) from bounded $X$ to sub-Gaussian tails, but proving this is beyond the scope of this paper. Assumption~\ref{ap3} incurs almost no loss of generality since we only assume that the population Gram matrix has bounded eigenvalues from above and below. These two assumptions are needed to upper bound the variance and prove the asymptotic normality of $\hat{\tau}_{\oracle}$ or $\hat{\tau}$.

Before presenting the main theoretical result of this section, we first introduce the following generic notation that will be frequently encountered in the rest of our paper: for $a, b \in \{0, 1\}$ and $l, m \in \{1, 2\}$, 
\begin{equation}
\label{notation_II}
\zetaSec{r[l,m]}(a,b,w(\pi_{[k]})) = \sum\limits_{k=1}^K w(\pi_{[k]})\frac{p_{[k]}}{n_{[k]}-1}\Ek\{(X_1^{\top}\Sigmakx^{-1}X_2)^2 r_l(a)r_m(b)\},
\end{equation}
where $w: [0, 1] \to \mathbb{R}$ is a function of the stratum-specific treatment-assignment probability $\pi_{[k]}$. Here, the new notation $\zetaSec{r[l,m]}(a,b,w(\pi_{[k]}))$ resembles the notation $\zeta^2_{\mathrm{I},r}(\pi_{[k]})$ introduced in \eqref{notation_I} previously. Armed with \eqref{notation_II}, we can introduce the following quantities that are useful in representing the variance of $\hat{\tau}_{\rm ora}$. To this end, define 
\begin{equation*}
\sigma^2 = \zeta^2_{H} + \zeta^2_{\mathrm{I},r}(\pi_{[k]}) + \zeta^2_{\mathrm{II}},
\end{equation*}
where $\zeta^2_{\mathrm{II}} = \zetaSec{Y(1)} + \zetaSec{Y(0)} - 2 \zetaSec{Y[1,2]}(1,0,1)$, and 
\begin{align*}
    \zetaSec{Y(1)} &= \zetaSec{Y[1,1]}\Big(1,1,\frac{1-\pi_{[k]}}{\pi^2_{[k]}}\Big)+ \zetaSec{Y[1,2]}\Big(1,1,\frac{(1-\pi_{[k]})^2}{\pi^2_{[k]}}\Big),\\
    \zetaSec{Y(0)} &= \zetaSec{Y[1,1]}\Big(0,0,\frac{\pi_{[k]}}{(1-\pi_{[k]})^2}\Big) + \zetaSec{Y[1,2]}\Big(0,0,\frac{\pi^2_{[k]}}{(1-\pi_{[k]})^2}\Big).
\end{align*}

We are now ready to present the statistical properties of $\hat\tau_{\oracle}$ in Theorem~\ref{thm1} below.

\begin{theorem}\label{thm1}
Under Assumptions \ref{ap1}--\ref{ap3}, when $p \lesssim n$, the following hold.
\begin{enumerate}[label = (\alph*)]
\item The bias and variance of the oracle estimator $\hat\tau_{\oracle}$ satisfy:
\begin{align*}
& E\{\sqrt n (\hat\tau_{\oracle} - \tau)\} = o(1), \text{ and} \\
& \var\{\sqrt n (\hat\tau_{\oracle} - \tau)\} = \sigma^2 + o (1).
\end{align*}

\item Furthermore, $\hat\tau_{\oracle}$ is $\sqrt{n}$-CAN, or more precisely:
\[\sqrt{n}(\hat\tau_{\oracle} - \tau)/\sigma\stackrel{d}\to \operatorname{N}(0,1).\]
\end{enumerate}
\end{theorem}

Since $\hat{\tau}_{\oracle}$ is a $U$-statistic, the proof of Theorem~\ref{thm1} generalizes the coupling technique of \citet{Bugni2019} to $U$-statistics, which could be of independent interest; see Appendix~\ref{app:thm1}. In Theorem~\ref{thm1}, we only need $p \lesssim n$ for the statements to hold because $\Sigmakx$ is known and we do not need $p < n$ to ensure $\Sigmakx$ to be invertible (guaranteed by Assumption~\ref{ap2}). In fact, the asymptotic normality of $\hat{\tau}_{\oracle}$ is maintained as long as $p = o (n^{2})$, if we scale $\hat{\tau}_{\oracle} - \tau$ by $\max \{\sqrt{n}, \sqrt{n^{2} / p}\}$ instead of $\sqrt{n}$ \citep{liu2020nearly}.

It is well known in the literature that $\sigma_{\rm OLS}^{2}$ is guaranteed to be smaller than or equal to the asymptotic variance of $\sqrt{n} \hat{\tau}_{\unadj}$. In the proposition below, we also establish the relationship between $\sigma_{\rm OLS}^{2}$ and $\sigma^{2}$. In particular, for $\sigma^{2}$ to be of the same order as $\sigma_{\rm OLS}^{2}$, $p = o (n)$ is still needed. Therefore, if $p = o (n)$, $\hat{\tau}_{\oracle}$ is also guaranteed to be never less efficient than $\hat{\tau}_{\unadj}$.
\begin{proposition}
\label{prop:var}
The following hold.
\begin{enumerate}[label = (\alph*)]
\item $\zeta^{2}_{\rm II} \geq 0$, or equivalently $\sigma_{\rm OLS}^{2} \leq \sigma^{2}$;

\item Under the assumptions of Theorem~\ref{thm1}, if $p = o (n)$, $\lim_{n \rightarrow \infty} \sigma^{2} / \sigma^{2}_{\rm OLS} = 1$.
\end{enumerate}
\end{proposition}

\begin{remark}\label{rem: zeta2}
\citet{ma2022regression} have shown that when $p$ is fixed, $\hat{\tau}_{\mathrm{OLS}}$ has a guaranteed efficiency gain compared to $\hat{\tau}_{\unadj}$. When $p = o(n)$, our oracle estimator $\hat{\tau}_{\oracle}$ achieves the same asymptotic variance as $\hat{\tau}_{\mathrm{OLS}}$, and is thus also more efficient than or as efficient as $\hat{\tau}_{\unadj}$.  The extra term $\zeta^2_{\mathrm{II}}$ in $\var \{\hat{\tau}_{\oracle}\}$ compared to $\var \{\hat{\tau}_{\rm OLS}\}$ is $O (p/n)$, and it can be shown that $\zeta^2_{\mathrm{II}}\geq0$ (see Appendix~\ref{app:var estimator}). As a result, when $p/n \to \gamma\in(0,1]$, the asymptotic variance of $\hat\tau_{\oracle}$ may exceed that of $\hat\tau_{\unadj}$ in the proportional asymptotic regime, potentially reducing efficiency unless further conditions are assumed. See Section~\ref{sec:conclusion} for some further discussions.
\end{remark}

\begin{remark}
\label{rem:comparison}
Recent work has also focused on covariate adjustment in randomized experiments in which $p$ diverges with $n$. Most closely related to our work, \citet{zhao2024hoif} employed second-order $U$-statistics to estimate linear adjustments under the finite-population model with simple randomization. Separately, \citet{chang2024exact} and \citet{lu2025debiased} derived a debiased regression-adjusted estimator within the same framework; \citet{zhao2024hoif} noted that these estimators are asymptotically equivalent and had been introduced in almost identical forms in observational studies \citep{liu2020nearly, liu2020rejoinder, liu2023new} using the framework of higher-order influence functions \citep{liu2017semiparametric}, leading to the same asymptotic normality. We focus on CAR under a superpopulation model, which differs from the finite-population model in both theoretical assumptions and randomization schemes. While \citet{jiang2025adjustments} also consider the regime in which $\sqrt{n} \lesssim p \lesssim n$ in the superpopulation model, they assume a correct linear relationship between potential outcomes and covariates. Under this predicate, the OLS estimator $\hat{\tau}_{\rm OLS}$, although a $V$-statistic instead of a $U$-statistic, is still $\sqrt{n}$-CAN. In contrast, Theorem \ref{thm1} allows for possible misspecification of the linear model. As we have illustrated in Section~\ref{sec:intuition}, the regression-adjusted estimator based on $V$-statistics may have bias diverging to infinity after being standardized by $\sqrt{n}$ under model misspecification.

\end{remark}

\subsection{The feasible estimator}
\label{sec:feasible}

Theorem \ref{thm1} establishes the asymptotic normality of the oracle estimator $\hat\tau_{\oracle}$, but this theoretical result is derived under the often unrealistic assumption that the stratum-specific population Gram matrix $\Sigmakx$ is known. In practice, $\Sigmakx$ shall be estimated from the observed data. As indicated at the end of Section~\ref{sec:intuition}, the feasible estimator $\hat{\tau}$ simply replaces $\Sigmakx$ by the corresponding sample covariance estimator $\HSigmakx$.

Recall from \eqref{feasible} in Section~\ref{sec:intuition} that
\[\hat\tau = \sum\limits_{k=1}^K p_{n[k]}\Big[\Big\{\bar Y_{[k]1} - \mathbf{U}_{n_{[k]},2}(\HSigmakx;1) \Big\}-\Big\{\bar Y_{[k]0} - \mathbf{U}_{n_{[k]},2}(\HSigmakx;0) \Big\}\Big].\]
We first present the most important theoretical result of this paper, concerning the asymptotic statistical properties of $\hat{\tau}$.
\begin{theorem}
\label{thm:main}
Under Assumptions \ref{ap1}--\ref{ap3}, the following hold. 
\begin{itemize}
\item[(a)] When $p<n$, the feasible estimator $\hat\tau$ is asymptotic unbiased in the sense that
\begin{align*}
E\{\sqrt n (\hat\tau - \tau)\} = o(1).
\end{align*}
\end{itemize}
\vspace{-1.5em}
We then also assume $p = o (n)$.
\begin{itemize}
\item[(b)] The asymptotic variance of $\hat\tau$ satisfies
    \[\var\{\sqrt{n} (\hat\tau - \tau)\} = \sigma^2 + o (1).\]

\item[(c)] Moreover, $\hat\tau$ is $\sqrt n$-CAN in the sense that 
    \[\sqrt{n}(\hat\tau - \tau)/\sigma \stackrel{d}\to \operatorname{N}(0,1).\]
In view of Proposition \ref{prop:var}, when $p = o (n)$, $\lim_{n \rightarrow \infty} \sigma^{2} / \sigma^{2}_{\rm OLS} = 1$, we also have 
\[\sqrt{n}(\hat\tau - \tau)/\sigma_{\rm OLS} \stackrel{d}\to \operatorname{N}(0,1).\]
\end{itemize}
\end{theorem}

In fact, to maintain the same statistical properties of $\hat{\tau}$ as in Theorem~\ref{thm:main}, in principle $\HSigmakx^{-1}$ can be replaced by any generic estimator $\HSigmak{}^{-1}$ of $\Sigmakx^{-1}$, as long as $\HSigmak{}$ satisfies the following: 
\begin{equation}
\label{high level}
\sqrt{n}\Big\{\mathbf{U}_{n_{[k]},2}(\HSigmak{}^{-1} - \Sigmakx^{-1};1) - \mathbf{U}_{n_{[k]},2}(\HSigmak{}^{-1} - \Sigmakx^{-1};0)\Big\} = o_P(1),
\end{equation}
for all $k \in \mathcal{K}$. Here, for $a = 0, 1$, the notation $\mathbf{U}_{n_{[k]},2}(\HSigmak{}^{-1} - \Sigmakx^{-1};a)$ is the same as $\mathbf{U}_{n_{[k]},2}(\Sigmakx^{-1};a)$ defined in \eqref{U-notation}, replacing $\Sigmakx^{-1}$ by $\HSigmak{}^{-1} - \Sigmakx^{-1}$. Condition~\eqref{high level}  characterizes the relevant functional convergence rate of $\HSigmak{}^{-1}$ to $\Sigmakx^{-1}$ such that the difference between $\hat\tau_{\oracle}$ and $\hat\tau$ is $o_P (n^{-1/2})$. In the proof, we show that the sample Gram matrix estimator $\HSigmakx$ satisfies Condition~\eqref{high level} when $p = o(n)$. However, it is possible to consider other estimators of $\Sigma_{[k]}^{-1}$, such as the ridge inverse Gram matrix estimator \citep{liu2025augmented, abadie2025unbiased}. We leave the justification whether other estimators of $\Sigmakx^{-1}$ also satisfy Condition~\eqref{high level} to future work.

In the proof of Theorem~\ref{thm:main} (see Appendix~\ref{app:thm:main}), the use of the sample Gram matrix $\HSigmakx$ introduces two technical challenges. First, since $\HSigmakx$ involves all the data in stratum $k$, it induces complicated dependencies into the corresponding $U$-statistic kernel. A standard kernel $h$ of a second-order $U$-statistic is data-independent, while the corresponding kernel of $\hat{\tau}$ depends on all $X_{i}$'s in stratum $k$ through $\HSigmakx$. To decouple such a complicated dependency structure, we obtain a technical lemma (Lemma~\ref{lem:repeated Sherman-Morrison}) that repeatedly applies the Sherman–Morrison formula to carry out a delicate leave-two-out analysis of $\HSigmakx^{-1}$. This result could be of independent interest in other problems dealing with inverse sample Gram matrices with large $p$ \citep{liu2017semiparametric, bao2025leave}.

The second challenge lies in characterizing the asymptotic variance of $\hat{\tau}$. To show that $\hat{\tau}$ has the same asymptotic variance $\sigma^{2}$ as $\hat\tau_{\oracle}$, a sufficient condition is $\norm{\HSigmakx^{-1} - \Sigmakx^{-1}}{2} = o_{P} (1)$, which holds automatically if $p = o (n)$ by standard matrix concentration bounds. When $p = O (n)$, the bias of $\hat{\tau}$ is still $o (n^{-1 / 2})$ but the asymptotic variance of $\hat{\tau}$ is generally not $\sigma^{2}$. In Section~\ref{sec:conclusion}, we will discuss the asymptotic variance of $\hat{\tau}$ when $p \asymp n$.

Finally, a couple of remarks are in order before closing this section.

\begin{remark}
\label{rem:pseudoinverse}
In fact, the bias of $\hat{\tau}$ is always $o (n^{-1 / 2})$, regardless of how $p$ scales with $n$: when $\hat{\Sigma}$ is not invertible, plugging-in the pseudo-inverse of $\hat{\Sigma}_{[k]}$ will not affect the order of the bias. This is a consequence of knowing the probability of treatment assignment by design. For the asymptotic variance of $\hat{\tau}$ to be no greater than that of $\hat{\tau}_{\unadj}$, a sufficient condition is to ensure that $\hat{\tau}$ has the same asymptotic variance as $\sigma^{2}$, which still needs $p = o (n)$.
\end{remark}

\begin{remark}
We use the original $X_i$ rather than the centered $\check X_i = X_i - \bar{X}_{[k]}$ to construct the second-order $U$-statistics. In the superpopulation model under random designs, centered covariates $\check X_i$ inject more complex dependencies that compound the analysis of $U$-statistics. However, we expect that the statistical properties of the resulting $\hat{\tau}$ continue to hold. A similar observation should hold when using the centered outcomes $\check Y_i = \mathbbm{1}\{A_i=a\}Y_i - \bar{Y}_{[k]a}$. Centering $Y$ is expected to change only the form of the asymptotic variance $\sigma^{2}$ by replacing $Y(a)$ with $Y(a) - \Ek[Y(a)]$, for $a=0,1$. 
\end{remark}

\section{Valid Inference with the Bias-Corrected Estimator}
\label{sec:inference}

In this section, we demonstrate how to estimate the variance of our new bias-corrected adjusted estimator $\hat\tau$, so that it can be used in subsequent tasks related to hypothesis testing or statistical inference. Recall from Section~\ref{sec:oracle} and Theorem~\ref{thm:main} that the asymptotic variance $\sigma^2$ can be decomposed into the following terms: $\sigma^2 = \zeta^2_{H} + \zeta^2_{\mathrm{I},r}(\pi_{[k]}) + \zeta^2_{\mathrm{II}}$. The general strategy that we follow is to estimate each component in $\sigma^{2}$ by a sample average or a $U$-statistic.

We now explain how to estimate each of the components of $\hat{\sigma}^{2}$. First, we estimate $\zeta^2_{H}$ by
\[\hat\zeta^2_{H} = \sum\limits_{k=1}^K p_{n[k]}\Big\{\big(\bar Y_{[k]1} - \sum\limits_{k'=1}^K p_{n[k']}\bar Y_{[k']1}\big) - \big(\bar Y_{[k]0} - \sum\limits_{k'=1}^K p_{n[k']}\bar Y_{[k']0}\big)\Big\}^2,\]
because $\zeta_{H}^{2} = E \{E_{[k]} \{r_{i} (1) - r_{i} (0)\} - E \{r_{i} (1) - r_{i} (0)\}\}^{2}$ as stated in Proposition~\ref{prop1}, and we simply estimate each population mean by a corresponding sample average.

For the second term $\zeta^2_{\mathrm{I},r}(\pi_{[k]})$, we first recall that 
\[\zeta^2_{\mathrm{I},r}(\pi_{[k]}) = \sum\limits_{k=1}^K p_{[k]}\Big(\frac{\sigma^2_{[k]r(1)}}{\pi_{[k]}} + \frac{\sigma^2_{[k]r(0)}}{1-\pi_{[k]}}\Big),\]
where $r_i(a) = Y_i(a) - X_i^{\top}\beta_{[k]}$, $\beta_{[k]} = (1-\pi_{n[k]})\beta_{[k]}(1) + \pi_{n[k]}\beta_{[k]}(0)$, $\ProjBeta{a} = \Sigmak{}^{-1} \eta_{[k]}(a)$, and $\eta_{[k]} (a) = \Ek\{XY(a)\}$. In Appendix~\ref{app:var estimator}, we show that $\zeta^2_{\mathrm{I}, r} (\pi_{[k]})$ can be equivalently represented as
\[\zeta^2_{\mathrm{I},r}(\pi_{[k]}) = \sigma^2_Y(\pi_{[k]}) - \sigma^2_{\mathrm{I},\eta}(\pi_{[k]}) - 2 \sigma_{\mathrm{I},\eta(1)\eta(0)},\]
where 
\begin{align*}
& \sigma^2_{Y}(\pi_{[k]}) = \sum\limits_{k=1}^K p_{[k]}\Big\{\frac{1}{\pi_{[k]}} \sigma^2_{[k]Y(1)} + \frac{1}{1-\pi_{[k]}} \sigma^2_{[k]Y(0)}\Big\}, \\
& \sigma^2_{\mathrm{I},\eta}(\pi_{[k]}) = \sum\limits_{k=1}^K p_{[k]}\Big\{\frac{1-\pi_{[k]}}{\pi_{[k]}} \eta^{\top}_{[k]}(1)\Sigmakx^{-1}\eta_{[k]}(1) + \frac{\pi_{[k]}}{1-\pi_{[k]}} \eta^{\top}_{[k]}(0)\Sigmakx^{-1}\eta_{[k]}(0)\Big\}, \\
& \sigma_{\mathrm{I},\eta(1)\eta(0)} = \sum\limits_{k=1}^K p_{[k]}\Big\{\eta^{\top}_{[k]}(1)\Sigmakx^{-1}\eta_{[k]}(0)\Big\}.
\end{align*}
As $\zeta_{H}^{2}$, $\sigma^2_Y(\pi_{[k]})$ can be estimated by the corresponding sample variance. If we directly plug in the OLS coefficients $\hat\beta_{[k]}(a),\ a\in\{0,1\}$, then the estimators of $\sigma^2_{\mathrm{I},\eta}(\pi_{[k]})$ and $\sigma_{\mathrm{I},\eta(1)\eta(0)}$ are $V$-statistics, which leads to non-negligible bias for the same reason as for $\hat{\tau}_{\rm OLS}$. Therefore, we can estimate $\zeta^2_{\mathrm{I},r}(\pi_{[k]})$ using $U$-statistics instead, following the same type of construction as our point estimate $\hat{\tau}$. We denote the resulting estimator of $\zeta^2_{\mathrm{I},r}(\pi_{[k]})$ by $\hat\zeta^2_{\mathrm{I},r}(\pi_{[k]})$, which is
\begin{align*}
\hat\zeta^2_{\mathrm{I},r}(\pi_{n[k]}) =\hat\sigma^2_Y(\pi_{n[k]}) - \hat\sigma^2_{\mathrm{I},\eta}(\pi_{n[k]}) - 2 \hat\sigma_{\mathrm{I},\eta(1)\eta(0)}.
\end{align*}
The explicit forms of $\hat\sigma^2_Y(\pi_{n[k]})$, $\hat\sigma^2_{\mathrm{I},\eta}(\pi_{n[k]})$, and $ \hat\sigma_{\mathrm{I},\eta(1)\eta(0)}$ are given in Appendix~\ref{app:var estimator}.
    
Finally, in terms of $\zeta^2_{\mathrm{II}} = \zeta^2_{\mathrm{II}, Y (1)} + \zeta^2_{\mathrm{II}, Y (0)} - 2 \zeta^2_{\mathrm{II}, Y [1, 2]} (1, 0, 1)$, we apply the same estimation strategy as for $\zeta^2_{\mathrm{I},r}(\pi_{[k]})$. Specifically, we estimate $\sigma^2_{\mathrm{II},Y(a)[1]}$ by the sample mean, while $\sigma^2_{\mathrm{II},Y(a)[1,2]}$ and $\sigma^2_{\mathrm{II},Y(1,0)}$ are estimated by $U$-statistics. Their explicit forms are also deferred to Appendix~\ref{app:var estimator}. The estimator of $\zeta^2_{\mathrm{II}}$ is simply denoted as $\hat\zeta^2_{\mathrm{II}}$. In summary, our proposed variance estimator $\hat{\sigma}^{2}$ of $\sigma^{2}$ reads as follows: $\hat\sigma^2 = \hat\zeta^2_{H} + \hat\zeta^2_{\mathrm{I},r}(\pi_{n[k]}) + \hat\zeta^2_{\mathrm{II}}$.

Our last main theoretical result shows that $\hat{\sigma}^{2}$ is a consistent estimator of $\sigma^{2}$ under mild assumptions, so we can use $\hat{\tau}$ and $\hat{\sigma}^{2}$ to conduct downstream inference.

\begin{theorem}
\label{thm:var}
Under Assumptions \ref{ap1}--\ref{ap3}, and assuming that $p = o(n)$, the following hold:
\begin{enumerate}[label = (\alph*)]
\item $\hat\sigma^2$ is a consistent estimator of $\sigma^2$: $\hat\sigma^2\stackrel{P}\to \sigma^2$.

\item As a consequence, after scaled by $\hat{\sigma}$ instead of $\sigma$, $\hat{\tau}$ is still $\sqrt{n}$-CAN: 
\begin{align*}
\sqrt{n} (\hat\tau - \tau)/\hat\sigma \stackrel{d}\to \operatorname{N}(0,1).
\end{align*}

\item Finally, the nominal $(1 - \alpha) \times 100\%$ large-sample Wald confidence interval $$\hat{\rm CI}_{\alpha} = \left[ \hat{\tau} - z_{1 - \alpha / 2} \dfrac{\hat{\sigma}}{\sqrt{n}}, \hat{\tau} + z_{1 - \alpha / 2} \dfrac{\hat{\sigma}}{\sqrt{n}} \right]$$ retains the correct coverage probability as $n \rightarrow \infty$.
\end{enumerate}
\end{theorem}

\begin{remark}\label{rem: high level2}
Similar to Condition~\eqref{high level}, $\hat\sigma^2$ is still a consistent estimator of $\sigma^2$ if we replace $\HSigmakx$ by any estimator $\HSigmak{}$ of $\Sigmakx$ if: (a) $\HSigmak{}$ is invertible almost surely and (b)
\begin{equation}\label{high level2}
\Ek \{X_i^{\top} (\HSigmak{}^{-1} - \Sigmakx^{-1}) X_i\} = o(p),\ \Ek\{X_i^{\top} (\HSigmak{}^{-1} - \Sigmakx^{-1})X_i\}^2 = O\left(\frac{p^3}{n}\right).
\end{equation}
We verify that $\HSigmakx$ satisfies Condition~\eqref{high level2} when $p=o(n)$ in Lemma~\ref{lem:variance} (see Appendix~\ref{app:matrices}).
\end{remark}

\begin{remark}
\label{rem:var estimation proportional}
When $p \asymp n$, $\hat{\sigma}^2$ is not a consistent estimator of $\sigma^2$ and their difference is of order $p / n$. Based on the empirical results in Section~\ref{sec:experiments}, $\hat{\sigma}^{2}$ may slightly overestimate $\sigma^{2}$, which makes sure that the resulting Wald CI is conservative. But we can still observe that $\hat{\sigma}^{2}$ is generally less than the asymptotic variance of $\hat{\tau}_{\unadj}$, achieving an efficiency gain.
\end{remark}

\section{Numerical Experiments}
\label{sec:experiments}

\subsection{Simulation studies}
\label{sec:sim}

In this section, we evaluate the performance of our proposed estimators, including both the oracle estimator $\hat\tau_{\oracle}$ and the feasible estimator $\hat{\tau}$, and compare them with the unadjusted estimator $\hat\tau_{\unadj}$ and the OLS estimator $\hat\tau_{\rm OLS}$. We use $\hat\sigma^2$ proposed in Theorem \ref{thm:var} to estimate the variance of $\hat\tau$. For the oracle estimator $\hat\tau_{\oracle}$, the corresponding variance estimator $\hat\sigma_{\oracle}^{2}$ is the same as $\hat{\sigma}^{2}$ except that we use $\Sigmakx$ instead of $\HSigmakx$. For $\hat\tau_{\unadj}$ and $\hat\tau_{\rm OLS}$, we use the variance estimators $\hat\sigma^2_{\rm unadj}$ and $\hat\sigma^2_{\rm OLS}$ proposed in \citet{Bugni2018} and \citet{gu2023regression}, respectively. The simulation is carried out under stratified block randomization, using a categorical covariate $X_{1} \in \{1, 2, 3, 4\}$ as the stratification variable, with corresponding probabilities $\{0.2, 0.2, 0.3, 0.3\}$. Specifically, we consider the following data-generating processes.

\paragraph{Model 1.}
In Model 1, the outcome is continuous, where for $i = 1, \cdots, n$, 
\begin{align*}
    Y_i(0) & = X_{1i}+ 2X_{0i}^{\top} \beta_0 - 0.5X_{0i,4}^2 + \epsilon_{0,i}, \text{ and} \\
    Y_i(1)&= X_{1i}+0.05 X_{0i}^{\top} \Sigma_0 ^{-1} X_{0i} + \epsilon_{1,i}.
\end{align*}
Here, $X_{0i}$ is a random vector of dimension $p_{0}$ and $X_{0i}\sim t_5(0,\Sigma)$, with $p_0=30$ and $\Sigma_{i,j} = 0.1^{|i-j|}$. Each component of $\beta_0$ is generated from the uniform distribution on $[-1,1]$ and then rescaled to $\norm{\beta_0}{2}=1$. Finally, $\epsilon_{0,i}, \epsilon_{1,i} \sim \operatorname{N}(0,0.1^2)$.

\paragraph{Model 2.}
In Model 2, the outcome is binary, where for $i = 1, \cdots, n$,
\begin{align*}
    \Pr(Y_i(0)=1) & = \mathrm{expit}\Big(-1+ X_0^{\top}\beta_0 -2X_{0i,1}^2 \Big), \text{ and} \\
    \Pr(Y_i(1)=1)&= \mathrm{expit}\Big(-3 +X_0^{\top} \beta_1+2X_{0i,2}^2 + 0.5X_{0i,3}^4 \Big).
\end{align*}
Here $\mathrm{expit}(z) = 1/(1+\exp(-z))$. $X_0$ is generated in the same fashion as in Model 1. We take $\beta_{1,j} = 0.5,\ j=1\dots p_0$, $\beta_{0,j} = 1.5,\ j=1,2,\dots,p$.

We set the sample size to $n = 1000$. The allocation ratio between the treatment and control groups is $1:1$, leading to an expected control (treatment) group sample size of $n_{[k_{\min}]0}=100$ ($n_{[k_{\min}]1} = 100$) in the smallest stratum. We denote the covariates used in the analysis stage as $X_{2i}$ for $i = 1, \cdots, n$. When $p \leq p_0$, $X_{2i}$ consists of the first $p$ components of $X_{0i}$; when $p > p_0$, $X_{2i} = [X_{0i}, Z_i]$, where $Z_i$ is a random vector of dimension $(p - p_0)$ with components $Z_{i, q} \stackrel{i.i.d.}{\sim} t_5,\ q=1,2,\dots,p-p_0$. We set $p = \lceil rn\rceil$, where $r \in \{0.02,0.05,0.1,0.2,0.3,0.4,0.5,0.6,0.7\}$. We let $r_0 = p_0/n=0.3$ be the ratio between the dimension of $X_{0}$ and the sample size $n$. 

In the simulation, we draw $R = 2000$ Monte Carlo replicates. The performance of different estimators is evaluated in terms of absolute bias, Monte Carlo standard deviation (SD), the ratio of the Monte Carlo SD to the estimated standard error (SE) ${\rm sd} / {\rm se}$, and the coverage probability (CP) of the associated nominal 95\% Wald confidence interval (CI). We also report the Monte Carlo CP, which is computed using the Monte Carlo SD instead of the estimated SE to evaluate the impact of bias on inference.  

\begin{itemize}
    \item In Figures \ref{fig:mod1}(a) and \ref{fig:mod2}(a), we report the absolute bias of all estimators included in the comparison. The bias of $\hat\tau_{\rm OLS}$ increases as the ratio $r = p / n$ increases, until it reaches $r_0 = 0.3$. This is expected because the additional covariates $Z$ in $X_{2}$ are independent of $Y (a)$ and thus will not further increase the bias, as noted in Remark~\ref{rem:comparison}. All other estimators, $\hat{\tau}_{\unadj}$, $\hat{\tau}_{\oracle}$, and $\hat{\tau}$, have a bias close to zero when varying $r$. As shown in Figures \ref{fig:mod1}(e) and \ref{fig:mod2}(e), when evaluating the impact of bias on CP, only the Wald CI centered at $\hat\tau_{\rm OLS}$ undercovers at around 90\% when $r \ge 0.3$.
    
    \item In Figure \ref{fig:mod1}(b), the Monte Carlo SD decreases with increasing $r$ for all three adjusted estimators $\hat\tau_{\rm OLS}$, $\hat\tau_{\oracle}$, and $\hat\tau$. All three estimators are more efficient than $\hat\tau_{\unadj}$. The efficiency gain of $\hat\tau_{\rm OLS}$ relative to $\hat\tau$ also increases with $r$, consistent with the claim in Proposition~\ref{prop:var} that the additional term $\zeta^2_{\mathrm{II}}$ is of order $O(p/n)$ and $\zeta^2_{\mathrm{II}} \geq 0$. However, in Figure \ref{fig:mod2}(b), $\hat\tau_{\rm OLS}$ has a larger SD compared to $\hat\tau_{\rm unadj}$ when $r$ reaches 0.6, which could cause harm. In contrast, $\hat\tau$ retains efficiency gains compared to $\hat\tau_{\rm unadj}$.
    
    \item In Figures \ref{fig:mod1}(c) and \ref{fig:mod1}(d), we examine the performance of our variance estimators and their corresponding CP. The ratios $\rm sd/se$ for all estimators are within the range $[0.85,1.05]$. As $r$ increases, $\hat\sigma$ shows a larger decrease in $\rm sd/se$ compared to $\hat\sigma_{\oracle}$, which may be caused by the violation of Condition~\eqref{high level} by $\HSigmakx$ when $p = o (n)$ does not hold. In other words, $\hat{\sigma}$ can overestimate the actual SD when $p$ is near $n$, which is not as concerning as long as $\hat{\sigma}$ is still smaller than the SD of $\hat{\tau}_{\unadj}$. The Wald CIs associated with both $\hat\tau_{\oracle}$ and $\hat\tau_{\unadj}$ have CP around 95\%, while $\hat\tau$ is slightly conservative because its SE overestimates SD. For $\hat\tau_{\rm OLS}$, the larger bias and the underestimated SE together lead to the undercoverage of its Wald CI for almost all $r$.
    
    \item In Figures \ref{fig:mod2}(a)--(e), we examine the performance of our estimators when the outcome is binary. The results similarly show that our proposed estimator $\hat\tau$ still has negligible bias, delivers valid inference, and is more efficient than $\hat\tau_{\unadj}$, further corroborating the practical utility of our proposed estimator. 
\end{itemize}

\begin{figure}[ht]
\centering
\includegraphics[width=0.95\textwidth]{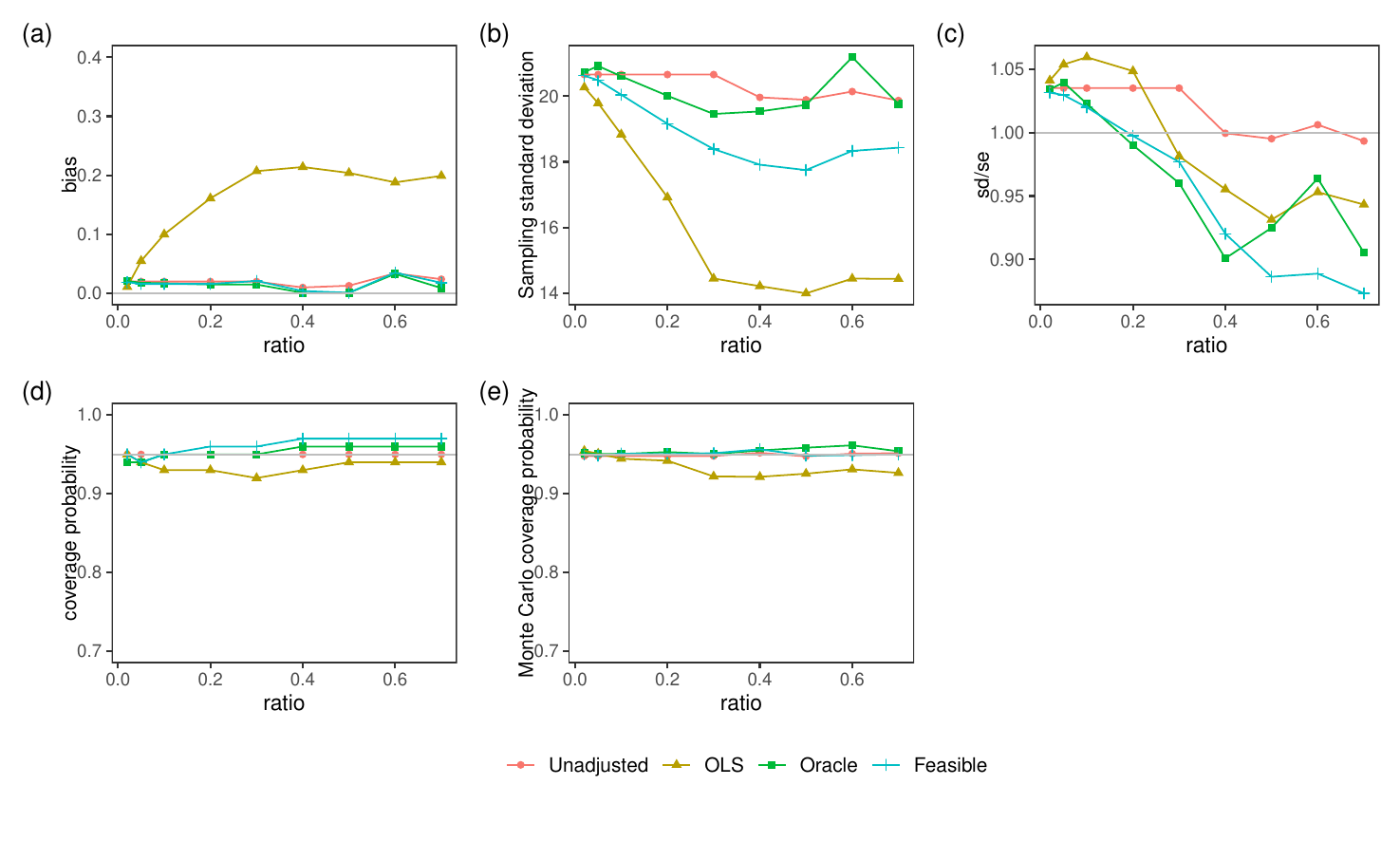}
\caption{Simulation results for different $p/n$ ratios of Model 1. Panels (a)–(e) summarize the finite-sample performance of the four estimators: $\hat\tau_{\rm unadj}$, $\hat\tau_{\rm OLS}$, $\hat\tau_{\oracle}$, and $\hat\tau$.
(a) Absolute bias; (b) Sampling standard deviation; (c) Ratio of empirical standard deviation to the estimated standard error; (d) Coverage probability; and (e) Monte Carlo coverage probability.}
\label{fig:mod1}
\end{figure}

\begin{figure}[ht]
\centering
\includegraphics[width=0.95\textwidth]{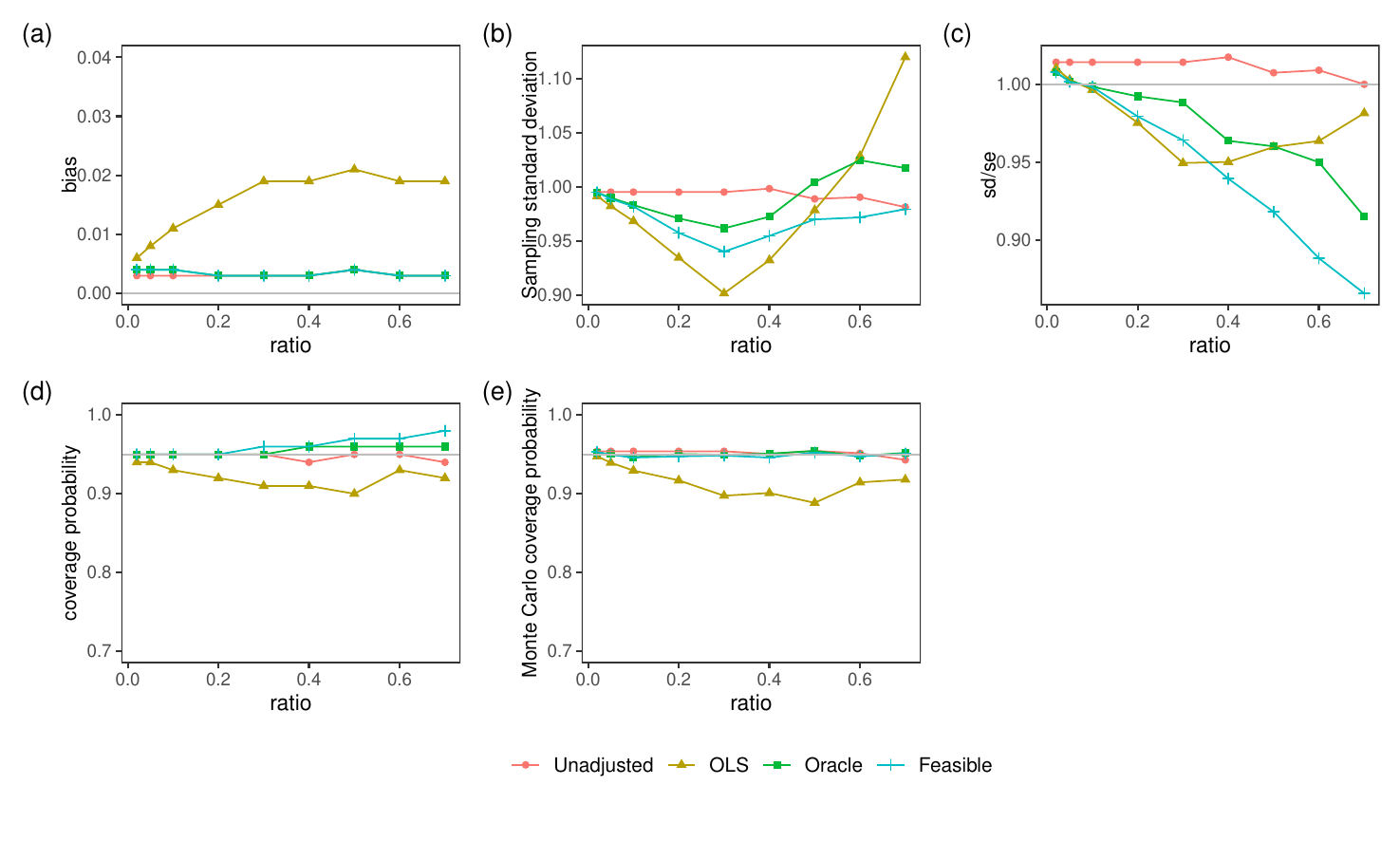}
\caption{Simulation results for different $p/n$ ratios of Model 2. Panels (a)–(e) summarize the finite-sample performance of the four estimators: $\hat\tau_{\rm unadj}$, $\hat\tau_{\rm OLS}$, $\hat\tau_{\oracle}$, and $\hat\tau$.
(a) Absolute bias; (b) Sampling standard deviation; (c) Ratio of empirical standard deviation to the estimated standard error; (d) Coverage probability; and (e) Monte Carlo coverage probability.}
\label{fig:mod2}
\end{figure}

\subsection{A semi-synthetic RCT data analysis}
\label{sec:data}

In this section, we analyze a semi-synthetic data simulated from an RCT of nefazodone and the cognitive behavioral-analysis system of psychotherapy (CBASP) \citep{keller2000comparison}. The purpose of this trial is to compare the effects of nefazodone, CBASP, and their combination on chronic depression. We take the combination as the treatment group indexed as $A = 1$ and nefazodone as the control group indexed as $A = 0$. The outcome of interest $Y$ is FinalHAMD, the final score of the 24-item Hamilton rating scale for depression. To generate the semi-synthetic data, we fit the real data with an additive model using the function \texttt{gam} from the \texttt{R} package \texttt{mgcv}. We stratify the data using GENDER and select seven covariates to fit the model, including AGE, HAMA, HAMA\_SOMATI, HAMD17, HAMD24, Mstatus2, and TreatPD. 
We include linear and quadratic transformations of continuous covariates and their interactions, resulting in $p = 30$ covariates in total. We model quadratic and interaction terms with a cubic spline, while the remaining five continuous covariates enter the model linearly. We then sample 600 units with replacement from the real data as our semi-synthetic data. Next, we implement stratified block randomization, using GENDER as the stratified variable with an allocation ratio of 1:1 and a block size of 6. We compare the performance of three estimators based on 2000 Monte Carlos: $\hat\tau_{\rm unadj}$, $\hat\tau_{\rm OLS}$ and our new estimator $\hat\tau$.

As shown in Table \ref{tab1}, the biases of $\hat\tau_{\unadj}$ and $\hat\tau$ are close to zero, while $\hat\tau_{\rm OLS}$ exhibits non-negligible bias. Meanwhile, $\hat\tau$ is about 30\% more efficient than $\hat\tau_{\unadj}$. The standard error (SE) of $\hat{\tau}$ is close to its Monte Carlo SD, and is still smaller than that of $\hat{\tau}_{\unadj}$. However, $\hat\tau_{\rm OLS}$ suffers from both a larger bias and an underestimated SE, which together lead to an undercoverage (90\%) of the nominal 95\% Wald CI centered at $\hat{\tau}_{\rm OLS}$. Taken together, the semi-synthetic data analysis further consolidates our main conclusions.

\begin{table}[H]
\centering
\caption{Results of the semi-synthetic RCT data (all scaled by a factor of $10^2$)}
\begin{tabular}{lllll}
\Xhline{3\arrayrulewidth}
         & Bias   & SD    & SE    & CP    \\
\Xhline{3\arrayrulewidth}
$\hat\tau_{\rm unadj}$    & 0.4  & 25.2 & 24.7 & 94.8 \\
$\hat\tau_{\rm OLS}$         & -3.2 & 14.9 & 12.6 & 89.3 \\
$\hat\tau$     & 0.3  & 19.8 & 20.7 & 96.0\\
\Xhline{3\arrayrulewidth} 
\end{tabular}
\label{tab1}
\end{table}

\section{Discussion}
\label{sec:conclusion}

In this paper, we develop a new covariate-adjusted treatment effect estimator $\hat{\tau}$ based on second-order $U$-statistics for CAR under the superpopulation model, together with its consistent variance estimator. We show that $\hat{\tau}$ is $\sqrt{n}$-CAN and is more efficient than $\hat{\tau}_{\unadj}$ in the assumption-lean setting if $p = o (n)$. Empirical results indicate that $\hat{\tau}$ is competitive in terms of various metrics compared to popular benchmarks.

\paragraph{Practical recommendation} 

In view of the results obtained in our paper, we recommend the following practice, which addresses the quote on page~\pageref{quotation} to some extent. When the number $p$ of adjusted covariates is relatively large compared to $n$, it is recommended to use our proposed point and interval estimates $\hat{\tau}$ and $\hat{\sigma}^{2}$ to evaluate the ATE for CAR. The bias of $\hat{\tau}$ is always negligible compared to its standard deviation and the asymptotic variance of $\hat{\tau}$ never exceeds that of $\hat{\tau}_{\unadj}$. This suggestion is especially justified if the covariates being adjusted are prognostic factors of the outcome, as the efficiency gain of $\hat{\tau}$ is guaranteed and the standard OLS estimator $\hat{\tau}_{\rm OLS}$ in general has non-negligible bias. When $p$ is small compared to $n$, $\hat{\tau}$ and $\hat{\sigma}^{2}$ are still safe to use, having finite sample performance comparable to $\hat{\tau}_{\rm OLS}$ and $\hat{\sigma}_{\rm OLS}^{2}$. 

In summary, we recommend using the proposed estimator $\hat{\tau}$ as the primary analysis, due to its efficiency gains over $\hat{\tau}_{\unadj}$ and its reduced bias relative to $\hat{\tau}_{\rm OLS}$, particularly when the number of covariates $p$ is moderate to large compared to the sample size $n$. At the same time, we suggest reporting $\hat{\tau}_{\unadj}$ and $\hat{\tau}_{\rm OLS}$, together with their corresponding CIs, as part of a sensitivity analysis, to provide a more complete picture of their bias–efficiency performance in the given analysis. Nevertheless, any such analysis strategy should be discussed with the relevant regulatory agencies when the results are intended for regulatory submission, in accordance with current regulatory requirements.

\paragraph{Extensions}

Our theoretical results are stated under the assumption $p = o (n)$, which paves the way to the analysis of the more difficult regime $p \asymp n$. When $p < n / K$, $\hat{\Sigma}_{[k]}$ is invertible with high probability and $\hat{\tau}$ remains nearly unbiased. But, as we mentioned in Remark~\ref{rem: zeta2}, the asymptotic variance of $\hat{\tau}$ may exceed that of $\hat{\tau}_{\unadj}$ or $\hat{\tau}_{\rm OLS}$. In the superpopulation model, random matrix theory is needed for a more delicate analysis in this regime. When $p > n$, the generalized inverse $\HSigmakx^{\dag}$ of the sample Gram matrix estimator $\HSigmakx$ or various shrinkage estimators of $\Sigmakx^{-1}$ can be used \citep{liu2025augmented, ledoit2004well, ding2024eigenvector}, and the resulting estimator based on $U$-statistics is still nearly unbiased. The same challenge lies in the analysis of the asymptotic variance, the characterization of asymptotic normality, and the construction of valid CIs \citep{zheng2025perturbed}, which we leave to future work.

Let $b_{[k]} (x; a) = E_{[k]} (Y_{i} (a) | X = x)$ denote the true outcome regression for the treatment group $a$ in the stratum $k$. When $p$ is fixed, a modified version of $\hat{\tau}$, denoted by $\hat{\tau}_{\rm eff}$, achieves the semiparametric efficiency bound of $\tau$ under CAR. $\hat{\tau}_{\rm eff}$ simply replaces $X$ by a set of basis transformations of $X$, say $\{z_{l}\}_{l = 1}^{k}$ with $k = k (n) \rightarrow \infty$ as $n \rightarrow \infty$. Under standard smoothness assumptions on $b_{[k]}$, with $k$ growing with $n$, the asymptotic variance of $\hat{\tau}_{\rm eff}$ simply replaces $r_{i} (a) = Y_{i} (a) - X_{i}^{\top} \beta_{[k]}$ by $Y_{i} (a) - b_{[k]} (X_{i}; a)$ in $\sigma^{2}_{\rm OLS}$, defined in Proposition~\ref{prop1}, which is the semiparametric efficiency bound of $\tau$ under CAR \citep{rafi2023efficient}.

\paragraph{Other future directions} We conclude our paper by discussing several other future directions. An immediate question is to consider nonlinear outcome working models, probably using generalized linear models with canonical link functions as a starting point \citep{guo2023generalized, cohen2024no}. Finally, another important problem to investigate urgently is the integration of variable selection methods \citep{van2024automated} into our procedure without deteriorating the statistical properties established here. 



\putbib[reference]
\end{bibunit}

\newpage

\appendix

\doublespacing

\begin{bibunit}[apalike]

\part*{Appendix} 
\allowdisplaybreaks

Appendix~\ref{app:illustration} describes the details of the numerical experiment used as an illustration in Section~\ref{sec:intuition}. In Appendix~\ref{app:var estimator}, we present the explicit form of our proposed variance estimator $\hat{\sigma}^{2}$ of $\hat{\tau}$ omitted in Section~\ref{sec:inference}, together with the proof of theoretical results related to the asymptotic variances of $\hat{\tau}$ and $\hat{\tau}_{\rm OLS}$ (Proposition~\ref{prop:var}). In Appendix~\ref{app:proof}, we prove the main theoretical results (Theorems~\ref{thm1}--\ref{thm:var}).  Appendix~\ref{app:technical} contains technical lemmas used in the proof (Appendix~\ref{app:proof}). Given a squared matrix $A$, we use $\tr (A)$ to denote its trace functional. For a scalar-valued random variable $X$, we let $\Vert X \Vert_{2}$ denote its $L_{2} (P)$-norm.

\section{Details on the Numerical Experiment in Section~\ref{sec:intuition}}
\label{app:illustration}

In this section, we describe the details of the numerical experiment served as the illustrating example in Section~\ref{sec:intuition}. The data generating process is the same as Model~1 in Section~\ref{sec:sim}, except that $X_0 \sim t_5(0, \Sigma)$ with $\Sigma$ being the identity matrix. Randomization was performed using a stratified block procedure across four strata, with a minimum of $n=100$ subjects per stratum-treatment combination.

\section{Additional Results on the Asymptotic Variance and the Variance Estimator}
\label{app:var estimator}

\subsection{Equivalent expression of \texorpdfstring{$\zeta^2_{\mathrm{I},r}(\pi_{[k]})$}{}}

We show how to derive the equivalent expression of $\zeta^2_{\mathrm{I},r}(\pi_{[k]})$ in Section \ref{sec:inference}. Recall that in Proposition~\ref{prop1}, we have 
\[\zeta^2_{\mathrm{I},r}(\pi_{[k]}) = \sum\limits_{k=1}^K p_{[k]}\Big(\frac{\sigma^2_{[k]r(1)}}{\pi_{[k]}} + \frac{\sigma^2_{[k]r(0)}}{1-\pi_{[k]}}\Big),\]
where $r_i(a) = Y_i(a) - X_i^{\top}\beta_{[k]}$, $\beta_{[k]} = (1-\pi_{n[k]})\beta_{[k]}(1) + \pi_{n[k]}\beta_{[k]}(0)$, $\ProjBeta{a} = \Sigmak{}^{-1} \eta_{[k]}(a)$, and $\eta_{[k]} (a) = \Ek\{XY(a)\}$. Then we can rewrite $\sigma^2_{[k]r(1)}$ as
\begin{align*}
    \sigma^2_{[k]r(1)} &= \var\{Y_i(1) - X_i\beta_{[k]}|B_i=k\} \\
    &=\var\{Y_i(1)|B_i=k\} - 2 \mathrm{Cov} \{Y_i(1), X_i\beta_{[k]}|B_i=k\} + \beta^{\top}_{[k]} \Sigmakx \beta_{[k]}\\
    & = \sigma^2_{[k]Y(1)} - 2 \eta^{\top}_{[k]}(1) \Sigmakx^{-1} \{(1-\pi_{[k]})\eta_{[k]}(1) + \pi_{[k]}\eta_{[k]}(0)\}  \\
    &\quad  +\{(1-\pi_{[k]})\eta_{[k]}(1)+ \pi_{[k]}\eta_{[k]}(0)\}^{\top}\Sigmakx^{-1}  \{(1-\pi_{[k]})\eta_{[k]}(1) + \pi_{[k]}\eta_{[k]}(0)\} \\
    &= \sigma^2_{[k]Y(1)}  + \{(1-\pi_{[k]})^2 - 2(1-\pi_{[k]})\} \eta^{\top}_{{[k]}}(1)\Sigmakx^{-1}\eta_{[k]}(1) + \pi_{[k]}^{2} \eta_{{[k]}}^{\top}(0)\Sigmakx^{-1}\eta_{[k]}(0) \\
    &\quad -2\pi_{[k]}^2 \eta_{{[k]}}^{\top}(1)\Sigmakx^{-1}\eta_{[k]}(0).
    \end{align*}
Similarly, we rewrite $\sigma^2_{[k]r(0)}$ as 
\begin{align*}
    &\sigma^2_{[k]r(0)} =  \sigma^2_{[k]Y(0)}  + \{\pi_{[k]}^2 - 2\pi_{[k]}\} \eta^{\top}_{{[k]}}(0)\Sigmakx^{-1}\eta_{[k]}(0) + (1-\pi_{[k]})^{2} \eta_{{[k]}}^{\top}(1)\Sigmakx^{-1}\eta_{[k]}(1) \\
    &\quad-2(1-\pi_{[k]})^2 \eta_{{[k]}}^{\top}(1)\Sigmakx^{-1}\eta_{[k]}(0).
\end{align*}
Then we have
\begin{align*}
    \frac{\sigma^2_{[k]r(1)}}{\pi_{[k]}} + \frac{\sigma^2_{[k]r(0)}}{(1-\pi_{[k]})} &= \frac{1}{\pi_{[k]}} \sigma^2_{[k]Y(1)} + \frac{1}{1-\pi_{[k]}} \sigma^2_{[k]Y(0)} \\
    &\quad - \Bigg\{\frac{1-\pi_{[k]}}{\pi_{[k]}} \eta^{\top}_{[k]}(1)\Sigmakx^{-1}\eta_{[k]}(1) + \frac{\pi_{[k]}}{1-\pi_{[k]}} \eta^{\top}_{[k]}(0)\Sigmakx^{-1}\eta_{[k]}(0)\Bigg\} - 2\eta^{\top}_{[k]}(1)\Sigmakx^{-1}\eta_{[k]}(0).
\end{align*}

Therefore, $\zeta^2_{\mathrm{I}, r} (\pi_{[k]})$ can be equivalently represented as
\[\zeta^2_{\mathrm{I},r}(\pi_{[k]}) = \sigma^2_Y(\pi_{[k]}) - \sigma^2_{\mathrm{I},\eta}(\pi_{[k]}) - 2 \sigma_{\mathrm{I},\eta(1)\eta(0)},\]
where \[\sigma^2_{Y}(\pi_{[k]}) = \sum\limits_{k=1}^K p_{[k]}\Big[\frac{1}{\pi_{[k]}} \sigma^2_{[k]Y(1)} + \frac{1}{1-\pi_{[k]}} \sigma^2_{[k]Y(0)}\Big],\] 
\[\sigma^2_{\mathrm{I},\eta}(\pi_{[k]}) = \sum\limits_{k=1}^K p_{[k]}\Big[\frac{1-\pi_{[k]}}{\pi_{[k]}} \eta^{\top}_{[k]}(1)\Sigmakx^{-1}\eta_{[k]}(1) + \frac{\pi_{[k]}}{1-\pi_{[k]}} \eta^{\top}_{[k]}(0)\Sigmakx^{-1}\eta_{[k]}(0)\Big],\]
and
\[\sigma_{\mathrm{I},\eta(1)\eta(0)} = \sum\limits_{k=1}^K p_{[k]}\Big\{\eta^{\top}_{[k]}(1)\Sigmakx^{-1}\eta_{[k]}(0)\Big\}.\]

\subsection{Explicit form of the variance estimator}

In this part, we provide the explicit form of our new variance estimator described in Section~\ref{sec:inference}. Recall that our variance estimator $\hat{\sigma}^{2}$ reads as follows:
\[\hat\sigma^2 = \hat\zeta^2_{H} + \hat\zeta^2_{\mathrm{I},r}(\pi_{n[k]}) + \hat\zeta^2_{\mathrm{II}}.\]
We now define each term of $\hat{\sigma}^{2}$.
\begin{itemize}
    \item $\hat\zeta^2_{H}$: 
    \[\hat\zeta^2_{H} = \sum\limits_{k=1}^K p_{n[k]}\Big\{\big(\bar Y_{[k]1} - \sum\limits_{k'=1}^K p_{n[k']}\bar Y_{[k']1}\big) - \big(\bar Y_{[k]0} - \sum\limits_{k'=1}^K p_{n[k']}\bar Y_{[k']0}\big)\Big\}^2.\]
    \item $\hat\zeta^2_{\mathrm{I},r}(\pi_{n[k]})$:
    \[\hat\zeta^2_{\mathrm{I},r}(\pi_{n[k]}) =\hat\sigma^2_Y(\pi_{n[k]}) - \hat\sigma^2_{\mathrm{I},\eta}(\pi_{n[k]}) - 2 \hat\sigma_{\mathrm{I},\eta(1)\eta(0)},\]
    where 
    \[\hat\sigma^2_{Y}(\pi_{n[k]}) = \sum\limits_{k=1}^K p_{n[k]}\Big[\frac{1}{\pi_{n[k]}} \frac{1}{n_{[k]1}}\sum\limits_{i\in[k]} A_i(Y_i-\bar Y_{[k]1})^2 + \frac{1}{1-\pi_{[k]}} \frac{1}{n_{[k]0}}\sum\limits_{i\in[k]} (1-A_i)(Y_i-\bar Y_{[k]0})^2\Big],\] 
\begin{align*}
    \sigma^2_{\mathrm{I},\eta}(\pi_{n[k]}) &= \sum\limits_{k=1}^K p_{n[k]}\Big[\frac{1-\pi_{n[k]}}{\pi_{n[k]}}\frac{1}{n_{[k]1}(n_{[k]1}-1)} \sum_{\substack{1\leq i\neq j\leq n\\i,j\in [k]}} A_iA_jY_iX^{\top}_i\HSigmakx^{-1}X_jY_j \\ 
    &\quad+ \frac{\pi_{[k]}}{1-\pi_{[k]}}\frac{1}{n_{[k]0}(n_{[k]0}-1)} \sum_{\substack{1\leq i\neq j\leq n\\i,j\in [k]}} (1-A_i)(1-A_j)Y_iX^{\top}_i\HSigmakx^{-1}X_jY_j\Big],
\end{align*}
and
\[\sigma^2_{\mathrm{I},\eta(1)\eta(0)} = \sum\limits_{k=1}^K p_{n[k]}\Bigg\{\frac{1}{n_{[k]1}n_{[k]0}} \sum_{\substack{i,j\in [k]}} A_i(1-A_j) Y_iX^{\top}_i\HSigmakx^{-1}X_jY_j\Bigg\}.\]
    \item $\hat\zeta^2_{\mathrm{II}}$: 
    \[\hat\zeta^2_{\mathrm{II}} =  \HzetaSec{Y(1)} +  \HzetaSec{Y(0)}-2\HzetaSec{Y(1,0)}.\] Here, each term can be expressed as
    \begin{align*}
    \HzetaSec{Y(1)} &= \sum\limits_{k=1}^K \frac{p_{n[k]}}{n_{[k]}-1}\Big\{\frac{1-\pi_{[n]k}}{\pi^2_{n[k]}}\hat\sigma^2_{\mathrm{II},Y(1)[1]} + \frac{(1-\pi_{[n]k})^2}{\pi^2_{n[k]}}\hat\sigma^2_{\mathrm{II},Y(1)[1,2]}\Big\},\\
    \HzetaSec{Y(0)} &= \sum\limits_{k=1}^K \frac{p_{n[k]}}{n_{[k]}-1}\Big\{\frac{\pi_{[n]k}}{(1-\pi_{[n]k})^2}\hat\sigma^2_{\mathrm{II},Y(0)[1]} + \frac{\pi_{[n]k}^2}{(1-\pi_{[n]k})^2}\hat\sigma^2_{\mathrm{II},Y(0)[1,2]}\Big\},\\
    \HzetaSec{Y(1,0)} & = \sum\limits_{k=1}^K \frac{p_{n[k]}}{n_{[k]}-1} \hat\sigma^2_{\mathrm{II},Y(1,0)},
\end{align*}
where
   \[\hat\sigma^2_{\mathrm{II},Y(a)[1]} =  \frac{1}{n_{[k]a}} \sum\limits_{i\in [k]} \mathbbm{1} \{A_i=a\} X_i^{\top} \HSigmakx^{-1}X_i Y_i^2,\]
\[\hat\sigma^2_{\mathrm{II},Y(a)[1,2]} = \frac{1}{n_{[k]a}(n_{[k]a}-1)} \sum\limits_{\substack{1\leq i \neq j\leq n \\ i,j\in[k]}} \mathbbm{1} \{A_i=a\} \mathbbm{1} \{A_j=a\} (X_i^{\top} \HSigmakx^{-1}X_j)^2 Y_iY_j,\]
for each group $a \in \{0,1\}$, and
\[\hat\sigma^2_{\mathrm{II},Y(1,0)} = \frac{1}{n_{[k]1}n_{[k]0}} \sum\limits_{\substack{1 \leq i \neq j \leq k \\ i,j\in[k]}} A_i (1-A_j) (X_i^{\top} \HSigmakx^{-1}X_j)^2 Y_iY_j.\]
\end{itemize}

\subsection{Proof of Proposition 2}

In this section, we prove the statements in Proposition~\ref{prop:var}.

\begin{proof}
We first prove the first statement: ${\zeta}^{2}_{\mathrm{II}} \geq 0$. Recall that
$\zeta^2_{\mathrm{II}} = \zetaSec{Y(1)} + \zetaSec{Y(0)} - 2 \zetaSec{Y[1,2]}(1,0,1)$, and 
\begin{align*}
    \zetaSec{Y(1)} &= \zetaSec{Y[1,1]}\left(1,1,\frac{1-\pi_{[k]}}{\pi^2_{[k]}}\right)+ \zetaSec{Y[1,2]}\left(1,1,\frac{(1-\pi_{[k]})^2}{\pi^2_{[k]}}\right),\\
    \zetaSec{Y(0)} &= \zetaSec{Y[1,1]}\left(0,0,\frac{\pi_{[k]}}{(1-\pi_{[k]})^2}\right) + \zetaSec{Y[1,2]}\left(0,0,\frac{\pi^2_{[k]}}{(1-\pi_{[k]})^2}\right).
\end{align*}
By spelling out ${\zeta}^{2}_{\mathrm{II}}$, we have 
\begin{align*}
   {\zeta}^{2}_{\mathrm{II}} &=\sum\limits_{k=1}^K\Big[ \frac{p_{[k]}}{n_{[k]}-1} \frac{(1-\pi_{[k]})}{\pi^2_{[k]}} \Ek\{Y_1(1)X^{\rm T}_1\Sigmakx^{-1} X_1 Y_1(1)\} \\
   & \quad + \frac{p_{[k]}}{n_{[k]}-1} \frac{(1-\pi_{[k]})^2}{\pi^2_{[k]}}  \Ek\{(X_1^{\rm T}\Sigmakx^{-1}X_2)^2 Y_1(1)Y_2(1)\} \\
    &\quad + \frac{p_{[k]}}{n_{[k]}-1} \frac{\pi_{[k]}}{(1-\pi_{[k]})^2} \Ek\{Y_1(0)X^{\rm T}_1\Sigmakx^{-1} X_1 Y_1(0)\} \\
    & \quad + \frac{p_{[k]}}{n_{[k]}-1} \frac{\pi_{[k]}^2}{(1-\pi_{[k]})^2}  \Ek\{(X_i^{\rm T}\Sigmakx^{-1}X_j)^2 Y_j(0)Y_i(0)\}\\
    &\quad - 2\frac{p_{[k]}}{n_{[k]}-1} \Ek\{(X_i^{\rm T}\Sigmakx^{-1}X_j)^2 Y_j(0)Y_i(1)\}\Big].
\end{align*}
Note that for $a\in\{0,1\}$,
\begin{equation}\label{not1}
    E\{X_1^{\top}\Sigmakx^{-1}X_1 Y^2_1(a)\} = E\{(X_1^{\top}\Sigmakx^{-1}X_2)^2 Y^2_1(a)\} =E\{(X_1^{\top}\Sigmakx^{-1}X_2)^2 Y^2_2(a)\}.
\end{equation}
Therefore, we can write ${\zeta}^{2}_{\mathrm{II}}$ as 
\begin{equation}
\label{sos}
\zeta^{2}_{\mathrm{II}} =\sum\limits_{k=1}^K\frac{p_{[k]}}{n_{[k]}-1}\Ek\Big[(X_1^{\top}\Sigmakx^{-1}X_2)^2g(\pi_{[k]},Y_1(0),Y_2(0),Y_1(1),Y_2(1))\Big],
\end{equation}
where
\begin{align*}
    &g(\pi_{[k]},Y_1(0),Y_2(0),Y_1(1),Y_2(1))\\
    &=\frac{(1-\pi_{[k]})}{\pi^2_{[k]}} Y^2_1(1) +  \frac{(1-\pi_{[k]})^2}{\pi^2_{[k]}}Y_1(1)Y_2(1) \\
    &\quad+ \frac{\pi_{[k]}}{(1-\pi_{[k]})^2} Y^2_1(0) + \frac{\pi_{[k]}^2}{(1-\pi_{[k]})^2}  Y_1(0)Y_2(0) - Y_1(0)Y_2(1) - Y_2(0)Y_1(1).
\end{align*}

By \eqref{not1}, as $Y_1$ and $Y_2$ are i.i.d.\ copies, conditional on stratum $B_i=k$, we have
\[\Ek\Big\{(X_1^{\top}\Sigmakx^{-1}X_2)^2\frac{(1-\pi_{[k]})}{\pi^2_{[k]}} Y^2_1(1)\Big\} = \Ek\Big[(X_1^{\top}\Sigmakx^{-1}X_2)^2\Big\{\frac{(1-\pi_{[k]})}{2\pi^2_{[k]}} Y^2_1(1)+\frac{(1-\pi_{[k]})}{2\pi^2_{[k]}} Y^2_2(1)\Big\}\Big],\]
and
\[\Ek\Big\{(X_1^{\top}\Sigmakx^{-1}X_2)^2\frac{\pi_{[k]}}{(1-\pi_{[k]})^2} Y^2_1(0)\Big\} = \Ek\Big[(X_1^{\top}\Sigmakx^{-1}X_2)^2\Big\{\frac{\pi_{[k]}}{2(1-\pi_{[k]})^2} Y^2_1(0) + \frac{\pi_{[k]}}{2(1-\pi_{[k]})^2} Y^2_2(0)\Big\}\Big].\]
In turn, we can replace $g(\pi_{[k]},Y_1(0),Y_2(0),Y_1(1),Y_2(1))$ by the following random variable in \eqref{sos}:
\begin{align*}
    & \frac{(1-\pi_{[k]})}{\pi^2_{[k]}} Y^2_1(1) +  \frac{(1-\pi_{[k]})^2}{\pi^2_{[k]}}Y_1(1)Y_2(1) + \frac{\pi_{[k]}}{(1-\pi_{[k]})^2} Y^2_1(0) \\
    &\quad + \frac{\pi_{[k]}^2}{(1-\pi_{[k]})^2}  Y_1(0)Y_2(0) - Y_1(0)Y_2(1) - Y_2(0)Y_1(1)\\
    &=\frac{1}{2\pi_{[k]}^2(1-\pi_{[k]})^2}\Bigg\{(1-\pi_{[k]})^4\{Y_1^2(1) + 2Y_1(1)Y_2(1) + Y_2(1)^2\} + \pi_{[k]}^4 \{Y_1^2(0) + 2Y_1(0)Y_2(0) + Y_2(0)^2\} \\
    &\quad + \pi_{[k]}(1-\pi_{[k]})\Big[\Big\{(1-\pi_{[k]})^2Y_1(1)^2 - 2\pi_{[k]}(1-\pi_{[k]}) Y_1(1)Y_2(0) + \pi^2_{[k]}Y^2_2(0)\Big\}  \\
    &\quad+ \Big\{(1-\pi_{[k]})^2Y_2(1)^2 - 2\pi_{[k]}(1-\pi_{[k]}) Y_2(1)Y_1(0)+ \pi^2_{[k]}Y^2_1(0)\Big\}\Big]\Bigg\}\\
    &= \frac{1}{2} \left(\frac{1-\pi_{[k]}}{\pi_{[k]}}\right)^2 \{Y_1(1) + Y_2(1)\}^2 + \frac{1}{2} \left(\frac{\pi_{[k]}}{1-\pi_{[k]}}\right)^2 \{Y_1(0) + Y_2(0)\}^2 \\
    &\quad + \frac{1}{\pi_{[k]}(1-\pi_{[k]})}\{(1-\pi_{[k]})Y_1(1) - \pi_{[k]}Y_2(0)\}^2\\
    &\geq 0,
\end{align*}
where the second equality follows from elementary algebra
\[\frac{1-\pi_{[k]}}{\pi^2_{[k]}} = \frac{(1-\pi_{[k]})^2}{\pi^2_{[k]}} + \frac{(1-\pi_{[k]})}{\pi_{[k]}},\]
and 
\[\frac{\pi_{[k]}}{(1-\pi_{[k]})^2} = \frac{\pi_{[k]}^2}{(1-\pi_{[k]})^2} + \frac{\pi_{[k]}}{(1-\pi_{[k]})}.\]
Then it is straightforward to see that $\zeta^{2}_{\rm II}$ can be written as a sum-of-squares, and hence is non-negative:
\begin{align}\label{Full zeta2}
    &{\zeta}^{2}_{\mathrm{II}} = \sum\limits_{k=1}^K \frac{p_{[k]}}{n_{[k]}-1}\Ek \Bigg[\{X_1^{\top}\Sigmakx^{-1}X_2\}^2 \Big\{\frac{1}{2} \left(\frac{1-\pi_{[k]}}{\pi_{[k]}}\right)^2 \{Y_1(1) + Y_2(1)\}^2 \notag \\
    & \quad + \frac{1}{2} \left(\frac{\pi_{[k]}}{1-\pi_{[k]}}\right)^2 \{Y_1(0) + Y_2(0)\}^2 + \frac{1}{\pi_{[k]}(1-\pi_{[k]})}\{(1-\pi_{[k]})Y_1(1) - \pi_{[k]}Y_2(0)\}^2\Big\}\Bigg].
\end{align}
Next, we prove part (b) of Proposition~\ref{prop:var}, which can be proved by showing ${\zeta}^{2}_{\mathrm{II}} = O (p/n)$ and if $p=o(n)$, ${\zeta}^{2}_{\mathrm{II}} = o (1)$. By \eqref{Full zeta2} and Assumption~\ref{ap2}, we have $\Ek [Y_i^2(a)|X_i]\leq M$, and thus 
\[{\zeta}^{2}_{\mathrm{II}}\leq \sum\limits_{k=1}^K \frac{Mp_{[k]}}{n_{[k]}-1}\Ek\{(X^{\top}_1\Sigmakx^{-1}X_2)^2\}, \]
and 
\[\Ek\{(X^{\top}_1\Sigmakx^{-1}X_2)^2\} = \Ek \{X_1^{\top}\Sigmakx^{-1}X_1\} = \tr\{\Ek(X_1X_1^{\top})\Sigmakx^{-1}\}=p.\]
Therefore, by the Markov inequality and the assumption that $K$ is bounded, we have 
\[{\zeta}^{2}_{\mathrm{II}} = O \Big(\frac{p}{n}\Big).\]
When $p=o(n)$, 
\begin{align*}
    \sigma^2 &= \zeta^2_{H} + \zeta^2_{\mathrm{I},r}(\pi_{[k]}) + \zeta^2_{\mathrm{II}} \\
    &= \zeta^2_{H} + \zeta^2_{\mathrm{I},r}(\pi_{[k]}) + o (1) \\
    &= \sigma^2_{\rm OLS} + o (1).
\end{align*}
\end{proof}

\section{Proof of the Main Theoretical Results}
\label{app:proof}

\subsection{Proof of Theorem~\ref{thm1}}
\label{app:thm1}
Recall that 
\[\hat\tau_{\oracle} = \sum\limits_{k=1}^K p_{n[k]}\Big[\Big\{\bar Y_{[k]1} - \mathbf{U}_{n_{[k]},2}(\Sigmakx;1) \Big\}-\Big\{\bar Y_{[k]0} - \mathbf{U}_{n_{[k]},2}(\Sigmakx;0) \Big\}\Big].\]
Define 
\[\hat\tau_{\oracle,a} = \sum\limits_{k=1}^K p_{n[k]}\Big\{\bar Y_{[k]a} - \mathbf{U}_{n_{[k]},2}(\Sigmakx;a) \Big\},\ a=0,1,\]
and 
\[\tau_{a} = E\{Y_i(a)\} = \sum\limits_{k=1}^K p_{[k]}\Ek\{Y_i(a)\},\ a=0,1.\]
We first show that $\hat\tau_{\oracle, 1}$ (by symmetry $\hat{\tau}_{\oracle, 0}$) is $\sqrt{n}$-CAN for $\tau_{1}$ (by symmetry $\tau_{0}$) and then derive the asymptotic variance for $\hat\tau_{\oracle}$. First, we have the following decomposition:
\begin{align*}
    \sqrt{n}&(\hat\tau_{\oracle,1}-\tau_1)\\ 
    &= 
\sqrt{n}\left[\sum\limits_{k=1}^K p_{n[k]}\Big\{\bar Y_{[k]1} -  \mathbf{U}_{n_{[k]},2}(\Sigmakx;1)\Big\} -\sum\limits_{k=1}^K p_{[k]} \Ek\{Y_i(1)\} \right] \\
&= \underbrace{\frac{1}{\sqrt n}\sum\limits_{k\in\mathcal{K}}\left\{\frac{1}{\pi_{[k]}}\sum\limits_{i=1}^n A_iI(B_i=k) \tilde Y_i(1)\right\}}_{\eqqcolon M_1} \\
& \quad - \underbrace{\frac{1}{\sqrt n}\sum\limits_{k\in\mathcal{K}}\left\{\frac{1}{n_{[k]}-1}\frac{1}{\pi^2_{[k]}} \sum_{\substack{1\leq i\neq j\leq n \\ i,j \in [k]}} (A_i -\pi_{n[k]})A_jX_i^{\top} \Sigmakx^{-1}X_j Y_j(1)\right\}}_{\eqqcolon M_2} \\
&\quad + \underbrace{\sqrt{n}\sum\limits_{k=1}^K (p_{n[k]} - p_{[k]}) \Ek\{Y_i(1)\}}_{\eqqcolon M_3} + o_P(1)\\
&\eqqcolon M_1 - M_2 + M_3.
\end{align*}

Here we use the coupling method in \citet{Bugni2018}. We order the data $(Y_i(0),Y_i(1),X_i)$ first by stratum $k$ and then by treatment $a$ in each stratum (placing units in the treatment group first and then the control group), leaving the distributions of $\bar Y_{[k]1}$ and $\mathbf{U}_{n_{[k]},2}(\Sigmakx;1)$ invariant. Then independently for each stratum $k$ and independently of $(A^{(n)},B^{(n)})$, let $(Y^k_i(0),Y^k_i(1),X^k_i)$ be i.i.d.\ draws with marginal distribution equal to the distribution $$(Y_i(0),Y_i(1),X_i)|B_i=k.$$ 
Define $N(k) = \sum_{i=1}^n I(B_i<k)$. Then we have 
 \[\sum\limits_{i=1}^n A_iI(B_i=k) \tilde Y_i(1) \stackrel{d}=\sum\limits_{i=N(k)+1}^{N(k)+n_{[k]1}} \tilde Y_i^k,\]
and
\[\sum_{\substack{1\leq i\neq j\leq n \\ i,j \in [k]}} (A_i -\pi_{n[k]})A_jX_i^\top\Sigmakx^{-1}X_j Y_j(1) \stackrel{d}= \sum_{\substack{N(k)+1\leq i\neq j\leq N(k+1)}} (A_i -\pi_{n[k]})A_jX^{k\top}_i\Sigmakx^{-1}X^k_j Y^k_j(1).\]
Then in stratum $k$, we define 
\[M_1^k \coloneqq \frac{1}{\sqrt n} \frac{1}{\pi_{[k]}} \sum\limits_{i=N(k)+1}^{N(k)+n_{[k]1}} \tilde Y_i^k\]
and 
    \[M^k_2 \coloneqq \frac{1}{\sqrt n(n_{[k]}-1)}\frac{1}{\pi^2_{[k]}} \sum_{\substack{N(k)+1\leq i\neq j\leq N(k+1)}} (A_i -\pi_{n[k]})A_jX^{k\top}_i\Sigmakx^{-1}X^k_j Y^k_j(1).\]
By this construction, we have 
\[M^k_1 \stackrel{d}=M_1|B_i=k,\]
and
\[M^k_2 \stackrel{d}=M_2|B_i=k,\]
conditional on $(A^{(n)},B^{(n)})$.
We introduce the short-hand notation $\Eka \{\cdot\}= E\{\cdot|B_i=k,A_i=1\}$.
Therefore, we have
\[\Eka\{M_1^k\} = 0,\]
and
\begin{align*}
        \Eka\{M^k_2\} &= \frac{1}{\sqrt n(n_{[k]}-1)}\frac{1}{\pi^2_{[k]}}\sum\limits_{i=N(k)+1}^{N(k+1)}\sum\limits_{\substack{j=N(k)+1 \\j\neq i}}^{N(k+1)} (A_i - \pi_{n[k]})A_j \Ek(X^{k\top}) \Sigma^{-1}_{[k]} \Ek(X^k Y^k(1))\\
        &= \frac{n_{[k]1}}{\sqrt n(n_{[k]}-1)}\frac{1}{\pi^2_{[k]}} (1-\pi_{n[k]})\Ek(X^{k\top}) \Sigma^{-1}_{[k]} \Ek(X^k Y^k(1)).
\end{align*}
By Assumption \ref{ap1}, as $\pi_{n[k]}\stackrel{p}\to \pi_{[k]}$, we have 
\[\Eka\{M^k_2\}\stackrel{P} \to 0.\]
Also, since the number of strata $K$ is finite, we have
\[\sqrt{n}\sum\limits_{k=1}^K (p_{n[k]} - p_{[k]}) = O_P (1).\]
Thus we have
\[E\{\sqrt{n}(\hat\tau_{\oracle,1}-\tau_1)\}\to0.\]

Next, we will show that $\sqrt{n}(\hat\tau_{\oracle,1}-\tau_1)$ is asymptotic normal and compute its asymptotic variance. By the coupling method, $M_1^k$ and $M_2^k$ are independent of $B^{(n)}$, which implies that $M_3$ (whose randomness depends only on $B_i$) is independent of both $M_1^k$ and $M_2^k$. By CLT, we have
\begin{equation}\label{AN1}
    M_3 \stackrel{d}\to \operatorname{N}(0,\sigma^2_{HY(1)}),
\end{equation}
where
\[\sigma^2_{HY(1)} = \sum\limits_{k=1}^K p_{[k]}[\Ek\{Y_i(1)\}]^2.\]
Therefore, we only need to show the asymptotic normality of $M_1^k$ and $M_2^k$. Regarding $M^k_1$, we can directly invoke Lemma B.2 in \citet{Bugni2018} to obtain:
\begin{equation}\label{AN2}
  M_1^k \stackrel{d}\to \operatorname{N}(0,\sigma^2_{Y(1)}),  
\end{equation}
where 
\[\sigma^2_{Y(1)} =\sum\limits_{k=1}^K\frac{p_{[k]}}{\pi_{[k]}} \sigma^2_{[k]Y(1)}.\]
In terms of $M_2^k$, by Lemma~\ref{lem:BG} proved later in Appendix~\ref{app:U-stats}, we have 
\begin{equation}\label{AN3}
    M_2^k \stackrel{d}\to \operatorname{N}(0,\zetaSec{Y(1)} + \sigma^2_{\mathrm{I},X(1)}),
\end{equation}
where
\[\sigma^2_{\mathrm{I},X(1)} = \frac{p_{[k]}(1-\pi_{[k]})}{\pi_{[k]}} \ProjBeta{1}^{\top}\Sigmakx \ProjBeta{1}.\]
By the same token, for any $a_1,a_2 \in \mathbb{R}$, we have $a_1M_1^k + a_2M_2^k$ is asymptotic normal, leading to the joint asymptotic normality of $(M_1^k,M_2^k)^{\top}$. We also calculate the covariance between  $M_1^k $ and $ M_2^k$, which is
\begin{align*}
\mathrm{Cov}&\Big\{\frac{1}{\sqrt{n}}\frac{1}{\pi_{[k]}} \sum\limits_{i= N(k)+1}^{N(k+1)} A_i Y_i(1), \frac{1}{\sqrt{n}(n_{[k]}-1)}\frac{1}{\pi^2_{[k]}} \sum\limits_{i= N(k)+1}^{N(k+1)} \sum_{\substack{j=N(k)+1 \\ j\neq i}}^{N(k+1)} (A_i - \pi_{[k]})(X_i^{k\top} \Sigmakx^{-1}X_j^k A_jY^k_j(1)\Big\} \\
    &= \frac{1}{n(n_{[k]}-1)}\frac{1}{\pi^3_{[k]}}\sum\limits_{i= N(k)+1}^{N(k+1)} \sum_{\substack{j=N(k)+1 \\ j\neq i}}^{N(k+1)} A_i (A_i - \pi_{[k]})A_j E \left\{ Y^k_i(1)  (X_i^{k\top} \Sigmakx^{-1} X_j^k Y^k_j(1)\right\}\\
    &= \frac{1}{n(n_{[k]}-1)}\frac{1}{\pi^3_{[k]}}\sum\limits_{i= N(k)+1}^{N(k+1)} \sum_{\substack{j=N(k)+1 \\ j\neq i}}^{N(k+1)} (1-\pi_{[k]})A^k_iA^k_j \ProjBeta{1}^{\top}\Sigmakx \ProjBeta{1} \\
    & = \frac{p_{[k]}(1-\pi_{[k]})}{\pi_{[k]}} \ProjBeta{1}^{\top}\Sigmakx \ProjBeta{1} + o (1) = \sigma^{2}_{\rm I, X (1)} + o (1).
\end{align*}

Combining the above covariance formula with \eqref{AN1}, \eqref{AN2}, and \eqref{AN3}, we have
\[ \sqrt{n}(\hat\tau_{\oracle,1}-\tau_1)\stackrel{d}\to \operatorname{N}(0,\sigma^2_{HY(1)}+\sigma^2_{Y(1)} + \zetaSec{Y(1)} - \sigma^2_{\mathrm{I},X(1)}).\]
By symmetry, we also have:
\[ \sqrt{n}(\hat\tau_{\oracle,0}-\tau_1)\stackrel{d}\to \operatorname{N}(0,\sigma^2_{HY(0)}+\sigma^2_{Y(0)} + \zetaSec{Y(0)} - \sigma^2_{\mathrm{I},X(0)}).\]
Finally, it is not difficult to conclude that
\[\sqrt{n}(\hat\tau_{\oracle} - \tau) \stackrel{d}\to \operatorname{N}(0,\zeta^2_{H} + \zeta^2_{\mathrm{I},r}(\pi_{[k]}) + \zeta^2_{\mathrm{II}}),\]
where $\zeta^2_{\mathrm{II}} = \zetaSec{Y(1)} + \zetaSec{Y(0)} - 2 \zetaSec{Y[1,2]}(1,0,1)$, 
\begin{align*}
    \zetaSec{Y(1)} &= \zetaSec{Y[1,1]}\left(1,1,\frac{1-\pi_{[k]}}{\pi^2_{[k]}}\right)+ \zetaSec{Y[1,2]}\left(1,1,\frac{(1-\pi_{[k]})^2}{\pi^2_{[k]}}\right), \text{ and} \\
    \zetaSec{Y(0)} &= \zetaSec{Y[1,1]}\left(0,0,\frac{\pi_{[k]}}{(1-\pi_{[k]})^2}\right) + \zetaSec{Y[1,2]}\left(0,0,\frac{\pi^2_{[k]}}{(1-\pi_{[k]})^2}\right).
\end{align*}

\subsection{Proof of Theorem~\ref{thm:main}}\label{app:thm:main}
Recall that 
\[\hat\tau = \sum\limits_{k=1}^K p_{n[k]}\Big[\Big\{\bar Y_{[k]1} - \mathbf{U}_{n_{[k]},2}(\HSigmakx^{-1};1) \Big\}-\Big\{\bar Y_{[k]0} - \mathbf{U}_{n_{[k]},2}(\HSigmakx^{-1};0) \Big\}\Big].\]
We can decompose $\sqrt{n} (\hat{\tau} - \tau)$ as follows:
\[\sqrt{n}(\hat\tau - \tau) = \sqrt{n}(\hat\tau - \hat\tau_{\oracle}) + \sqrt{n}(\hat\tau_{\oracle} - \tau).\]
By Theorem~\ref{thm1}, we have already obtained the following:
\[\sqrt{n}(\hat\tau_{\oracle} - \tau)/\sigma \stackrel{d}\to \operatorname{N}(0,1).\]
We are left to show that 
\[\sqrt{n}(\hat\tau - \hat\tau_{\oracle}) = o_P(1),\]
where the left-hand side can be written as 
\[\sqrt{n}(\hat\tau - \hat\tau_{\oracle}) =\frac{1}{\sqrt{n}}\sum\limits_{k=1}^K n_{[k]}\Big\{\mathbf{U}_{n_{[k]},2}(\HSigmakx^{-1} - \Sigmakx^{-1};1) - \mathbf{U}_{n_{[k]},2}(\HSigmakx^{-1} - \Sigmakx^{-1};0)\Big\}.\]
Therefore, we need to show that
\begin{equation}
\label{desired}
\frac{1}{\sqrt{n}}\sum\limits_{k=1}^K n_{[k]}\Big\{\mathbf{U}_{n_{[k]},2}(\HSigmakx^{-1} - \Sigmakx^{-1};1) - \mathbf{U}_{n_{[k]},2}(\HSigmakx^{-1} - \Sigmakx^{-1};0)\Big\} = o_P(1),
\end{equation}
which essentially reduces to show that particular functionals of $\hat{\Sigma}_{[k]}^{-1} - \Sigma_{[k]}^{-1}$ vanish to zero in probability at a sufficiently fast speed.

We first show that in the treatment group, 
\begin{equation}
\label{desired treatment}
\frac{1}{\sqrt{n}}\sum\limits_{k=1}^K n_{[k]}\mathbf{U}_{n_{[k]},2}(\HSigmakx^{-1} - \Sigmakx^{-1};1) = o_P(1).
\end{equation}
Similar to the proof of Theorem \ref{thm1}, we use the coupling method. Let $(Y^k_i(0),Y^k_i(1),X_i^k)$ be i.i.d.\ and have the same marginal distribution as $(Y_i(0),Y_i(1),X_i)|B_i=k$, independent of $(A^{(n)},B^{(n)})$. We also define
\[\mathbf{U}^k_{n_{[k]},2}(\Sigmakx^{-1};1) = \frac{1}{\B{n}(\B{n}-1)}\frac{1}{\pi^2_{n[k]}} \sum\limits_{N(k)+1\leq i\neq j\leq N(k+1)} (A_i - \pi_{n[k]})X_{i}^{k\top} \Sigmakx^{-1} A_j X^k_j Y^k_j, \]
and
\[\mathrm{Res}_k(X^k,Y^k(1))=\frac{1}{\sqrt{n}}n_{[k]}\mathbf{U}^k_{n_{[k]},2}(\HSigmakx^{-1} - \Sigmakx^{-1};1).\]
To facilitate the proof, we spell out the above notation $\mathrm{Res}_k(X^k,Y^k(1))$ as follows:
\begin{align*}
    &\mathrm{Res}_k(X^k,Y^k(1))\\
    &= \frac{1}{\sqrt n(n_{[k]}-1)}\frac{1}{\pi^2_{n[k]}} \sum\limits_{N(k)+1\leq i\neq j\leq N(k+1)} (A_i -\pi_{n[k]})X_{i}^{k \top} (\HSigmakx^{-1} - \Sigmakx^{-1}) X^k_j A_jY^k_j.
\end{align*}

\begin{remark}
It should be noted that we also use $X_i^k$ to construct a second-moment sample matrix that has the same distribution as $\HSigmakx$ conditional on $B_i=k$. But we decide to abuse the notation and use $\HSigmakx$ without adding the superscript to avoid clutter.
\end{remark}

As $\mathbf{U}^k_{n_{[k]},2}(\HSigmak{XX}^{-1} - \Sigmakx^{-1};1)$ is independent of $B_i$ and the total number of strata $K$ is fixed, we will show that 
\[\frac{1}{\sqrt n(n_{[k]}-1)}\frac{1}{\pi^2_{n[k]}} \sum\limits_{N(k)+1\leq i\neq j\leq N(k+1)} (A_i -\pi_{n[k]})X_{i}^{k \top} (\HSigmakx^{-1} - \Sigmakx^{-1}) X^k_j A_jY^k_j = o_P(1),\]
which can help to obtain the desired conclusion. We first compute the mean of $\mathrm{Res}_k(X^k,Y^k(1))$, which could justify part (a) in Theorem \ref{thm:main}. As introduced in the Proof of Theorem 1, $\Eka \{\cdot\}= E\{\cdot|B_i=k,A_i=1\}$, and then we have
\begin{align*}
    &\Eka\{\Res(X^k,Y^k(1))\}\\
    &=\frac{1}{\sqrt n(n_{[k]}-1)}\Eka\Big\{\frac{1}{\pi^2_{n[k]}} \sum\limits_{N(k)+1\leq i\neq j\leq N(k+1)} (A_i -\pi_{n[k]})X^{\top}_i (\HSigmakx^{-1} - \Sigmakx^{-1}) X_j A_jY_j\Big\}\\
    &= \frac{1}{\sqrt n(n_{[k]}-1)} \sum\limits_{N(k)+1\leq i\neq j\leq N(k+1)} \frac{1}{\pi^2_{n[k]}}(A_i -\pi_{n[k]})A_j\   \Eka\Big\{X_{i}^{k \top} (\HSigmakx^{-1} - \Sigmakx^{-1}) X^k_jY^k_j(1)\Big\}\\
    &=  \frac{1}{\sqrt n(n_{[k]}-1)} \sum\limits_{N(k)+1\leq i\neq j\leq N(k+1)} \frac{1}{\pi^2_{n[k]}}(A_i -\pi_{n[k]})A_j \ \Ek\Big\{X_{1}^{k \top} (\HSigmakx^{-1} - \Sigmakx^{-1}) X^k_2Y^k_2(1)\Big\}.
\end{align*}
We directly compute and analyze the order of $ \sum\limits_{N(k)+1\leq i\neq j\leq N(k+1)} (A_i -\pi_{n[k]})A_j$: 
\begin{align*}
    &\frac{1}{n_{[k]}-1}\sum\limits_{N(k)+1\leq i\neq j\leq N(k+1)} (A_i -\pi_{n[k]})A_j\\ 
    &= \frac{1}{n_{[k]}-1}\sum\limits_{N(k)+1\leq i\neq j\leq N(k+1)} (A_iA_j - \pi_{n[k]}A_j)\\
    &= \frac{1}{n_{[k]}-1}\{n_{[k]a}(n_{[k]a}-1) - \frac{n_{[k]a}}{n_{[k]}}(n_{[k]}-1)n_{[k]a}\} \\
    &= (\pi_{n[k]}-1)\frac{n_{[k]a}}{n_{[k]}-1} = O_P (1).
\end{align*}

We then bound $\Ek\Big\{X_1^{k\top} (\HSigmakx^{-1} - \Sigmakx^{-1}) X^k_2Y^k_2(1)\Big\}$. We first write 
\begin{align*}
\HSigmakx = (\HSigmakx)_{-i} + \frac{1}{n} X_iX_i^{\top}, \text{ where } (\HSigmakx)_{-i} = \frac{1}{n} \sum\limits_{l\neq i} X_l X_l^{\top}
\end{align*}
denotes the leave-$i$-out version of the sample Gram matrix estimator. By Lemma \ref{lem:Sherman-Morrison} presented later in Appendix~\ref{app:matrices}, we have 
\[\HSigmakx^{-1} = (\HSigmakx)_{-i}^{-1} -(\HSigmakx)_{-i}^{-1} X_i(n + X_i^{\top}(\HSigmakx)_{-i}^{-1}X_i)^{-1}X_i^{\top}(\HSigmakx)_{-i}^{-1}.\]

Then we have
\begin{align*}
    &\Ek\Big\{X_{1}^{k \top} (\HSigmakx^{-1} - \Sigmakx^{-1}) X^k_2 Y_2^k\Big\} \\
    &= \Ek\Big[\Big\{X_{1}^{k \top} ((\HSigmakx)_{-1}^{-1} - \Sigmakx^{-1})X^k_2 - \frac{X_{1}^{k \top}(\HSigmakx)_{-1}^{-1} X^k_1}{n+X_{1}^{k \top}(\HSigmakx)_{-1}^{-1} X^k_1} X_1^{k\top}(\HSigmakx)_{-1}^{-1} X_2^k\Big\}Y_2^k\Big]\\
    &= \Ek\Big[\Big\{X_{1}^{k \top} ((\HSigmakx)_{-(1,2)}^{-1} - \Sigmakx^{-1})X^k_2 - \frac{X_{2}^{k \top}(\HSigmakx)_{-(1,2)}^{-1} X^k_2}{n+X_{2}^{k \top}(\HSigmakx)_{-(1,2)}^{-1} X^k_2} X_1^{k\top}(\HSigmakx)_{-(1,2)}^{-1} X_2^k\Big\}Y_2^k\Big] \\
    &\quad - \Ek\Bigg[\Big\{\frac{X_{1}^{k \top}(\HSigmakx)_{-1}^{-1} X^k_1}{n+X_{1}^{k \top}(\HSigmakx)_{-1}^{-1} X^k_1} \Big( 1-\frac{X_{2}^{k \top}(\HSigmakx)_{-(1,2)}^{-1} X^k_2}{n+X_{2}^{k \top}(\HSigmakx)_{-(1,2)}^{-1} X^k_2} \Big) \Big\} X_1^{k\top}(\HSigmakx)_{-(1,2)}^{-1} X_2^k Y_2^k \Bigg].
\end{align*}
The first line of the above display is $o (1)$, because according to Lemma~\ref{lem: sample covariance1} in Appendix~\ref{app:matrices},
\begin{align*}
\Ek \big[ X^{k\top}_1 \{(\HSigmakx)^{-1}_{-(1,2)} - \Sigmakx^{-1}\} X_2^k Y_2^k \big] = o(1)
\end{align*}
and
\begin{align*}
\Ek \Big\{X_{1}^{k\top} (\HSigmakx)_{-(1,2)}^{-1} X_{2}^{k} \frac{X_{2}^{k \top}(\HSigmakx)^{-1}_{-(1,2)} X^k_2}{n+X_{2}^{k \top}(\HSigmakx)_{-(1,2)}^{-1} X^k_2} Y_{2} \Big\} = o(1).
\end{align*}

We are left to analyze the term 
\begin{equation}
\label{part 2}
\Ek\Bigg[\Big\{\frac{X_{1}^{k \top}(\HSigmakx)_{-1}^{-1} X^k_1}{n+X_{1}^{k \top}(\HSigmakx)_{-1}^{-1} X^k_1} \Big( 1-\frac{X_{2}^{k \top}(\HSigmakx)_{-(1,2)}^{-1} X^k_2}{n+X_{2}^{k \top}(\HSigmakx)_{-(1,2)}^{-1} X^k_2} \Big) \Big\} X_1^{k\top}(\HSigmakx)_{-(1,2)}^{-1} X_2^k Y_2^k \Bigg].
\end{equation}
By Lemma~\ref{lem:Sherman-Morrison} (Sherman-Morrison formula), we have 
\begin{equation}
\label{S-M-twice}
X_{1}^{k \top}(\HSigmakx)_{-1}^{-1} X^k_1 = \Big( 1-\frac{X_{2}^{k \top}(\HSigmakx)_{-(1,2)}^{-1} X^k_2}{n+X_{2}^{k \top}(\HSigmakx)_{-(1,2)}^{-1} X^k_2} \Big) X_{1}^{k \top}(\HSigmakx)_{-(1,2)}^{-1} X^k_1,
\end{equation}
which leads to
\begin{equation}
\label{S-M-twice-consequence}
\frac{X_{1}^{k \top}(\HSigmakx)_{-1}^{-1} X^k_1}{n+X_{1}^{k \top}(\HSigmakx)_{-1}^{-1} X^k_1} = \frac{X_{1}^{k \top}(\HSigmakx)_{-(1,2)}^{-1} X^k_1}{n + X_{1}^{k \top}(\HSigmakx)_{-(1,2)}^{-1} X^k_1 + X_{2}^{k \top}(\HSigmakx)_{-(1,2)}^{-1} X^k_2}.
\end{equation}

Applying \eqref{S-M-twice-consequence}, we can rewrite \eqref{part 2} as
\begin{align*}
\eqref{part 2} = \Ek \left[ \frac{(X_{1}^{k \top} (\HSigmakx)_{-(1,2)}^{-1} X^k_1 / n) X_{1}^{k\top} (\HSigmakx)_{-(1,2)}^{-1} \frac{1}{1 + X_{2}^{k \top}(\HSigmakx)_{-(1,2)}^{-1} X^k_2 / n} X_{2}^{k} Y_{2}^{k}}{1 + \{X_{1}^{k \top} (\HSigmakx)_{-(1,2)}^{-1} X^k_1 + X_{2}^{k \top}(\HSigmakx)_{-(1,2)}^{-1} X^k_2\} / n} \right].
\end{align*}
Under Assumption~\ref{ap2}, we have $X_{j}^{k \top} (\HSigmakx)_{-(1, 2)}^{-1} X^k_j / n \lesssim p / n$ for $j = 1, 2$, then we have 
\begin{align*}
& \ \frac{1}{1 + \{X_{1}^{k \top}(\HSigmakx)_{-(1, 2)}^{-1} X^k_1 + X_{2}^{k \top}(\HSigmakx)_{-(1, 2)}^{-1} X^k_2\} / n} \\
= & \ \sum\limits_{j=1}^m (-1)^j \Big(\frac{X_{1}^{k \top}(\HSigmakx)_{-(1, 2)}^{-1} X^k_1 + X_{2}^{k \top}(\HSigmakx)_{-(1, 2)}^{-1} X^k_2}{n}\Big)^{j} + O \left( \left( \frac{p}{n} \right)^{m + 1} \right).
\end{align*}
Because $p = o(n)$, we can always find some possibly diminishing sequence $\gamma = \gamma_{n} \in (0, 1)$ such that $p/n = n^{-\gamma}$. Then for large enough $m$, except the summation, the remaining term of \eqref{part 2} is smaller than 
\[\Ek\Bigg\{\frac{(X_{1}^{k \top} (\HSigmakx)_{-(1,2)}^{-1} X^k_1 / n) X_{1}^{k\top} (\HSigmakx)_{-(1,2)}^{-1} \frac{1}{1 + X_{2}^{k \top}(\HSigmakx)_{-(1,2)}^{-1} X^k_2 / n} X_{2}^{k} Y_{2}^{k}}{n^{\gamma(m+1)}}\Bigg\}=o(1)\]
which is obtained by Lemma~\ref{lem: sample covariance1}. Therefore, we have
\begin{equation}
\label{key decomposition}
\begin{split}
    &\quad\eqref{part 2} \\
    &= \Ek \left[ \frac{(X_{1}^{k \top} (\HSigmakx)_{-(1,2)}^{-1} X^k_1 / n) X_{1}^{k\top} (\HSigmakx)_{-(1,2)}^{-1} \frac{1}{1 + X_{2}^{k \top}(\HSigmakx)_{-(1,2)}^{-1} X^k_2 / n} X_{2}^{k} Y_{2}^{k}}{1 + \{X_{1}^{k \top} (\HSigmakx)_{-(1,2)}^{-1} X^k_1 + X_{2}^{k \top}(\HSigmakx)_{-(1,2)}^{-1} X^k_2\} / n} \right] \\
    &=\Ek \Bigg[\sum\limits_{j=1}^m (-1)^j \Big(\frac{X_{1}^{k \top}(\HSigmakx)_{-(1, 2)}^{-1} X^k_1 + X_{2}^{k \top}(\HSigmakx)_{-(1, 2)}^{-1} X^k_2}{n}\Big)^{j} \frac{X_{1}^{k \top} (\HSigmakx)_{-(1,2)}^{-1} X^k_1}{n} \\
    & \quad \quad \quad \quad X_{1}^{k\top} (\HSigmakx)_{-(1,2)}^{-1}X_{2}^{k} \frac{Y_2^k}{1 + X_{2}^{k \top}(\HSigmakx)_{-(1,2)}^{-1} X^k_2 / n}\Bigg] + o(1).
    \end{split}
\end{equation}
We can write 
\begin{align*}
& \Big(\frac{X_{1}^{k \top}(\HSigmakx)_{-(1, 2)}^{-1} X^k_1 + X_{2}^{k \top}(\HSigmakx)_{-(1, 2)}^{-1} X^k_2}{n}\Big)^{j} \\
& = \sum\limits_{j_1=0}^j \binom{j}{j_{1}} \Big( \frac{X_{1}^{k \top}(\HSigmakx)_{-(1, 2)}^{-1} X^k_1}{n} \Big)^{j_1} \Big( \frac{X_{2}^{k \top}(\HSigmakx)_{-(1, 2)}^{-1} X^k_2}{n} \Big)^{j - j_1}.
\end{align*}
We let
\[g_{j_1} (X_1) = \Big( \frac{X_{1}^{k \top}(\HSigmakx)_{-(1, 2)}^{-1} X^k_1}{n} \Big)^{j_1},\]
and 
\[f_{j_2} (X_2) = \frac{(X_{2}^{k \top}(\HSigmakx)_{-(1,2)}^{-1} X^k_2 / n)^{j_2}}{1 + X_{2}^{k \top}(\HSigmakx)_{-(1,2)}^{-1} X^k_2 / n}.\]
Then we can control the main term of \eqref{key decomposition} by applying Lemma~\ref{lem: sample covariance1_plus} as follows
\begin{align*}
    \eqref{part 2} 
    &\lesssim \sum\limits_{j=1}^m (-1)^j \sum\limits_{j_1=0}^j \binom{j}{j_1} \Ek\left[ g_{j_1 + 1}(X_1)X_{1}^{k\top} (\HSigmakx)_{-(1,2)}^{-1}X_{2}^{k} f_{j - j_1}(X_2)\right].
\end{align*}
For each $j_1$, by Lemma \ref{lem: sample covariance1}, we have 
\begin{equation}
\label{key projection}
\begin{split}
&\Ek\left[ g_{j_1 + 1}(X_1)X_{1}^{k\top} (\HSigmakx)_{-(1,2)}^{-1}X_{2}^{k} f_{j - j_1}(X_2)\right] \\
&\lesssim \Vert g_{j_1+1} \Vert_{2} \Vert f_{j-j_1} \Vert_{2} \\
&= \frac{1}{n^{j+1}}\Big\Vert (X_{1}^{k \top}(\HSigmakx)_{-(1, 2)}^{-1} X^k_1)^{j_1+1} \Big\Vert_{2} \Bigg\Vert \frac{(X_{2}^{k \top}(\HSigmakx)_{-(1, 2)}^{-1} X^k_2)^{j-j_1}}{1+(X_{2}^{k \top}(\HSigmakx)_{-(1, 2)}^{-1} X^k_2)/n} \Bigg\Vert_{2}\\
&\leq \frac{1}{n^{j+1}}\Big\Vert (X_{1}^{k \top}(\HSigmakx)_{-(1, 2)}^{-1} X^k_1)^{j_1+1} \Big\Vert_{2} \Big\Vert (X_{2}^{k \top}(\HSigmakx)_{-(1, 2)}^{-1} X^k_2)^{j-j_1} \Big\Vert_{2} \\
& \lesssim \Big(\frac{p}{n}\Big)^{j+1}.
\end{split}
\end{equation}
Then we have 
\begin{equation}
\label{key expansion}
\begin{split}
\eqref{part 2} & \lesssim \sum\limits_{j=1}^m (-1)^j \sum\limits_{j_1=0}^j \binom{j}{j_1} \left( \frac{p}{n} \right)^{j + 1} \\
& = \sum\limits_{j=1}^m \left( \frac{p}{n} \right)^{j + 1}(-1)^j 2^{j} \\
& = \frac{p}{n}\sum\limits_{j=1}^m \left( -\frac{2p}{n} \right)^{j}\\
& =\frac{p}{n} \frac{1-(-2p/n)^m}{1+2p/n}\\
&= o(1).
\end{split}
\end{equation}
As $p=o(n)$, we can choose $p$ such that $2p/n<1$, and take $m=\lceil \log n \rceil$ to obtain our result.

Combining the above results with Assumption~\ref{ap2} assuming that $E[Y_i(1)|X_i]$ is bounded by constant $M$, we can conclude that:
\begin{equation*}
    \Ek\Big\{X_{1}^{k \top} (\HSigmakx^{-1} - \Sigmakx^{-1}) X^k_2Y^k_2(1)\Big\}=o(1),
\end{equation*}
and 
\begin{align}
    &\Eka\{\Res(X^k,Y^k(1))\} \notag\\
    &=  \frac{1}{\sqrt n(n_{[k]}-1)} \sum\limits_{N(k)+1\leq i\neq j\leq N(k+1)} \frac{1}{\pi^2_{n[k]}}(A_i -\pi_{n[k]})A_j \ \Ek\Big\{X_{1}^{k \top} (\HSigmakx^{-1} - \Sigmakx^{-1}) X^k_2Y^k_2(1)\Big\}\notag\\
    &= O(1)\cdot o(1) = o(1). \label{eq:cinp_expectation}
\end{align}
Next, we will show that for any $k, a$,
\[\var\{\Res(X^k,Y^k(1))|B_i=k,A_i=a\} = o (1).\]
In particular, by elementary algebra, we rewrite $\Eka\{\Res(X^k,Y^k(1))\}$ as follows:
{\small
\begin{align*}
&\Eka\{\Res^2(X^k,Y^k(1))\}\\
&= \frac{1}{n(n_{[k]}-1)^2}\frac{1}{\pi^4_{n[k]}}\Eka\Bigg( \sum_{\substack{1\leq i\neq j\leq n \\ i,j \in [k]}} (A_i -\pi_{n[k]})X_{i}^{k \top} (\HSigmakx^{-1} - \Sigmakx^{-1}) X^k_j A_jY^k_j\Bigg)^2\\
&=\frac{1}{n(n_{[k]}-1)^2}\frac{1}{\pi^4_{n[k]}}\Bigg[\sum_{\substack{1\leq i\neq j\leq n \\ i,j \in [k]}} (A_i -\pi_{n[k]})^2A_j\Eka\{X_{i}^{k \top} (\HSigmakx^{-1} - \Sigmakx^{-1}) X^k_j Y^k_j\}^2 \\
&\  + \sum_{\substack{1\leq i\neq j\leq n \\ i,j \in [k]}} (A_i -\pi_{n[k]})(A_j -\pi_{n[k]})A_iA_j\Eka\{X_{i}^{k \top} (\HSigmakx^{-1} - \Sigmakx^{-1}) X^k_j\}^2Y^k_i Y^k_j\\
&\  + \sum_{\substack{1\leq i\neq j\neq l\leq n \\ i,j,l \in [k]}} (A_i -\pi_{n[k]})^2A_jA_l\Eka\{X_{i}^{k \top} (\HSigmakx^{-1} - \Sigmakx^{-1}) X^k_j\}\{X_{i}^{k \top} (\HSigmakx^{-1} - \Sigmakx^{-1}) X^k_l\}Y^k_j Y^k_l \\
&\  + \sum_{\substack{1\leq i\neq j\neq l\leq n \\ i,j,l \in [k]}} (A_i -\pi_{n[k]})(A_j -\pi_{n[k]})A_jA_l\Eka\{X_{i}^{k \top} (\HSigmakx^{-1} - \Sigmakx^{-1}) X^k_j\}\{X_j^{k\top} (\HSigmakx^{-1} - \Sigmakx^{-1}) X^k_l\}Y^k_j Y^k_l\\
&\  + \sum_{\substack{1\leq i\neq j\neq l\leq n \\ i,j,l \in [k]}} (A_i -\pi_{n[k]})(A_l -\pi_{n[k]})A_jA_i\Eka\{X_{i}^{k \top} (\HSigmakx^{-1} - \Sigmakx^{-1}) X^k_j\}\{(X^k_l)^{\top} (\HSigmakx^{-1} - \Sigmakx^{-1}) X^k_i\}Y^k_j Y^k_i\\
&\  +  \sum_{\substack{1\leq i\neq j\neq l\leq n \\ i,j,l \in [k]}} (A_i -\pi_{n[k]})(A_l -\pi_{n[k]})A_j^2\Eka\{X_{i}^{k \top} (\HSigmakx^{-1} - \Sigmakx^{-1}) X^k_j\}\{(X^k_l)^{\top} (\HSigmakx^{-1} - \Sigmakx^{-1}) X^k_j\}(Y^k_j)^2\\
&\  + \sum_{\substack{1\leq i\neq j\neq l\neq m\leq n \\ i,j,l,m \in [k]}} (A_i -\pi_{n[k]})(A_l -\pi_{n[k]})A_iA_m\Eka\{X_{i}^{k \top} (\HSigmakx^{-1} - \Sigmakx^{-1}) X^k_j\}\{(X^k_l)^{\top} (\HSigmakx^{-1} - \Sigmakx^{-1}) X^k_m\}Y^k_j Y^k_m\Bigg]\\
&= \frac{1}{n(n_{[k]}-1)^2}\frac{1}{\pi^4_{n[k]}}\Bigg[\sum_{\substack{1\leq i\neq j\leq n \\ i,j \in [k]}} (A_i -\pi_{n[k]})^2A_j\Ek\{X_{1}^{k \top} (\HSigmakx^{-1} - \Sigmakx^{-1}) X^k_2 Y^k_2\}^2 \\
&\  + \sum_{\substack{1\leq i\neq j\leq n \\ i,j \in [k]}} (A_i -\pi_{n[k]})(A_j -\pi_{n[k]})A_iA_j\Ek\{X_{1}^{k \top} (\HSigmakx^{-1} - \Sigmakx^{-1}) X^k_2\}^2Y^k_1 Y^k_2\\
&\  + \sum_{\substack{1\leq i\neq j\neq l\leq n \\ i,j,l \in [k]}} (A_i -\pi_{n[k]})^2A_jA_l\Ek\{X_{1}^{k \top} (\HSigmakx^{-1} - \Sigmakx^{-1}) X^k_2\}\{X_{1}^{k \top} (\HSigmakx^{-1} - \Sigmakx^{-1}) X^k_3\}Y^k_2 Y^k_3 \\
&\  + \sum_{\substack{1\leq i\neq j\neq l\leq n \\ i,j,l \in [k]}} (A_i -\pi_{n[k]})(A_j -\pi_{n[k]})A_jA_l\Ek\{X_{1}^{k \top} (\HSigmakx^{-1} - \Sigmakx^{-1}) X^k_2\}\{X_{2}^{k \top} (\HSigmakx^{-1} - \Sigmakx^{-1}) X^k_3\}Y^k_2 Y^k_3\\
&\  + \sum_{\substack{1\leq i\neq j\neq l\leq n \\ i,j,l \in [k]}} (A_i -\pi_{n[k]})(A_l -\pi_{n[k]})A_jA_i\Ek\{X_{1}^{k \top} (\HSigmakx^{-1} - \Sigmakx^{-1}) X^k_2\}\{X^{k\top}_3 (\HSigmakx^{-1} - \Sigmakx^{-1}) X^k_1\}Y^k_2 Y^k_1\\
&\  +  \sum_{\substack{1\leq i\neq j\neq l\leq n \\ i,j,l \in [k]}} (A_i -\pi_{n[k]})(A_l -\pi_{n[k]})A_j^2\Ek\{X_{1}^{k \top} (\HSigmakx^{-1} - \Sigmakx^{-1}) X^k_2\}\{X^{k\top}_3 (\HSigmakx^{-1} - \Sigmakx^{-1}) X^k_2\}(Y^k_2)^2\\
&\  + \sum_{\substack{1\leq i\neq j\neq l\neq m\leq n \\ i,j,l,m \in [k]}} (A_i -\pi_{n[k]})(A_l -\pi_{n[k]})A_lA_m\Ek\{X_{1}^{k \top} (\HSigmakx^{-1} - \Sigmakx^{-1}) X^k_2\}\{X^{k\top}_3 (\HSigmakx^{-1} - \Sigmakx^{-1}) X^k_4\}Y^k_2 Y^k_4\Bigg].
\end{align*}} 

\noindent By Assumption \ref{ap1}, the conditional mean of $Y_i(a)|X_i$ is bounded up to the fourth moment. From the last equality of the above display, except for the prefactors involving $A_i$'s, by \Holder's inequality, we need to analyze the following terms: 
\begin{align*}
    &\Ek\{X_{1}^{k \top} (\HSigmakx^{-1} - \Sigmakx^{-1}) X^k_2\}^2,\\
    &\Ek\{X_{1}^{k \top} (\HSigmakx^{-1} - \Sigmakx^{-1}) X^k_2\}\{X_{1}^{k \top} (\HSigmakx^{-1} - \Sigmakx^{-1}) X^k_3\}, \text{ and } \\
    &\Ek\{X_{1}^{k \top} (\HSigmakx^{-1} - \Sigmakx^{-1}) X^k_2\}\{X^{k\top}_3 (\HSigmakx^{-1} - \Sigmakx^{-1}) X^k_4\}.
\end{align*}
By Lemmas \ref{lem:leave-out-matrix} and \ref{lem:repeated Sherman-Morrison} in Appendix~\ref{app:matrices}, we have 
\begin{align}
    &\Ek\{X_{1}^{k \top} (\HSigmakx^{-1} - \Sigmakx^{-1}) X^k_2\}^2=o(n), \label{eq:matrix1}\\
    &\Ek\{X_{1}^{k \top} (\HSigmakx^{-1} - \Sigmakx^{-1}) X^k_2\}\{X_{1}^{k \top} (\HSigmakx^{-1} - \Sigmakx^{-1}) X^k_3\}=o(1),\label{eq:matrix2}\\
    &\Ek\{X_{1}^{k \top} (\HSigmakx^{-1} - \Sigmakx^{-1}) X^k_2\}\{X^{k\top}_3 (\HSigmakx^{-1} - \Sigmakx^{-1}) X^k_4\}=o(1)\label{eq:matrix3}.
\end{align}
We then analyze the prefactors involving $A_i$'s. In fact, these prefactors have at most the same rate as the total number of permutations, that is, 
\begin{equation*}
\left.
\begin{aligned}
  &\sum_{\substack{1\leq i\neq j\leq n \\ i,j \in [k]}} (A_i -\pi_{n[k]})^2A_j\\
  &\sum_{\substack{1\leq i\neq j\leq n \\ i,j \in [k]}} (A_i -\pi_{n[k]})(A_j -\pi_{n[k]})A_iA_j
\end{aligned}
\right\} = O(n^2),
\end{equation*}
and
\begin{equation*}
\left.
\begin{aligned}
    &\sum_{\substack{1\leq i\neq j\neq l\leq n \\ i,j,l \in [k]}} (A_i -\pi_{n[k]})^2A_jA_l\\ &\sum_{\substack{1\leq i\neq j\neq l\leq n \\ i,j,l \in [k]}} (A_i -\pi_{n[k]})(A_j -\pi_{n[k]})A_jA_l\\
    &\sum_{\substack{1\leq i\neq j\neq l\leq n \\ i,j,l \in [k]}} (A_i -\pi_{n[k]})(A_l -\pi_{n[k]})A_jA_i\\
&\sum_{\substack{1\leq i\neq j\neq l\leq n \\ i,j,l \in [k]}} (A_i -\pi_{n[k]})(A_l -\pi_{n[k]})A_j^2 
\end{aligned}
\right\}= O(n^3).
\end{equation*}

For the last term, we can control it by the following argument: 
\begin{align*}
    &\sum_{\substack{1\leq i\neq j\neq l\neq m\leq n \\ i,j,l,m \in [k]}} (A_i -\pi_{n[k]})(A_l -\pi_{n[k]})A_lA_m\\
    & = -n^3_{[k]a}(1-\pi) + O(n^2_{[k]a}) \\
    & = O(n^3).
\end{align*}
Combining \eqref{eq:matrix1}--\eqref{eq:matrix3} with the above results, we can then obtain the following result:
\begin{equation}
\label{eq:key}
\Eka\{\Res^2(X^k,Y^k(1))\} = o(1).
\end{equation}
The desired result \eqref{desired treatment} for the treatment group follows by Markov's inequality, \eqref{eq:cinp_expectation} and \eqref{eq:key}:
\begin{equation*}
\frac{1}{\sqrt{n}}\sum\limits_{k=1}^K n_{[k]}\mathbf{U}_{n_{[k]},2}(\HSigmak{XX}^{-1} - \Sigmakx^{-1};1) = o_P(1).
\end{equation*}
Finally, the desired result \eqref{desired} holds by symmetry, which concludes the proof.

\subsection{Proof of Theorem~\ref{thm:var}}
\label{app:thm:var}

The consistency of $\hat{\zeta}^2_H$ was established in \citet{Bugni2018}, that is 
\[\hat{\zeta}^2_H \stackrel{P}\to \zeta^2_H.\]
Therefore, we are left to show the following:
\[\hat\zeta^2_{\mathrm{I},r}(\pi_{n[k]})\stackrel{P}\to  \zeta^2_{\mathrm{I},r}(\pi_{[k]}) \text{ and } \hat\zeta^2_{\mathrm{II}} \stackrel{P}\to  \zeta^2_{\mathrm{II}}.\]
Recall that we have 
\[\hat\zeta^2_{\mathrm{I},r}(\pi_{[k]}) = \hat\sigma^2_Y(\pi_{[k]}) - \hat\sigma^2_{\mathrm{I},\eta}(\pi_{[k]}) - 2 \hat\sigma_{\mathrm{I},\eta(1)\eta(0)},\]
and
\[\hat\zeta^2_{\mathrm{II}} =  \HzetaSec{Y(1)} +  \HzetaSec{Y(0)}-2\HzetaSec{Y(1,0)}.\]
 
Similar to the proof of Theorem \ref{thm1}, we construct $(Y^k_i(0),Y^k_i(1),X^k_i)$ for $i \in \{1, \cdots, n\}$, which are i.i.d.\ with marginal distribution equal to the distribution of $(Y_i(0),Y_i(1),X_i)|B_i=k$ independently for each stratum $k$ and independently of $(A^{(n)},B^{(n)})$. Note that for all $k \in \mathcal{K}$, $(Y^k_i(0),Y^k_i(1),X^k_i)$ are mutually independent due to coupling and a bounded number of strata; therefore, we only need to show that in each stratum $k$, our variance estimator converges to the asymptotic variance in probability. In the proof, we only show the results for the treatment group, and the results for the control group follow by symmetry.

In $\hat\zeta^2_{\mathrm{I},r}(\pi_{[k]})$, the terms that involve treatment groups are 
\begin{equation}\label{varterm1}
    \sum\limits_{k=1}^K p_{n[k]}\Big[\frac{1}{\pi_{n[k]}} \frac{1}{n_{[k]1}}\sum\limits_{i\in[k]} A_i(Y_i-\bar Y_{[k]1})^2\Big]
\end{equation} and
\begin{equation}\label{varterm2}
    \sum\limits_{k=1}^K p_{n[k]}\Big[\frac{1-\pi_{n[k]}}{\pi_{n[k]}}\frac{1}{n_{[k]1}(n_{[k]1}-1)} \sum_{\substack{1\leq i\neq j\leq n\\i,j\in [k]}} A_iA_jY_iX^{\top}_i\HSigmakx^{-1}X_jY_j\Big].
\end{equation}
Specifically, \eqref{varterm1} is used to estimate
\[\sum\limits_{k=1}^K\frac{p_{[k]}}{\pi_{[k]}} \sigma^2_{[k]Y(1)}\]
and \eqref{varterm2} is used to estimate
\[\sum\limits_{k=1}^K\frac{p_{[k]}(1-\pi_{[k]})}{\pi_{[k]}} \eta^{\top}_{[k]}(1)\Sigmakx^{-1}\eta_{[k]}(1).\]
As shown in \citet{Bugni2018}, 
\[\sum\limits_{k=1}^K p_{n[k]}\Big[\frac{1}{\pi_{n[k]}} \frac{1}{n_{[k]1}}\sum\limits_{i\in[k]} A_i(Y_i-\bar Y_{[k]1})^2\Big]\stackrel{P}\to \sum\limits_{k=1}^K p_{[k]}\frac{\sigma^2_{[k]Y(1)}}{\pi_{[k]}}.\]
In addition, by Assumption~\ref{ap1}, $\pi_{n[k]} \stackrel{P}\to \pi_{[k]}$, and thus we are left to show that, in each stratum $k$, we have 
\begin{equation}\label{var_target1}
    \frac{1}{n_{[k]}(n_{[k]}-1)} \sum\limits_{\substack{1\leq i \neq j\leq n \\ i,j\in[k]}} A_i A_j Y_i X_i^{\top} \HSigmakx^{-1} X_j Y_j - \eta^{\top}_{[k]}{(1)} \Sigmakx^{-1}\eta^{\top}_{[k]}{(1)} = o_P(1).
\end{equation}
In $\hat\zeta^2_{\mathrm{II}}$, the term involves treatment group is 
\[\HzetaSec{Y(1)} = \sum\limits_{k=1}^K \frac{p_{n[k]}}{n_{[k]}-1}\Big\{\frac{1-\pi_{[n]k}}{\pi^2_{n[k]}}\hat\sigma^2_{\mathrm{II},Y(1)[1]} + \frac{(1-\pi_{[n]k})^2}{\pi^2_{n[k]}}\hat\sigma^2_{\mathrm{II},Y(1)[1,2]}\Big\}\]
and
\[\hat\sigma^2_{\mathrm{II},Y(1)[1]} =  \frac{1}{n_{[k]1}} \sum\limits_{i\in [k]} A_i X_i^{\top} \HSigmakx^{-1}X_i Y_i^2,\]
\[\hat\sigma^2_{\mathrm{II},Y(1)[1,2]} = \frac{1}{n_{[k]1}(n_{[k]1}-1)} \sum\limits_{\substack{1\leq i \neq j\leq n \\ i,j\in[k]}} A_iA_j (X_i^{\top} \HSigmakx^{-1}X_j)^2 Y_iY_j.\]
They are constructed to estimate, respectively, $\sigma^2_{\mathrm{II},Y(1)[1]} = \Ek\{Y_1(1)X_1^{\top}\Sigmakx^{-1}X_1^kY_1^k(1)\}$ and $\sigma^2_{\mathrm{II},Y(1)[1,2]} = \Ek\{(X_1^{\top}\Sigmakx^{-1}X_2)^2Y_1(1)Y_2(1)\}$.

We will show that our variance estimator in each stratum $k$ converges in probability in the following two steps: (1) We show that, when we replace $\HSigmakx$ by $\Sigmakx$ in the variance estimator, this oracle variance estimator converges to the asymptotic variance in probability; (2) The difference between our proposed variance estimator and the oracle variance estimator is $o_P (1)$. 

In the first step, we will show that 
\begin{align*}
    &\frac{1}{n_{[k]}(n_{[k]}-1)} \sum\limits_{\substack{1\leq i \neq j\leq n \\ i,j\in[k]}} A_i A_j Y_i X_i^{\top} \Sigmakx^{-1} X_j Y_j - \eta^{\top}_{[k]}{(1)} \Sigmakx^{-1}\eta^{\top}_{[k]}{(1)} = o_P(1),\\
    &\frac{1}{n_{[k]}-1}\Big\{ \frac{1}{n_{[k]}} \sum\limits_{i\in [k]} A_i X_i^{\top} \Sigmakx^{-1}X_i Y_i^2\Big\}-\frac{1}{n_{[k]}-1} E\{Y_1(1)X_1^{\top}\Sigmakx^{-1}X_1^kY_1^k(1)\}=o_P(1),\\
    &\frac{1}{n_{[k]}-1}\Big\{ \frac{1}{n_{[k]}(n_{[k]}-1)} \sum\limits_{\substack{1\leq i \neq j\leq n \\ i,j\in[k]}} A_i A_j (X_i^{\top} \Sigmakx^{-1}X_j)^2 Y_iY_j\Big\} \\
    & - \frac{1}{n_{[k]}-1} E\{(X_1^{\top}\Sigmakx^{-1}X_2)^2Y_1(1)Y_2(1)\}=o_P(1).
\end{align*}
As they involve the true $\Sigmakx$, they are unbiased. To show that they are consistent estimators, we only need to show that the variance of each term converges to zero. First, conditional on $B=k$ and $A=1$, we have
\begin{align*}
    &\var\Big\{\frac{1}{n_{[k]}(n_{[k]}-1)} \sum\limits_{\substack{1\leq i \neq j\leq n \\ i,j\in[k]}} A_i A_j Y_i X_i^{\top} \Sigmakx^{-1} X_j Y_j\Big|B_i=k, A_i=1\} \\
    &= \frac{2}{n^2_{[k]}} \var \Big\{Y_1(1) X_1^{\top} \Sigmakx^{-1} X_2 Y_2(1)|B_i=k, A_i=1\Big\} \\
    &\quad + \frac{4}{n_{[k]}} \mathrm{Cov}\Big\{Y_1(1) X_1^{\top} \Sigmakx^{-1} X_2 Y_2(1),Y_1(1) X_1^{\top} \Sigmakx^{-1} X_3 Y_3(1) \Big|B_i=k, A_i=1\}\\
    &\quad + \mathrm{Cov}\Big\{Y_1(1) X_1^{\top} \Sigmakx^{-1} X_2 Y_2(1),Y_3(1) X_3^{\top} \Sigmakx^{-1} X_4 Y_4(1) \Big|B_i=k, A_i=1\} \\ 
    & \leq \frac{2}{n^2_{[k]}} \Ek \Big\{Y_1 X_1^{\top} \Sigmakx^{-1} X_2 Y_2\Big\}^2 + \frac{4}{n_{[k]}} \ProjBeta{1}^{\top} \Ek\{X_1X_1^{\top}Y^2_1(1)\}\ProjBeta{1} \\
    & = O \left( \frac{p}{n^2} \right) + O \left( \frac{1}{n} \right) = o(1).
\end{align*}
Next, we bound the conditional variance of $\dfrac{1}{n_{[k]}(n_{[k]}-1)} \sum\limits_{i\in [k]} A_i X_i^{\top} \Sigmakx^{-1}X_i Y_i^2$ as follows:
\begin{align*}
    &\var\Big\{\frac{1}{n_{[k]}(n_{[k]}-1)} \sum\limits_{i\in [k]} A_i X_i^{\top} \Sigmakx^{-1}X_i Y_i^2|B_i=k, A_i=1\Big\} \\
    &\leq \frac{1}{n_{[k]}(n_{[k]}-1)^2}\Ek\Big\{(X_1^{\top} \Sigmakx^{-1}X_1)^2 Y_1^4(1)\Big\} \\
    &= O \left( \frac{p^2}{n^3} \right) = o(1).
\end{align*}
We then turn to $\dfrac{1}{n_{[k]}-1}\Big\{ \dfrac{1}{n_{[k]}(n_{[k]}-1)} \sum\limits_{\substack{1\leq i \neq j\leq n \\ i,j\in[k]}} A_i A_j (X_i^{\top} \Sigmakx^{-1}X_j)^2 Y_iY_j\Big\}$.
Similarly, we have 
\begin{align*}
    &\var\Big\{ \frac{1}{n_{[k]}(n_{[k]}-1)^2} \sum\limits_{\substack{1\leq i \neq j\leq n \\ i,j\in[k]}} A_i A_j (X_i^{\top} \Sigmakx^{-1}X_j)^2 Y_iY_j|B_i=k, A_i=1\Big\} \\
    &= \frac{1}{n_{[k]}^4} \var\Big\{(X_1^{\top} \Sigmakx^{-1}X_2)^2 Y_1(1)Y_2(1)|B_i=k, A_i=1\Big\} \\
    &\quad + \frac{1}{n_{[k]}^3} \mathrm{Cov}\Big\{(X_1^{\top} \Sigmakx^{-1}X_2)^2 Y_1(1)Y_2(1),(X_1^{\top} \Sigmakx^{-1}X_3)^2 Y_1(1)Y_3(1)|B_i=k, A_i=1\Big\}\\
    &\leq  \frac{1}{n_{[k]}^4} \Ek\{(X_1^{\top} \Sigmakx^{-1}X_2)^4 Y^2_1Y^2_2\} + \frac{1}{n_{[k]}^3}\Ek\Big\{(X_1^{\top} \Sigmakx^{-1}X_2)^2 Y_1(1)Y_2(1)(X_1^{\top} \Sigmakx^{-1}X_3)^2 Y_1(1)Y_3(1)\Big\}.
\end{align*}
For the second term in the last line of the above display, by Assumption \ref{ap1}, $E[Y_i(1)|X_i]$ and $E[Y^2_i(1)|X_i]$ are bounded, then 
\[\Ek\Big\{(X_1^{\top} \Sigmakx^{-1}X_2)^2 Y_1(1)Y_2(1)(X_1^{\top} \Sigmakx^{-1}X_3)^2 Y_1(1)Y_3(1)\Big\} = O(\Ek[(X_1^{\top} \Sigmakx^{-1}X_2)^2(X_1^{\top} \Sigmakx^{-1}X_3)^2]).\]
As $X_1$, $X_2$ and $X_3$ are mutually independent, we have
\begin{equation*}
    \Ek[(X_1^{\top} \Sigmakx^{-1}X_2)^2(X_1^{\top} \Sigmakx^{-1}X_3)^2] = \Ek [(X_1^{\top}\Sigmakx^{-1}X_1)^2] = O(p^2).   
\end{equation*}
For the first term, by Assumption \ref{ap1}, we need to show
\[\frac{1}{n_{[k]}^4} \Ek\{(X_1^{\top} \Sigmakx^{-1}X_2)^4\} = o(1).\]
To this end, we have the following.
\begin{align*}
    \frac{1}{n_{[k]}^4} \Ek\{(X_1^{\top} \Sigmakx^{-1}X_2)^4\}&=\frac{1}{n_{[k]}^4}\Ek\{(X_1^{\top} \Sigmakx^{-1}X_2)^2(X_1^{\top} \Sigmakx^{-1}X_2)^2\} \\
    &\leq \frac{1}{n_{[k]}^4} \Ek\{\sup_{X_1,X_2}(X_1^{\top} \Sigmakx^{-1}X_2)^2\} \Ek \{(X_1^{\top} \Sigmakx^{-1}X_2)^2\}\\
    &\leq \frac{M^4}{n_{[k]}^4} \Ek\{(1^{\top} \Sigmakx^{-1}1)^2\} \Ek \{X_1^{\top} \Sigmakx^{-1}X_1\}\\
    & = \frac{M^4}{n^4_{[k]}}\cdot O(p^2) \cdot O(p) = o(1).
\end{align*}
Therefore, we have shown that the variance of each oracle variance estimator is $o(1)$. By Chebyshev's inequality, we derive that
\begin{align*}
    &\frac{1}{n_{[k]}(n_{[k]}-1)} \sum\limits_{\substack{1\leq i \neq j\leq n \\ i,j\in[k]}} A_i A_j Y_i X_i^{\top} \Sigmakx^{-1} X_j Y_j - \ProjBeta{1}^{\top} \Sigmakx\ProjBeta{1} = o_P(1),\\
    &\frac{1}{n_{[k]}-1}\Big\{ \frac{1}{n_{[k]}} \sum\limits_{i\in [k]} A_i X_i^{\top} \Sigmakx^{-1}X_i Y_i^2\Big\}-\frac{1}{n_{[k]}-1} E\{Y_1(1)X_1^{\top}\Sigmakx^{-1}X_1^kY_1^k(1)\}=o_P(1),\\
    &\frac{1}{n_{[k]}-1}\Big\{ \frac{1}{n_{[k]}(n_{[k]}-1)} \sum\limits_{\substack{1\leq i \neq j\leq n \\ i,j\in[k]}} A_i A_j (X_i^{\top} \Sigmakx^{-1}X_j)^2 Y_iY_j\Big\} \\
    & - \frac{1}{n_{[k]}-1} E\{(X_1^{\top}\Sigmakx^{-1}X_2)^2Y_1(1)Y_2(1)\}=o_P(1).
\end{align*}

In the second step, we need to show that
\begin{align}
    &\frac{1}{n_{[k]}(n_{[k]}-1)} \sum\limits_{\substack{1\leq i \neq j\leq n \\ i,j\in[k]}} A_i A_j Y_i X_i^{\top} (\HSigmakx^{-1}-\Sigmakx^{-1}) X_j Y_j  = o_P(1), \label{eq:variance1}\\
    &\frac{1}{n_{[k]}-1}\Big\{ \frac{1}{n_{[k]}} \sum\limits_{i\in [k]} A_i X_i^{\top} (\HSigmakx^{-1}-\Sigmakx^{-1})X_i Y_i^2\Big\}=o_P(1),\label{eq:variance2}\\
    &\frac{1}{n_{[k]}-1}\Big\{ \frac{1}{n_{[k]}(n_{[k]}-1)} \sum\limits_{\substack{1\leq i \neq j\leq n \\ i,j\in[k]}} A_i Y_i [(X_i^{\top}  \HSigmakx^{-1} X_j)^2 - (X_i^{\top}  \Sigmakx^{-1} X_j)^2] A_j Y_j\Big\}=o_P(1)\label{eq:variance3}.
\end{align}
Equation \eqref{eq:variance1} can be shown with a strategy similar to the proof of Theorem \ref{thm:main}, and therefore we omit it in this proof. Next, we analyze $\dfrac{1}{n_{[k]}(n_{[k]}-1)} \sum\limits_{i\in [k]} A_i X_i^{\top} \HSigmakx^{-1}X_i Y_i^2$. By Lemma~\ref{lem:variance}, we have
\[\Ek\{A_i X_i^{\top} (\HSigmakx^{-1} - \Sigmakx^{-1})X_i Y_i^2\} = O(p \sqrt{p / n}),\]
so
\[\Ek\Big\{\frac{1}{n_{[k]}(n_{[k]}-1)} \sum\limits_{i\in [k]} A_i X_i^{\top} (\HSigmakx^{-1} -\Sigmakx^{-1}) X_i Y_i^2\Big\} = O\Big(\frac{p}{n}\cdot \sqrt{\frac{p}{n}}\Big) = o(1).\]
Then we compute the variance of \eqref{eq:variance2}. We have 
\begin{align*}
    &\var\Big\{\frac{1}{n_{[k]}(n_{[k]}-1)} \sum\limits_{i\in [k]} A_i X_i^{\top} (\HSigmakx^{-1} -\Sigmakx^{-1}) X_i Y_i^2|B_i=k,A_i=1\Big\} \\
    &\leq \Eka\Big\{\frac{1}{n_{[k]}(n_{[k]}-1)} \sum\limits_{i\in [k]} A_i X_i^{\top} (\HSigmakx^{-1} -\Sigmakx^{-1}) X_i Y_i^2\Big\}^2\\
    &= \frac{1}{n_{[k]}(n_{[k]}-1)^2} \Ek\big[\{X_1^{\top} (\HSigmakx^{-1} -\Sigmakx^{-1}) X_1\}^2 Y_1^4\big] \\
    &\quad + \frac{1}{n_{[k]}(n_{[k]}-1)} \Ek\big[\{X_1^{\top} (\HSigmakx^{-1} -\Sigmakx^{-1}) X_1\}\{X_2^{\top} (\HSigmakx^{-1} -\Sigmakx^{-1}) X_2\} Y_1^2Y_2^2\big] \\
    & = O \left( \frac{p^4}{n^4} \right) + \frac{1}{n_{[k]}(n_{[k]}-1)}\cdot O\{(\Ek\big[\{X_1^{\top} ((\HSigmakx)_{-1,2}^{-1} -\Sigmakx^{-1}) X_1\}])^2\} \\
    & = O \left( \frac{p^4}{n^4} \right) = o (1).
\end{align*}

We then proceed to analyze the following component in \eqref{eq:variance3}.
\[\frac{1}{n_{[k]}-1}\Big\{ \frac{1}{n_{[k]}(n_{[k]}-1)} \sum\limits_{\substack{1\leq i \neq j\leq n \\ i,j\in[k]}} A_i Y_i (X_i^{\top} \HSigmakx^{-1}X_j)^2 A_j Y_j\Big\}.\]

We first analyze the difference between this term and its oracle version (replacing $\HSigmakx$ by $\Sigmakx$) as 
\[\frac{1}{n_{[k]}-1}\Big[ \frac{1}{n_{[k]}(n_{[k]}-1)} \sum\limits_{\substack{1\leq i \neq j\leq n \\ i,j\in[k]}} A_i Y_i \{X_i^{\top}(\HSigmakx^{-1} - \Sigmakx^{-1})X_j\}\{X_i^{\top}(\HSigmakx^{-1}+\Sigmakx^{-1})X_j\} A_j Y_j\Big].\]
By Lemma~\ref{lem:repeated Sherman-Morrison}, the expectation of the difference can be bounded as follows:
\begin{align*}
    & \frac{1}{n_{[k]}-1} E [A_1 Y_1 \{X_1^{\top}((\HSigmakx)^{-1}_{-1,2} - \Sigmakx^{-1})X_2\} \{X_1^{\top}((\HSigmakx)^{-1}_{-1,2}+\Sigmakx^{-1})X_2\} A_2 Y_2] \\
    &=\frac{1}{n_{[k]}-1} E [A_1 Y_1 \{X_1^{\top}((\HSigmakx)^{-1}_{-1,2} - \Sigmakx^{-1}) \Ek [A_2 Y_2 X_2 X_2^\top] ((\HSigmakx)^{-1}_{-1,2}+\Sigmakx^{-1})X_1\}] \\
    & = o (n^{-1}).
\end{align*}
Then we analyze its variance. Similar to the proof of Theorem~\ref{thm:main}, we have
\begin{align*}
    &\var\Big\{ \frac{1}{n_{[k]}(n_{[k]}-1)^2} \sum\limits_{\substack{1\leq i \neq j\leq n \\ i,j\in[k]}} A_i A_j [(X_i^{\top} \HSigmakx^{-1}X_j)^2-(X_i^{\top} \Sigmakx^{-1}X_j)^2 ] Y_iY_j|B_i=k, A_i=1\Big\}\\
    &\leq \Eka\Big[\Big\{ \frac{1}{n_{[k]}(n_{[k]}-1)^2} \sum\limits_{\substack{1\leq i \neq j\leq n \\ i,j\in[k]}} A_i A_j [(X_i^{\top} \HSigmakx^{-1}X_j)^2-(X_i^{\top} \Sigmakx^{-1}X_j)^2 ] Y_iY_j\Big\}^2\Big] \\
    &= \frac{2}{n^4_{[k]}}\Ek\Big[\{X_1^{\top}(\HSigmakx^{-1} - \Sigmakx^{-1})X_2\}^2\{X_1^{\top}(\HSigmakx^{-1}+\Sigmakx^{-1})X_2\}^2 Y_1(1)Y_2(1)\Big]\\
    &\quad +\frac{4}{n^3_{[k]}}\Ek\Big[\{X_1^{\top}(\HSigmakx^{-1} - \Sigmakx^{-1})X_2\}\{X_1^{\top}(\HSigmakx^{-1}+\Sigmakx^{-1})X_2\} Y_1(1)Y_2(1)\\
    &\quad\quad \times\{X_1^{\top}(\HSigmakx^{-1} - \Sigmakx^{-1})X_3\}\{X_1^{\top}(\HSigmakx^{-1}+\Sigmakx^{-1})X_3\} Y_1(1)Y_3(1)\Big]\\
    &\quad+ \frac{1}{n^2_{[k]}}\Ek\Big[\{X_1^{\top}(\HSigmakx^{-1} - \Sigmakx^{-1})X_2\}\{X_1^{\top}(\HSigmakx^{-1}+\Sigmakx^{-1})X_2\} Y_1(1)Y_2(1)\\
    &\quad\quad \times\{X_3^{\top}(\HSigmakx^{-1} - \Sigmakx^{-1})X_4\}\{X_3^{\top}(\HSigmakx^{-1}+\Sigmakx^{-1})X_4\} Y_3(1)Y_4(1)\Big].
\end{align*}
We focus on the following terms: 
\begin{itemize}
    \item $\dfrac{1}{n^4_{[k]}} \Ek\Big[\big\{X_1^{\top}(\HSigmakx^{-1} - \Sigmakx^{-1})X_2\big\}^2\big\{X_1^{\top}(\HSigmakx^{-1} + \Sigmakx^{-1})X_2\big\}^2 Y_1^2(1)Y_2^2(1)\Big]$;
    \item $\begin{aligned}
        &\frac{1}{n^3_{[k]}}\Ek\Big[\{X_1^{\top}(\HSigmakx^{-1} - \Sigmakx^{-1})X_2\}\{X_1^{\top}(\HSigmakx^{-1}+\Sigmakx^{-1})X_2\} \\
    &\quad\times\{X_1^{\top}(\HSigmakx^{-1} - \Sigmakx^{-1})X_3\}\{X_1^{\top}(\HSigmakx^{-1}+\Sigmakx^{-1})X_3\} Y^2_1(1)Y_2(1)Y_3(1)\Big];
    \end{aligned}
    $
    \item $\begin{aligned}
        &\frac{1}{n^2_{[k]}}\Ek\Big[\{X_1^{\top}(\HSigmakx^{-1} - \Sigmakx^{-1})X_2\}\{X_1^{\top}(\HSigmakx^{-1}+\Sigmakx^{-1})X_2\} \\
    &\quad\times\{X_3^{\top}(\HSigmakx^{-1} - \Sigmakx^{-1})X_4\}\{X_3^{\top}(\HSigmakx^{-1}+\Sigmakx^{-1})X_4\} Y_1(1)Y_2(1)Y_3(1)Y_4(1)\Big].
    \end{aligned}$
\end{itemize}

The first term can be controlled as follows: 
\begin{align*}
    &\frac{1}{n^4_{[k]}} E\Big[\big\{X_1^{\top}(\HSigmakx^{-1} - \Sigmakx^{-1})X_2\big\}^2\big\{X_1^{\top}(\HSigmakx^{-1} + \Sigmakx^{-1})X_2\big\}^2 Y_1^2Y_2^2\Big] \\
    & \leq \frac{1}{n^4_{[k]}} E\big\{X_1^{\top}(\HSigmakx^{-1} - \Sigmakx^{-1})X_2\big\}^2\big\{X_1^{\top}(\HSigmakx^{-1} + \Sigmakx^{-1})X_2\big\}^2\cdot O (1) \\
    &\leq \frac{1}{n^4_{[k]}} E\big\{X_1^{\top}(\HSigmakx^{-1} - \Sigmakx^{-1})X_2\big\}^2 \cdot O(p^2) \\
    &=O \left( \frac{p^4}{n^5} \right) = o \left( \frac{1}{n} \right).
\end{align*}

The second term can be controlled by using Lemma~\ref{lem:repeatSH2} in Appendix~\ref{app:matrices}:
\begin{align*}
    &\frac{1}{n^3_{[k]}}\Ek\Big[\{X_1^{\top}(\HSigmakx^{-1} - \Sigmakx^{-1})X_2\}\{X_1^{\top}(\HSigmakx^{-1}+\Sigmakx^{-1})X_2\} \\
    &\quad\times\{X_1^{\top}(\HSigmakx^{-1} - \Sigmakx^{-1})X_3\}\{X_1^{\top}(\HSigmakx^{-1}+\Sigmakx^{-1})X_3\} Y^2_1(1)Y_2(1)Y_3(1)\Big] \\
    &= \frac{1}{n^3_{[k]}} E\Big[\{X_1^{\top}((\HSigmakx)^{-1}_{-(1,2,3)} - \Sigmakx^{-1})^2X_1\}\Sigmakx^2\{X_1^{\top}((\HSigmakx)^{-1}_{-(1,2,3)}+\Sigmakx^{-1})^2X_1\}\Big]\cdot O(1) +o(1)\\
    & = \frac{1}{n^3_{[k]}} \cdot O \left( \frac{p^4}{n} \right)\cdot \frac{9}{\lambda^2_{\min}} = O \left( \frac{p^3}{n^4} \right) = o (1).
\end{align*}

For the third part, we have the following:
\begin{align*}
 &\frac{1}{n^2_{[k]}}\Ek\Big[\{X_1^{\top}(\HSigmakx^{-1} - \Sigmakx^{-1})X_2\}\{X_1^{\top}(\HSigmakx^{-1}+\Sigmakx^{-1})X_2\} \\
    &\quad\times\{X_3^{\top}(\HSigmakx^{-1} - \Sigmakx^{-1})X_4\}\{X_3^{\top}(\HSigmakx^{-1}+\Sigmakx^{-1})X_4\} Y_1(1)Y_2(1)Y_3(1)Y_4(1)\Big]\\
    &= \frac{1}{n^2_{[k]}} E\Big[\big\{X_1^{\top}((\HSigmakx)_{-(1,2,3,4)}^{-1} - \Sigmakx^{-1})X_2\big\}^2Y_1Y_2\Big]^2 + o (1)\\
    &= \frac{1}{n^2_{[k]}} \cdot O \left( \frac{p^2}{n} \right)\\
    &= O \left( \frac{p^{2}}{n^{3}} \right) = o \left( \frac{1}{n} \right).
\end{align*}

In summary, we have shown that 
\begin{itemize}
    \item $\dfrac{1}{n^4_{[k]}} \Ek\Big[\big\{X_1^{\top}(\HSigmakx^{-1} - \Sigmakx^{-1})X_2\big\}^2\big\{X_1^{\top}(\HSigmakx^{-1} + \Sigmakx^{-1})X_2\big\}^2 Y_1^2(1)Y_2^2(1)\Big]=o(1)$;
    \item $\begin{aligned}
        &\frac{1}{n^3_{[k]}}\Ek\Big[\{X_1^{\top}(\HSigmakx^{-1} - \Sigmakx^{-1})X_2\}\{X_1^{\top}(\HSigmakx^{-1}+\Sigmakx^{-1})X_2\} \\
    &\quad\times\{X_1^{\top}(\HSigmakx^{-1} - \Sigmakx^{-1})X_3\}\{X_1^{\top}(\HSigmakx^{-1}+\Sigmakx^{-1})X_3\} Y^2_1(1)Y_2(1)Y_3(1)\Big]=o(1);
    \end{aligned}
    $
    \item $\begin{aligned}
        &\frac{1}{n^2_{[k]}}\Ek\Big[\{X_1^{\top}(\HSigmakx^{-1} - \Sigmakx^{-1})X_2\}\{X_1^{\top}(\HSigmakx^{-1}+\Sigmakx^{-1})X_2\} \\
    &\quad\times\{X_3^{\top}(\HSigmakx^{-1} - \Sigmakx^{-1})X_4\}\{X_3^{\top}(\HSigmakx^{-1}+\Sigmakx^{-1})X_4\} Y_1(1)Y_2(1)Y_3(1)Y_4(1)\Big]=o(1).
    \end{aligned}$
\end{itemize}
which together lead to
\[\var\Big\{ \frac{1}{n_{[k]}(n_{[k]}-1)^2} \sum\limits_{\substack{1\leq i \neq j\leq n \\ i,j\in[k]}} A_i A_j [(X_i^{\top} \HSigmakx^{-1}X_j)^2-(X_i^{\top} \Sigmakx^{-1}X_j)^2 ] Y_iY_j|B_i=k, A_i=1\Big\}=o(1).\]

Then in each stratum $k$, we have that
\begin{align*}
    &\frac{1}{n_{[k]}(n_{[k]}-1)} \sum\limits_{\substack{1\leq i \neq j\leq n \\ i,j\in[k]}} A_i A_j Y_i X_i^{\top} \HSigmakx^{-1} X_j Y_j - \eta^{\top}_{[k]}(1)\Sigmakx^{-1}\eta_{[k]}(1) = o_P(1),\\
    &\frac{1}{n_{[k]}-1}\Big\{ \frac{1}{n_{[k]}} \sum\limits_{i\in [k]} A_i X_i^{\top} \HSigmakx^{-1}X_i Y_i^2\Big\}-\frac{1}{n_{[k]}-1} E\{Y_1(1)X_1^{\top}\Sigmakx^{-1}X_1^kY_1^k(1)\}=o_P(1),\\
    &\frac{1}{n_{[k]}-1}\Big\{ \frac{1}{n_{[k]}(n_{[k]}-1)} \sum\limits_{\substack{1\leq i \neq j\leq n \\ i,j\in[k]}} A_i A_j (X_i^{\top} \HSigmakx^{-1}X_j)^2 Y_iY_j\Big\}-
    \frac{1}{n_{[k]}-1} E\{(X_1^{\top}\Sigmakx^{-1}X_2)^2Y_1(1)Y_2(1)\}\\
    & =o_P(1).
\end{align*}
Finally, we can follow the same strategy to prove the consistency of the variance estimator in the control group. The consistency of the estimator of the covariance part follows from similar arguments and the Cauchy-Schwarz inequality. We then sum all terms in each stratum $k$ and eventually obtain the desired result.

\section{Preparatory Technical Results}
\label{app:technical}

\subsection{Asymptotic normality of \texorpdfstring{$U$}{}-statistics}
\label{app:U-stats}

Since our newly proposed estimator is based on $U$-statistics, its asymptotic normality relies on the classical results of \citet{bhattacharya1992class} and their extensions.

\begin{lemma}\label{lem:BG}
    Let 
    \[M_2 = \sum\limits_{k \in \mathcal{K}} \frac{1}{\sqrt n (n_{[k]} - 1)}\frac{1}{\pi^2_{[k]}} \sum_{\substack{1 \leq i \neq j \leq n \\ i, j \in [k]}} (A_i -\pi_{n[k]}) A_j X_i^{\top} \Sigmakx^{-1} X_j Y_j(1).\]
    Under Assumptions \ref{ap1}--\ref{ap3}, we have 
    \[ M_2 \stackrel{d}\to \operatorname{N}(0, \zetaSec{Y(1)} + \sigma^2_{\mathrm{I},X(1)}) \]
\end{lemma}

\begin{proof}
We first derive the asymptotic normality of $M_2$ and then calculate its asymptotic variance. Using the same strategy and notation in the proof of Theorem~\ref{thm1}, we construct i.i.d.\ random variables $(Y^k_i(0),Y^k_i(1),X^k_i)$ with marginal distributions equal to the marginal distribution of $(Y_i(0),Y_i(1),X_i)|B_i=k$ independent of $(A^{(n)},B^{(n)})$, independently for each stratum $k$. Note that for any $k \in \mathcal{K}$, $(Y^k_i(0),Y^k_i(1),X^k_i)$ are mutually independent by coupling construction, and therefore we only need to derive the asymptotic normality in each stratum $k$ and then sum them up to obtain the asymptotic normality of $M_2$. Then in stratum $k$, we construct
\[M^k_2 = \frac{1}{\sqrt n(n_{[k]}-1)}\frac{1}{\pi^2_{[k]}} \sum_{\substack{N(k)+1\leq i\neq j\leq N(k+1)}} (A_i -\pi_{n[k]})A_jX^{k\top}_i\Sigmakx^{-1}X^k_j Y^k_j(1),\]
and 
\[M^k_2 \stackrel{d}=M_2|B_i=k,\]
conditional on $(A^{(n)},B^{(n)})$. Then we apply Hoeffding decomposition to decompose $M_2^k$ into two orthogonal terms: a linear summation part and a second-order degenerate $U$-statistic part, in each stratum $k$. The CLT of the linear part can be established by Lindeberg's CLT and the CLT of the degenerate second-order U-statistic part is proved by using \Levy's martingale CLT.

Let $Z^k = (X^k,Y^k(1))$. We first construct a symmetric function $h_{[k]}(Z^k_1,Z^k_2)$ and then use Hoeffding decomposition to analyze the second-order $U$-statistics. Denote
\[U_{[k]}(Z^k_1,Z^k_2) =  X_{1}^{k \top} \Sigmakx^{-1}X^k_2 Y^k_2(1),\]
and 
\[h_{[k]} (Z^k_1,Z^k_2) = U_{[k]}(Z^k_1,Z^k_2) + U_{[k]}(Z^k_1,Z^k_2).\]
Conditional on $(A^{(n)},B^{(n)})$, by Hoeffding decomposition, 
\begin{align*}
M^k_2 = \frac{1}{\sqrt n (n_{[k]}-1)}\sum_{\substack{1\leq i\neq j\leq n \\ i,j\in [k]}} & (A_i-\pi_{[k]})A_j \Big\{ \Ek\{h_{[k]} (Z^k_1, Z^k_2)\} + \{h_{[k]}^{(1)}(Z^k_i) + h^{(1)}(Z^k_j)\} \\
& + \frac{1}{\sqrt n (n_{[k]}-1)}  h_{[k]}^{(2)}(Z^k_i,Z^k_j) \Big\},
\end{align*}
where
\begin{align*}
\Ek(h_{[k]}(Z^k_1,Z^k_2)) &= \Ek[X^k_1]^{\top}\ProjBeta{1}, \\
h_{[k]}^{(1)}(u) &= \Ek\{h_{[k]}(Z^k,u)\} =  u^{\top} \Sigmakx^{-1} \Ek\{X^k Y^k(1)\} =  u^{\top} \ProjBeta{1}, \\
h_{[k]}^{(2)}(Z^k_1,Z^k_2) &= h_{[k]}(Z^k_1,Z^k_2) - h_{[k]}^{(1)}(Z^k_1) - h_{[k]}^{(1)}(Z^k_2) - \Ek(h_{[k]}(Z^k_1,Z^k_2))\\
&= X_1^{k\top}\Sigmakx^{-1}X^k_2 Y^k_2(1) + X_2^{k\top}\Sigmakx^{-1}X^k_1 Y^k_1(1) \\
&\quad - X_1^{k\top}\ProjBeta{1} - X_2^{k\top}\ProjBeta{1} - \Ek[X_1]^{\top}\ProjBeta{1}.
\end{align*}
Let 
\begin{align*}
    R_{1} &= \frac{1}{\sqrt n (n_{[k]}-1)}\sum_{\substack{1\leq i\neq j\leq n \\ i,j\in [k]}} (A_i - \pi_{[k]})A_j (X^k_i + X^k_j)^{\top} \ProjBeta{1}, \text{ and} \\
    R_{2} &= \frac{1}{\sqrt n (n_{[k]}-1)}\sum_{\substack{1\leq i\neq j\leq n \\ i,j\in [k]}}(A_i-\pi_{[k]})A_j \Big[X_{i}^{k \top}\Sigmakx^{-1}X^k_j Y^k_j(1) + X_j^{k\top}\Sigmakx^{-1}X^k_i Y^k_i(1) \\
    &\quad - X_{i}^{k \top}\ProjBeta{1} - X_j^{k\top}\ProjBeta{1}\Big].
\end{align*}
These two parts are uncorrelated. We first analyze $R_{1}$. Let $G_{ij} = (A_i - \pi_{[k]})A_j$ and $\sgij = G_{ij} + G_{ji}$. We can then write $R_1$ as 
\begin{align*}
    R_{1} &= \frac{1}{\sqrt n (n_{[k]}-1)}\frac{1}{\pi_{[k]}^2}\sum_{\substack{1\leq i < j\leq n \\ i,j\in [k]}} \sgij (X^k_i + X^k_j)^{\top} \ProjBeta{1} \\
    &= \frac{1}{\sqrt n (n_{[k]}-1)^2}\frac{1}{\pi_{[k]}^2}\sum\limits_{i=1}^n X_{i}^{k \top}\ProjBeta{1} \sum_{j\neq i}\sgij,\ i,j\in [k] \\
    &= \frac{1}{\sqrt n (n_{[k]}-1)^2}\frac{1}{\pi_{[k]}^2}\sum\limits_{i=1}^n W_i X_{i}^{k \top}\ProjBeta{1},\ i,j\in [k] \\
    &= \frac{1}{\pi_{[k]}^2}\sum\limits_{i=1}^n R_{1,i},
\end{align*}
where
\[R_{1,i} = \frac{1}{\sqrt n(n_{[k]}-1)} W_i X_i^{k\top}\ProjBeta{1}.\]
We can now invoke Lindeberg's CLT by showing that 
\[\sum\limits_{i=1}^n \Ek\{R^2_{1,i}\} = O_P(1),\]
and also invoke Lyapunov's condition by showing that
\[\sum\limits_{i=1}^n \Ek\{R^3_{1,i}\}\stackrel{P}\to 0.\]
To this end, conditional on $B_i=k$, we have
\begin{align*}
\Ek\{R^2_{1,i}\} &= \frac{1}{n(n_{[k]}-1)^2} \Ek\{W^2_i\} \Ek\{X_i^{\top}\ProjBeta{1}\}^2 \\
&= \frac{1}{n(n_{[k]}-1)^2} \ProjBeta{1}^{\top} \Sigmakx \ProjBeta{1} \Ek \Big\{ \Big( \sum\limits_{j=1}^n \sgij - \overline G_{ii} \Big)^2 \Big\}, \\
\Big( \sum\limits_{j=1}^n \sgij - \overline G_{ii} \Big)^2 & = \Big( \sum\limits_{j=1}^n\sgij \Big)^2 - 2\overline G_{ii}\sum\limits_{j=1}^n \sgij + \overline G_{ii}^2 \\
&= \{n_{[k]a}(2A_i - \pi_{[k]a}) - A_i \pi_{[k]a}\}^2+ O(n_{[k]a}) \\
&= A_i \cdot O(n^2_{[k]a}) + O(n^2_{[k]a}),
\end{align*}
and also $\sum\limits_{i=1}^n  \Big( \sum\limits_{j=1}^n \sgij - \overline G_{ii} \Big)^2 = O(n^3)$. Therefore, $\sum\limits_{i=1}^n \Ek\{R^2_{1,i}\} = O(1)$. Similarly, we have
\[R^3_{1,i} = \frac{1}{n^{3/2}(n_{[k]}-1)^3} \Big\{\sum\limits_{j\neq i}\sgij X_{i}^{k \top}\ProjBeta{1} \Big\}^3 = O_P(n^{-3/2}).\]
Then 
\[\sum\limits_{i=2}^n \Ek\{R^3_{1,i}\}=O_P(n^{-1/2}).\]
Therefore, $R_1$ satisfies Lyapunov's condition, which leads to $R_1$ being asymptotically normal. 

Similar to the strategy for analyzing $R_1$, we first write
\begin{align*}
R_{2} &= \frac{1}{\sqrt n (n_{[k]}-1)}\frac{1}{\pi^4}\sum_{\substack{N(k)+1\leq i < j\leq N(k+1)}} \sgij \Big[(X_i^{k\top}\Sigmakx^{-1}X_j^k Y_j^k(1) + X_j^{k\top}\Sigmakx^{-1}X_i Y_i(1)) \\
& \quad - (X_i^{k\top}\ProjBeta{1} + X_j^{k\top}\ProjBeta{1})\Big] \\
& = \frac{1}{\sqrt n (n_{[k]}-1)^2}\frac{1}{\pi^4}\sum\limits_{i=N(k)+1}^{N(k+1)} \sum\limits_{j=1}^{i-1}\sgij\Big[ (X_i^{k\top}\Sigmakx^{-1}X_j^k Y_j^k(1) + X_j^{k\top}\Sigmakx^{-1}X^k_i Y^k_i(1)) \\
& \quad - (X_i^{k\top}\ProjBeta{1} + X_j^{k\top}\ProjBeta{1})\Big] \\
& = \frac{1}{\pi^4}\sum\limits_{i=N(k)+2}^{N(k+1)} R_{2,i},
\end{align*}
where 
\[R_{2,i} = \frac{1}{\sqrt n (n_{[k]}-1)^2}\sum\limits_{j=1}^{i-1}\sgij\Big[ (X_i^{k\top}\Sigmakx^{-1}X_j^k Y_j^k(1) + X_j^{k\top}\Sigmakx^{-1}X^k_i Y^k_i(1)) - (X_i^{k\top}\ProjBeta{1} + X_j^{k\top}\ProjBeta{1})\Big].\]
We can show that $R_{2,i}$ is a martingale sequence. Following a similar strategy as in the proof of Theorem~1 in \cite{bhattacharya1992class}, we need to show the following:
\begin{align}
\sum\limits_{i = 1}^n \Ek \{R^2_{2, i} |\mathcal{F}^k_{i - 1}\} = \sigma^2_{2, n} & < \infty, \label{eq4} \\
\sum\limits_{i = 1}^n \Ek \{R^2_{2, i} I (|R_{2, i}| > \epsilon)\} & \stackrel{P}\to 0. \label{eq5}
\end{align}

For \eqref{eq5}, by Markov inequality, we have the following.
\begin{align*}
\Ek \{R^2_{2,i} I (|R_{2, i}| > \epsilon)\} & \leq \sqrt{\Ek \{R^4_{2, i}\}} \sqrt{\Pr (|R_{2, i}| > \epsilon)} \leq \sqrt{\Ek \{R^4_{2, i}\}} \frac{1}{\epsilon} \sqrt{\Ek \{R_{2, i}^2\}}.
\end{align*}
Let
\[h_R(Z^k_i,Z^k_j) = \frac{1}{\sqrt n (n_{[k]}-1)^2}X_i^{k\top}\Sigmakx^{-1}X_j^k Y_j^k(1) + X_j^{k\top}\Sigmakx^{-1}X^k_i Y^k_i(1) - X_i^{k\top}\ProjBeta{1} - X_j^{k\top}\ProjBeta{1}.\]
We first compute the fourth moment of $R_{2, i}$: 
\begin{align*}
\Ek\{R^4_{2,i}\} & = \Ek \Big\{\sum\limits_{j=1}^{i-1} \sgij h_R(Z^k_i,Z^k_j) \Big\}^4 \\
& = \sum\limits_{j_1,j_2,j_3,j_4} \overline{G}_{i,j_1}\overline{G}_{i,j_2}\overline{G}_{i,j_3}\overline{G}_{i,j_4} \Ek\{h_R(Z^k_i,Z^k_{j_1})h_R(Z^k_i,Z^k_{j_2})h_R(Z^k_i,Z^k_{j_3})h_R(Z^k_i,Z^k_{j_4})\}.
\end{align*}
As we have $\Ek \{h_R (Z^k_1, z) \} = 0$, $\forall z$, the remaining terms not equal to zero are $h_R^4(Z^k_i,Z^k_j)$ and $h_R^2(Z^k_i,Z^k_{j_1})h_R^2(Z^k_i,Z^k_{j_2})$. For the latter, by Lemma~\ref{lem:variance}, we can show that
\[\Ek\{h_R^2(Z^k_i,Z^k_{j_1})h_R^2(Z^k_i,Z^k_{j_2})\} = \Ek[\Ek\{h_R^2(Z^k_i,Z^k_{j_1})|Z^k_i\}\Ek\{h_R^2(Z^k_i,Z^k_{j_2}|Z^k_i\}] = O(1).\]
Then we have
\[ \Ek\{R^4_{2,i}\} \leq \Big(\sum\limits_{j=1}^{i-1} \sgij^4 + \sum\limits_{j_1\neq j_2} \overline{G}_{i,j_1}^2\overline{G}_{i,j_2}^2\Big) \Ek\{h_R^4(Z^k_1,Z^k_2)\}\leq i \Ek\{h_R^4(Z^k_1,Z^k_2)\}.\]
Again, by Lemma~\ref{lem:variance}, we have 
\[\Ek\{R^2_{2,i}\} \to i \sigma^2_{2,i},\]
Then we have 
\begin{align*}
    \sum\limits_{i=N(k)}^n \Ek\{R^2_{2,i}I(|R_{2,i}|>\epsilon)\} &\leq \frac{n}{\epsilon} \sqrt{\Ek\{h_R^4(Z^k_i,Z^k_j)} + \frac{\sigma_{2,n}}{\sqrt{n}} .
\end{align*}
We are then left to show 
\begin{equation}\label{BG_cond5}
n^2 \Ek\{h_R^4(Z^k_1,Z^k_2)\} \to 0.
\end{equation}
For \eqref{eq4}, similarly, we have 
\begin{align*}
& \Ek \{R^2_{2,i}|\mathcal{F}_{i-1}\} = \Ek \Big\{\sum\limits_{j=1}^{i-1} \sgij h_R(Z^k_i,Z^k_j)|\mathcal{F}_{i-1} \Big\}^2 \\
&= \sum\limits_{j=1}^{i-1} \overline{G}_{i,j}^2 \Ek\{h_R^2(Z^k_i,Z^k_{j})|\mathcal{F}_{i-1}\} + \sum\limits_{j_1\neq j_2} \overline{G}_{i,j_{1}}\overline{G}_{i,{j_2}} \Ek\{h_R(Z^k_i,Z^k_{j_{1}})h_R(Z^k_i,Z^k_{j_{2}})|\mathcal{F}_{i-1}\}.
\end{align*}
We can verify that the expectation of the first term is the asymptotic variance and that the mean of the second term is zero. Then we only need to show that the variance of each term converges to zero, which is equivalent to showing the following two statements
\begin{align}
    &n^3 \Ek \left(\int h_R^2(Z^k_1, u) F (\diff u)\right)^2 \to 0, \text{ and} \label{BG_cond3}\\ 
     &n^4 \Ek \left(\int h_R(Z^k_1, u) h_R(Z^k_2, u) F (\diff u)\right)^2 \to 0. \label{BG_cond4}
\end{align}
The following lemma is devoted to showing that \eqref{BG_cond5}--\eqref{BG_cond4} hold. Therefore, we can obtain that $R_2$ is asymptotic normal. Combining $ R_1$ with $ R_2$ is asymptotically normal, and since $R_1$ is orthogonal to $R_2$, this leads to $M_2$ being asymptotically normal. Then, by Lemma \ref{lem:variance}, we reach our conclusion.

\begin{lemma}\label{lem:BGconditions}
Recall that we define the following notation in the proof of Lemma \ref{lem:BG}:
\[h_R(Z^k_i,Z^k_j) = \frac{1}{\sqrt n(n_{[k]}-1)}\{X_i^{k\top}\Sigmakx^{-1}X_j^k Y_j^k(1) + X_j^{k\top}\Sigmakx^{-1}X^k_i Y^k_i(1) - X_i^{k\top}\ProjBeta{1} - X_j^{k\top}\ProjBeta{1}\}.\]
Under Assumptions \ref{ap1}--\ref{ap3}, conditional on $B_i=k$, we have: as $n \rightarrow \infty$, almost surely,
\begin{enumerate}[label = (\alph*)]
\item \eqref{BG_cond5} holds: $n^2 \Ek\{h_R^4(Z^k_1,Z^k_2)\} \to 0$;
\item \eqref{BG_cond3} holds: $n^3 \Ek \left(\int h_R^2(Z^k_1, u) F (\diff u)\right)^2 \to 0$;
\item \eqref{BG_cond4} holds: $n^4 \Ek \left(\int h_R(Z^k_1, u) h_R(Z^k_2, u) F (\diff u)\right)^2 \to 0$.
\end{enumerate}
\end{lemma}

\begin{proof}
\noindent\textbf{Condition \eqref{BG_cond5}:}
\begin{align*}
    h_R&(Z^k_1,Z^k_2)^4 \\
    &= \frac{1}{n^2(n_{[k]}-1)^4}\left\{X^{k\top}_1\left(\Sigmakx^{-1}X^{k}_2Y^k_2(1) - \ProjBeta{1}\right) + X^{k\top}_2\left(\Sigmakx^{-1}X_1^kY_1^k(1)-\ProjBeta{1}\right)\right\}^4 \\
    & =  \frac{1}{n^2(n_{[k]}-1)^4}\Bigg[\left\{X^{k\top}_1\left(\Sigmakx^{-1}X^{k}_2Y^k_2(1) - \ProjBeta{1}\right)\right\}^4 \\
    &\quad + 4\left\{X^{k\top}_1\left(\Sigmakx^{-1}X^{k}_2Y^k_2(1) - \ProjBeta{1}\right)\right\}^3\left\{X^{k\top}_2\left(\Sigmakx^{-1}X_1^kY_1^k(1)-\ProjBeta{1}\right)\right\}\\
    &\quad + 6\left\{X^{k\top}_1\left(\Sigmakx^{-1}X^{k}_2Y^k_2(1) - \ProjBeta{1}\right)\right\}^2\left\{X^{k\top}_2\left(\Sigmakx^{-1}X_1^kY_1^k(1)-\ProjBeta{1}\right)\right\}^2 \\
    &\quad + 4\left\{X^{k\top}_1\left(\Sigmakx^{-1}X^{k}_2Y^k_2(1) - \ProjBeta{1}\right)\right\}\left\{X^{k\top}_2\left(\Sigmakx^{-1}X_1^kY_1^k(1)-\ProjBeta{1}\right)\right\}^3 \\
    &\quad + \left\{X^{k\top}_2\left(\Sigmakx^{-1}X_1^kY_1^k(1)-\ProjBeta{1}\right)\right\}^4\Bigg] \\
    &= O\left\{ \frac{1}{n^2(n_{[k]}-1)^4}(X_1^{k\top}\Sigmakx^{-1}X^{k}_2)^4(Y^k_1(1) + Y^k_2(1))^4\right\},
\end{align*}
Then we need to show
\[n^{-4} \Ek\{(X_1^{\top}\Sigmakx^{-1}X_2)^4(Y^k_1(1) + Y^k_2(1))^4\} = o (1),\]
which follows from \Holder's inequality and Assumption \ref{ap2} as follows:  
\begin{align*}
&n^{-4} \Ek\{(X_1^{k\top}\Sigmakx^{-1}X^k_2)^4(Y^k_1(1) + Y^k_2(1))^4\}\\
&= n^{-4} \Ek\{(X_1^{k\top}\Sigmakx^{-1}X^k_2)^4 \Ek\{(Y^k_1(1) + Y^k_2(1))^4|X_1,X_2\}\}\\
&\leq C_1 n^{-4} \Ek\{(X_1^{k\top}\Sigmakx^{-1}X^k_2)^2(X_1^{k\top}\Sigmakx^{-1}X^k_2)^2\}\\
&\leq C_1n^{-4} \Vert (X_1^{k\top}\Sigmakx^{-1}X^k_2)^2 \Vert_{\infty} \Ek\{(X_1^{k\top}\Sigmakx^{-1}X^k_2)^2\}\\
&\leq C_1M^4n^{-4} (1^{\top}\Sigmakx^{-1}1)^2\Ek\{(X_1^{k\top}\Sigmakx^{-1}X^k_1)\} \\
&\leq C_1M^4n^{-4}p^3 = O (p^{3} / n^{4}) = o (1).
\end{align*}
\noindent\textbf{Condition \eqref{BG_cond3}:}

The integration is a conditional expectation, so we have 
\begin{align*}
    \int &h_R^2(Z^k_1, u) F (\diff u) \\
    &=  \frac{1}{n(n_{[k]}-1)^2}\Big\{\tr\Big[\Ek\{(Y^{k}(1))^2 X^kX^{k\top}\}\Sigmakx^{-2}X^k_1X_1^{k\top}Y^2(1)\Big] \\
    &\quad +Y^k_1X^{k\top}_1\Sigmakx^{-1}X^k_1Y^k_1 + \ProjBeta{1}^{\top} X_1^k X_1^{k\top} \ProjBeta{1} + \ProjBeta{1}^{\top}\Sigmakx\ProjBeta{1} \\
    &\quad + \tr(\Ek\{Y^k(1)X^kX^{k\top}\}\Sigmakx^{-2}X^k_1X_1^{k\top}Y^k(1)) - 2 \ProjBeta{1}^{\top}X^k_1X_1^{k\top}\ProjBeta{1} \\
    &\quad- 2 X_1^{k\top}\Sigmakx^{-1}\Ek\{X^kX^{k\top}Y^k(1)\}\ProjBeta{1}\Big\} 
\end{align*}

By Assumption \ref{ap2} and \ref{ap3}, $Y^2(1)$ and $Y(1)$ are bounded conditional on $X_1$, so that $\Ek\{Y^2(1) XX^{\top}\}$ and $\Ek\{Y(1) XX^{\top}\}$ are equivalent to $c\Sigmakx$, then we have 
\begin{equation*}
     \int h_R^2(Z^k_1, u) F (\diff u) = O_p\Bigg( \frac{1}{n(n_{[k]}-1)^2}Y(1)X^{\top}\Sigmakx^{-1} XY(1) +X^{\top}\Sigmakx^{-1} XY(1)+ XX^{\top} \Bigg).
\end{equation*}
Then we have 
\[\Ek \left(\int h_R^2(Z^k_1, u) F (\diff u)\right)^2 = O\Big(\frac{1}{n^2(n_{[k]}-1)^4}\Ek\{(X^{\top}\Sigmakx^{-1}X)^2Y^4(1)\}\Big),\]
Then we have
\[\frac{n^3}{n^2(n_{[k]}-1)^4} \Ek\{(X^{\top} \Sigmakx^{-1}X)^2Y^4(1)\} \leq n^{-3}p^2=o(1).\]

\noindent\textbf{Condition \eqref{BG_cond4}:}

First, we have
\begin{align*}
    \int &h_R(Z^k_1, u) h_R(Z^k_2, u) F (\diff u) \\
    &=\frac{1}{n(n_{[k]}-1)^2}\Big[ 2X_1^{k\top}\Sigmakx^{-1}\Ek\{Y(1)X\Sigmakx^{-1}X^{\top}Y(1)\}\Sigmakx^{-1}X^k_2  + X_1^{k\top}\Sigmakx^{-1}X^k_2 \Ek\{X\Sigmakx^{-1}X^{\top}Y(1)\} \\
    &\quad - X_1^{k\top}\Ek\{X\Sigmakx^{-1}X^{\top}Y(1)\}\ProjBeta{1} - X_2^{k\top}\Ek\{X\Sigmakx^{-1}X^{\top}Y(1)\}\ProjBeta{1} \\
    &\quad + Y^k_2(1)X_2^{k\top}\Sigmakx^{-1}X_1^kY_1^k(1) - \ProjBeta{1}^{\top}\{X^k_1 Y^k_1(1) +X_2^k Y^k_2(1))  + \ProjBeta{1}^{\top} \Sigmakx \ProjBeta{1}\\
    &\quad - \ProjBeta{1}^{\top} X^k_1X^{k\top}_2\ProjBeta{1}\Big].
\end{align*}
As
\[\Ek\{(X\Sigmakx^{-1}X^{\top})Y^2(1)\} = O(1) \text{ and } \Ek\{(X\Sigmakx^{-1}X^{\top})Y(1)\} = O(1),\]
we have 
\[E \left(\int h(X_1, u) h(X_2, u) F (\diff u)\right)^2 = O\Big(\frac{p}{n^2(n_{[k]}-1)^4}\Big). \] 
Then we have
\[n^4 E \left(\int h_R(Z^k_1, u) h_R(Z^k_2, u) F(\diff u)\right)^2 \to 0.\]
\end{proof}

Therefore, we have obtained that $M_2^k$ is asymptotically normal, and the remaining part is to compute its asymptotic variance. To simplify our calculation, we assume that $X^k_i$ is centered in stratum $k$, denoted as $\tilde X_i$, without loss of generality. Then we have
    \begin{align*}
    \var&\{M^k_2|B_i=k\} \\
    &= \frac{1}{n(n_{[k]}-1)^2}\frac{1}{\pi^4_{[k]}}\Bigg[ \sum\limits_{1\leq i\neq j \leq n_{[k]}} (A_i-\pi_{[k]})^2 A_j^2 \Ek\{\tilx_i^{\top} \Sigmakx^{-1} \tilx_j Y_j(1)\}^2 \\
    &\quad + \sum\limits_{1\leq i\neq j\leq n_{[k]}} (A_i-\pi_{[k]})(A_j-\pi_{[k]})A_iA_j \Ek\{\tilx_i^{\top}\Sigmakx^{-1}\tilx_j Y_j(1)Y_i(1) \tilx_i^{\top}\Sigmakx^{-1}\tilx_j\} \\
    &\quad + \sum\limits_{1 \leq i\neq j\neq l\leq n_{[k]}} (A_i-\pi_{[k]})^2A_jA_l \Ek\{\tilx^{\top}_i \Sigmakx^{-1} \tilx_j Y_j(1) Y_l(1) \tilx_l^{\top} \Sigmakx^{-1} \tilx_i\}
    \Bigg] \\
    &= \frac{1}{n(n_{[k]}-1)^2}\frac{1}{\pi^4_{[k]}} \Bigg[  \sum\limits_{1\leq i\leq j \neq n_{[k]}} (A_i-\pi_{[k]})^2 A_j^2 \Ek\{Y_j(1)\tilx^{\top}_j\Sigmakx^{-1} \tilx_j Y_j(1)\} \\
    &\quad + \sum\limits_{1\leq i\neq j\leq n_{[k]}} (A_i-\pi_{[k]})(A_j-\pi_{[k]})A_iA_j \Ek\{(\tilx_i^{\top}\Sigmakx^{-1}\tilx_j)^2 Y_j(1)Y_i(1)\} \\
    &\quad + \sum\limits_{1 \leq i\neq j\neq l\leq n_{[k]}} (A_i-\pi_{[k]})^2A_jA_l \ProjBeta{1}^{\top} \Sigmakx \ProjBeta{1}
    \Bigg] \\
    &= \frac{p_{[k]} }{n_{[k]}(n_{[k]}-1)^2} \frac{1}{\pi^4_{[k]}}\Bigg[\sum\limits_{1\leq i\neq j\leq n_{[k]}} \{(1-2\pi_{[k]})A_iA_j + \pi^2_{[k]}A_j\}\Ek\{Y_j(1)\tilx^{\top}_j\Sigmakx^{-1} \tilx_j Y_j(1)\} \\
    &\quad + \sum\limits_{1\leq i\neq j\leq n_{[k]}} (1-\pi_{[k]})^2 A_i A_j \Ek\{(\tilx_i^{\top}\Sigmakx^{-1}\tilx_j)^2 Y_j(1)Y_i(1)\} \\
    &\quad + \sum\limits_{1 \leq i\neq j\neq l\leq n_{[k]}} \{(1-2\pi_{[k]})A_iA_jA_l + \pi^2_{[k]}A_jA_l\} \ProjBeta{1}^{\top} \Sigmakx \ProjBeta{1}
    \Bigg]  \\
    & = \frac{p_{[k]}}{n_{[k]}-1} \frac{(1-\pi_{[k]})}{\pi^2_{[k]}} \Ek\{Y_j(1)\tilx^{\top}_j\Sigmakx^{-1} \tilx_j Y_j(1)\} \\
    &\quad + \frac{p_{[k]}}{n_{[k]}-1} \frac{(1-\pi_{[k]})^2}{\pi^2_{[k]}}  \Ek\{(\tilx_i^{\top}\Sigmakx^{-1}\tilx_j)^2 Y_j(1)Y_i(1)\} \\
    &\quad +\frac{p_{[k]}(1-\pi_{[k]})}{\pi_{[k]}} \ProjBeta{1}^{\top} \Sigmakx \ProjBeta{1},
\end{align*}
where $\ProjBeta{1} = \Sigmakx^{-1} \eta_{[k]} (1)$ is the population-level regression coefficient. Because $M_2^k$'s are mutually independent across different values of $k$, we have 
\[\var\{M_2\} = \sum\limits_{k=1}^K \var\{M_2^k|B_i=k\} = \zetaSec{Y(1)} + \sigma^2_{\mathrm{I},X(1)}.\]
\end{proof}

\subsection{Results related to the sample Gram matrix estimator}
\label{app:matrices}

We first record the following well-known Sherman-Morrison formula \citep{derezinski2021sparse} that is used throughout the paper.

\begin{lemma}
\label{lem:Sherman-Morrison}
Given a real-valued $p \times p$-matrix $A$ and two $p$-dimensional vectors $u$ and $v$, the following identities hold as long as all the inverses appeared are well defined:
\begin{align}
& (A + u v^{\top})^{-1} = A^{-1} - \frac{A^{-1} u v^{\top} A^{-1}}{1 + v^{\top} A^{-1} u}, \label{eq:Sherman-Morrison 1} \\
& v^{\top} (A + u v^{\top})^{-1} = \frac{v^{\top} A^{-1}}{1 + v^{\top} A^{-1} u}. \label{eq:Sherman-Morrison 2}
\end{align}
\end{lemma}

Next, the following version of the matrix Bernstein or moment inequality \citep{rudelson1999random, tropp2012user, vershynin2012close, bandeira2023matrix} is frequently used when controlling the error of using the sample Gram matrix $\hat{\Sigma}_{[k]}$ as an estimator of the population Gram matrix $\Sigma_{[k]}$, since we consider the superpopulation model here.

\begin{lemma}
\label{lem:matrix Bernstein}
Given a sequence $\{W_{i}\}_{i = 1}^{n}$ of independent and symmetric random matrices with dimension $p$. Assume that each matrix satisfies:
\begin{align*}
E [W_i] = 0 \text{ and } \lambda_{\max} (W_i) \lesssim k \text{ almost surely,}
\end{align*}
for some finite $k$, which is possibly a sequence $k = k (n)$ growing with $n$. Let $S_{n} = \sum\limits_{i = 1}^{n} W_{i}$. Then for all $t \geq 0$,
\begin{align*}
P \left( \lambda_{\max} (S_{n}) \geq t \right) \leq p \cdot \exp \left( - \frac{t^{2} / 2}{\nu^{2} + k t / 3} \right), \text{ where } \nu^{2} = \Big\Vert \sum_{i = 1}^{n} E [W_i^{2}] \Big\Vert_{\op}.
\end{align*}
In particular, when $k = p$, the following hold:
\begin{align*}
\frac{1}{n} S_{n} = O_{P} \left( \frac{p}{n} + \sqrt{\frac{p}{n}} \right).
\end{align*}
\end{lemma}

It is worth noting that for $\{X_{i}\}_{i = 1}^{n}$ satisfying Assumptions~\ref{ap1}--\ref{ap3}, we have
\begin{align}
\label{matrix l_inf}
\lambda_{\max} (X_i X_i^{\top}) \lesssim p \text{ almost surely.}
\end{align}

\begin{lemma}\label{lem: sample covariance1}
Suppose that $\{X_i\}_{i=1}^n$ are i.i.d.\ random vectors. Under Assumptions~\ref{ap1}--\ref{ap3}, when $p=o(n)$, given any two squared integrable functions $f_1, f_2: \mathbb{R}^{p} \rightarrow \mathbb{R}$ of $X$, we have
\begin{align*}
& E\{f_1 (X_1) X^{\top}_1\HSigmaX^{-1}_{-(1,2)}X_2 f_2 (X_2)\} = O (1), \text{ and } \\
& E\big[f_1 (X_1) X^{\top}_1\{\HSigmaX^{-1}_{-(1,2)} - \SigmaX^{-1}\}X_2 f_2 (X_2)\big] = o (1).
\end{align*}
\end{lemma}

\begin{proof}
We first show that $E\{f (X_1) X^{\top}_1\SigmaX^{-1}X_2 f_2 (X_2)\} = O (1)$, which is a consequence of the following sequence of (in)equalities: 
\begin{align*}
E\{f_1 (X_1) X^{\top}_1\SigmaX^{-1}X_2 f_2 (X_2) \} &= E\{f_1 (X_1) X^{\top}_1\SigmaX^{-1}\SigmaX \SigmaX^{-1}X_2 f_2 (X_2)\} \\
& = E\{f_1 (X) X\}\SigmaX^{-1}E\{XX^{\top}\}\SigmaX^{-1}E\{X f_2 (X)\}\\
& = E [\Pi(f_1|X) \cdot \Pi(f_2|X)] \\
& \leq \Vert \Pi(f_1|X) \Vert_{L_2 (P)} \cdot \Vert \Pi(f_2|X) \Vert_{L_2 (P)} \\
& \leq \Vert f_1 \Vert_{L_2 (P)} \cdot \Vert f_2 \Vert_{L_2 (P)} = O(1),
\end{align*}
where we introduce the notation that $\Pi(f | X) = X^{\top}\SigmaX^{-1} E\{X f (X)\}$ is the $L_{2} (P)$-projection of $f$ on the linear span of $X$. Then we can show that   
\[E\big[f_1 (X_1) X^{\top}_1\{\HSigmaX^{-1}_{-(1,2)} - \SigmaX^{-1}\}X_2 f_2 (X_2)\big] = o(1),\]
which follows from:
\begin{align*}
& |E\{f_1 (X_1) X^{\top}_1\{\HSigmaX^{-1}_{-(1,2)} - \SigmaX^{-1}\}X_2 f_2 (X_2)\}| = |E\{f_1 (X) X^{\top}\} E\{\HSigmaX^{-1}_{-(1,2)} - \SigmaX^{-1}\}E\{X f_2 (X)\}| \\
& \lesssim \Vert \Pi(f_1|X) \Vert_{L_2 (P)} \cdot \Vert \Pi(f_2|X) \Vert_{L_2 (P)} \cdot \norm{E\{\HSigmaX^{-1}_{-(1,2)} - \SigmaX^{-1}\}}{\op} \lesssim \norm{E\{\HSigmaX^{-1}_{-(1,2)} - \SigmaX^{-1}\}}{\op}.
\end{align*}
As $X_i$ is uniformly bounded (hence sub-Gaussian), by Lemma~\ref{lem:matrix Bernstein} and \eqref{matrix l_inf}, we have 
\[\norm{\HSigmaX - \SigmaX}{{\rm{op}}}= O_P \left( \frac{p}{n} + \sqrt{\frac{p}{n}} \right).\]
By Weyl's inequality,
\[|\lambda_{\min}(\HSigmaX) - \lambda_{\min}(\SigmaX)|\leq \norm{\HSigmaX - \SigmaX}{{\rm{op}}}.\]
Since we assume that $p = o(n)$, with probability converging to 1,
\begin{align*}
\norm{\HSigmaX - \SigmaX}{{\rm{op}}} \leq \frac{\lambda_{\min} (\Sigma)}{2}.
\end{align*}
Then we have, with probability converging to 1,
\[\lambda_{\min}(\HSigmaX) \geq \frac{\lambda_{\min}(\SigmaX)}{2}\]
and 
\[\norm{\HSigmaX^{-1}}{\op} \leq \frac{2}{\lambda_{\min}(\SigmaX)}.\]
Therefore, we have
\begin{align*}
    \norm{E\{\HSigmaX^{-1}_{-(1,2)} - \SigmaX^{-1}\}}{\op} &= \norm{E\big[\HSigmaX^{-1}\{\HSigmaX_{-(1,2)} - \SigmaX\}\SigmaX^{-1}\big]}{\op} \\
    &\leq E [\norm{\HSigmaX^{-1}}{\op} \norm{\{\HSigmaX_{-(1,2)} - \SigmaX\}}{\op}] \norm{\SigmaX^{-1}}{\op} \\
    &= o(1).
\end{align*}
Thus, we finish the proof of $E\big[f_1 (X_1) X^{\top}_1\{\HSigmaX^{-1}_{-(1,2)} - \SigmaX^{-1}\}X_2 f_2 (X_2)\big] = o (1)$.
\end{proof}

\begin{lemma}\label{lem:leave-out-matrix}
Suppose that $\{X_i\}_{i=1}^n$ are i.i.d.\ random vectors. $\SigmaX = E\{X_iX_i^{\top}\}$, $\HSigmaX = n^{-1} \sum_{i=1}^n X_iX_i^{\top}$. Under Assumptions \ref{ap1}--\ref{ap3}, when $p=o(n)$, we have
\begin{itemize}
    \item $E\{X_1^{\top}(\HSigmaX^{-1} - \SigmaX ^{-1}) X_2\}^2 = o(n)$;
    \item $E\{X_1^{\top}(\HSigmaX^{-1} - \SigmaX ^{-1}) X_2\}\{X_1^{\top}(\HSigmaX^{-1} - \SigmaX ^{-1}) X_3\} = o(1)$;
    \item $E\{X_1^{\top}(\HSigmaX^{-1} - \SigmaX ^{-1}) X_2\}\{X_3^{\top}(\HSigmaX^{-1} - \SigmaX ^{-1}) X_4\} = o(1)$.
\end{itemize}    
\end{lemma}

\begin{proof}
The main challenge lies in the fact that $\HSigmaX$ is not independent of $X_1,X_2,X_3,X_4$. To decouple this correlation, we apply the Sherman-Morrison formula (Lemma \ref{lem:Sherman-Morrison}) multiple times. We first introduce the following notation: 
\begin{align*}
f_{I} (Q; (i,j)) = X_i^{\top} (\HSigmaX_{-I}^{-1} - Q) X_j.
\end{align*}
Here, $Q \in \mathbb{R}^{p \times p}$ denotes a deterministic matrix, the subscript $I \subseteq \{0,1,2,\dots,n\}$ indicates that the sample Gram matrix is estimated leaving out units with indices in $I$. Specifically, we use $(0)$ as the subscript to denote the sample Gram matrix estimator computed using all units. When $(i, j) \subseteq I$, $\HSigmaX_{-I}^{-1}$ is independent of $X_i$ and $X_j$, we can use a strategy similar to that we applied in the proof of Lemma~\ref{lem: sample covariance1}. We introduce the following lemma to show that, when the Sherman–Morrison formula is applied multiple times, the convergence rate of the terms of interest remains unchanged after removing the dependence on $\HSigmaX$. This lemma could be of independent interest for related problems that involve leave-$m$-out analysis of inverse sample Gram matrices \citep{liu2017semiparametric}.

\begin{lemma}
\label{lem:repeated Sherman-Morrison}
Under the same conditions as in Lemma~\ref{lem:leave-out-matrix}, we have
\begin{itemize}
\item $E \{f^2_{(0)}(\SigmaX ^{-1};(1,2))\} = E \{f^2_{(1,2)}(\SigmaX ^{-1};(1,2))\}+O(p)$; 
    
\item $E \{f_{(0)}(\SigmaX ^{-1};(1,2))f_{(0)}(\SigmaX ^{-1};(1,3))\} = E \{f_{(1,2,3)}(\SigmaX ^{-1};(1,2))f_{(1,2,3)}(\SigmaX ^{-1};(1,3))\}+o(1)$;
    
\item $E \{f_{(0)}(\SigmaX ^{-1};(1,2))f_{(0)}(\SigmaX ^{-1};(3,4))\} =E \{f_{(1,2,3,4)}(\SigmaX ^{-1};(1,2))f_{(1,2,3,4)}(\SigmaX ^{-1};(3,4))\}+o(1)$.
\end{itemize}
\end{lemma}
We continue the proof of Lemma~\ref{lem:leave-out-matrix} by taking Lemma~\ref{lem:repeated Sherman-Morrison} as a given and will prove Lemma~\ref{lem:repeated Sherman-Morrison} afterward. By Lemma~\ref{lem:repeated Sherman-Morrison}, we only need to show that $E\{f^2_{(1,2)}(\SigmaX ^{-1};1,2)\} = o(n)$. We have
\begin{align}
   &E\{f^2_{(1,2)}(\SigmaX ^{-1};1,2)\} \notag\\
   &= E\big[X_1^{\top}\{\HSigmaX_{-(1,2)}^{-1} - \SigmaX ^{-1})\} X_2X_2^{\top}\{\HSigmaX_{-(1,2)}^{-1} - \SigmaX ^{-1})\} X_1 \big] \notag \\
   &= E\Big(X_1^{\top}E\big[\{\HSigmaX_{-(1,2)}^{-1} - \SigmaX ^{-1})\} X_2X_2^{\top}\{\HSigmaX_{-(1,2)}^{-1} - \SigmaX ^{-1})\}|X_1 \big]X_1 \Big)\notag  \\
   & = E\big[X_1^{\top}\{\HSigmaX_{-(1,2)}^{-1} - \SigmaX ^{-1})\} \SigmaX \{\HSigmaX_{-(1,2)}^{-1} - \SigmaX ^{-1})\} X_1 \big]\notag  \\
   &\leq \norm{\SigmaX }{\op} E\{\norm{\HSigmaX_{-(1,2)}^{-1} - \SigmaX ^{-1}}{\op}^2\} E\{\norm{X_1}{2}^2\}\notag \\
   &= O(1) \cdot O \left( \frac{p}{n} \right) \cdot O(p) \notag  \\
   &= o(n). \notag
\end{align}
Therefore, we have
\begin{equation}
    E \{f^2_{(0)}(\SigmaX ^{-1};(1,2))\} = o(n). \label{eq: SH-result1}
\end{equation}
Similarly, for the second statement, we only need to show
\begin{align*}
E\{f_{(1,2,3)}(\SigmaX ^{-1};(1,2))f_{(1,2,3)}(\SigmaX ^{-1};(1,3))\} = o(1).
\end{align*}
Following a similar strategy as in Lemma~\ref{lem: sample covariance1}, we have 
\begin{align*}
    &|E\{f_{(1,2,3)}(\SigmaX ^{-1};(1,2))f_{(1,2,3)}(\SigmaX ^{-1};(1,3))\}| \\
    &=|E(X)^{\top} E [\HSigmaX_{-(1,2,3)}^{-1} - \SigmaX ^{-1}] \SigmaX  E [\HSigmaX_{-(1,2,3)}^{-1} - \SigmaX ^{-1}] E(X)| \\
    &\leq \norm{E(X)}{\op}^2 \norm{\SigmaX }{\op} \norm{E [\HSigmaX_{-(1,2,3)}^{-1} - \SigmaX^{-1}]}{\op}^2 \notag\\
    &=o(1).
\end{align*}
Then we have
\begin{equation}
    E\{f_{(0)}(\SigmaX ^{-1};(1,2))f_{(0)}(\SigmaX ^{-1};(1,3))\} = o(1). \label{eq: SH-result2}
\end{equation}
Similarly, we can obtain the desired results for $E\{f_{(1,2,3,4)}(\SigmaX ^{-1};(1,2))f_{(1,2,3,4)}(\SigmaX ^{-1};(3,4))\} = o(1)$
and 
\begin{equation}
    E\{f_{(0)}(\SigmaX ^{-1};(1,2))f_{(0)}(\SigmaX ^{-1};(3,4))\} = o(1). \label{eq: SH-result3}
\end{equation}
In summary, equations \eqref{eq: SH-result1}, \eqref{eq: SH-result2} and \eqref{eq: SH-result3} are proved.
\end{proof}
We next prove Lemma~\ref{lem:repeated Sherman-Morrison}.
\begin{proof}[Proof of Lemma~\ref{lem:repeated Sherman-Morrison}]
By Lemma~\ref{lem:Sherman-Morrison}, we have
    \begin{align}
        f_{(0)}(\SigmaX ^{-1};(i,j)) &= f_{(i)}(\SigmaX ^{-1};(i,j)) - \gamma_{(i)}(i) f_{(i)}(0;(i,j)), \label{eq:repeatS}\\
        f_{(0)}(0;(i,j)) &= f_{(i)}(0;(i,j))\{1-\gamma_{(i)}(i)\}, \label{eq:repeat0}
    \end{align}
    where we let 
    \[\gamma_{I}(i) = \frac{X_i^{\top}\HSigmaX_{-I}^{-1}X_i}{n+\sum\limits_{j=i}^{|I|}X_j^{\top}\HSigmaX_{-I}^{-1}X_j}.\]
    Here, $|I|$ is the number of elements in $I$. Then we can obtain
     \begin{align*}
        \Ek \{f_{(0)}(\SigmaX ^{-1};(i,j))\} &= \Ek \{f_{(i)}(\SigmaX ^{-1};(i,j))\} - \Ek \{f_{(i)}(0;(i,j)) \gamma_{(i)}(i)\}, \\
        \Ek \{f_{(0)}(0;(i,j))\} &= \Ek \{(1-\gamma_{i}(i))f_{(i)}(0;(i,j))\}.
     \end{align*}
     The key is that we can represent all three terms by the second-order combinations of their leave-one-out forms $f_{(i)}(\SigmaX ^{-1};(i,j))$ and $f_{(i)}(0;(i,j))$ with the product $\gamma_{I}(i)$. Then we apply equation \eqref{eq:repeatS} to the first goal:
\begin{align*}
    &f^2_{(0)}(\SigmaX ^{-1};(1,2)) \\
    &= \{f_{(1)}(\SigmaX ^{-1};(1,2)) - \gamma_{(1)}(1) f_{(1)}(0;(1,2))\}^2 \\
    &= f^2_{(1)}(\SigmaX ^{-1};(1,2)) + \gamma^2_{(1)}(1) f^2_{(1)}(0;(1,2)) - 2\gamma_{(1)}(1) f_{(1)}(\SigmaX ^{-1};(1,2))f_{(1)}(0;(1,2)).
\end{align*}
For the second term, we can directly bound $E\{f^2_{(1)}(0;(1,2))\}$ as $\gamma_{(1)}(1)<1$. For the third term, we can obtain a similar conclusion because
\begin{align}\label{eqgamma}
    E\Big\{\gamma_{i}(i)f_{(1)}(\SigmaX ^{-1};(1,2))f_{(1)}(0;(1,2))\Big\} &= E\{\gamma_{i}(i)X_1^{\top} (\hat\Sigma^{-1}_{-1} - \Sigma^{-1})X_2X_2^{\top} \HSigmaX^{-1}_{-1} X_1\} \notag\\
    &\leq E|X_1^{\top} \hat\Sigma^{-1}_{-1}X_2X_2^{\top} \HSigmaX_{-1}^{-1} X_1| +E| X_1^{\top} \Sigma^{-1}X_2X_2^{\top} \HSigmaX^{-1}_{-1} X_1|\notag\\
    &= E\{X_1^{\top} \hat\Sigma^{-1}_{-1}X_2X_2^{\top} \HSigmaX_{-1}^{-1} X_1\} +E\{ X_1^{\top} \Sigma^{-1}X_2X_2^{\top} \HSigmaX^{-1}_{-1} X_1\}.
\end{align}

That is, we are left to analyze the following three terms
\begin{equation}
    f^2_{(1)}(\SigmaX ^{-1};(1,2)),\ f^2_{(1)}(0;(1,2)),\ f_{(1)}(\SigmaX ^{-1};(1,2))f_{(1)}(0;(1,2)). \label{eqSH}
\end{equation}
Then we can apply the Sherman--Morrison formula to the terms in \eqref{eqSH}, leaving the 2nd unit out. By \eqref{eq:repeatS} and \eqref{eq:repeat0}, the result only concerns 
\begin{equation*}
    f^2_{(1,2)}(\SigmaX ^{-1};(1,2)),\ f^2_{(1,2)}(0;(1,2)),\ f_{(1,2)}(\SigmaX ^{-1};(1,2))f_{(1,2)}(0;(1,2)). 
\end{equation*}
In particular, we have
\[\Ek\{X^{\top}_1\SigmaX ^{-1}X_2X^{\top}_2\SigmaX ^{-1}X_1\} = O(p)\]
and 
\begin{equation}
    E f^2_{(1,2)}(0;1,2) = E\big\{f_{(1,2)}(\SigmaX ^{-1};1,2) + X^{\top}_1\SigmaX ^{-1}X_2\big\}^2=O(p). \label{eq:SHR2}
\end{equation}
As of $f_{(1,2)}(\SigmaX ^{-1};(1,2))f_{(1,2)}(0;(1,2))$, we have
\begin{align}
    &|E\{f_{(1,2)}(\SigmaX ^{-1};(1,2))f_{(1,2)}(0;(1,2))\}| \notag \\
    &=|E\{f_{(1,2)}(\SigmaX ^{-1};(1,2))\}\{f_{(1,2)}(\SigmaX ^{-1};(1,2)) + X_1^{\top} \SigmaX ^{-1}X_2\}|\notag\\
    &=O(p)+o(1) = O(p). \label{eq:SHR3}
\end{align}
We now substitute \eqref{eq:SHR2} and \eqref{eq:SHR3} back into the Sherman–Morrison formula step by step. Then we can show that 
\[E \{f^2_{(0)}(\SigmaX ^{-1};(1,2))\} = E \{f^2_{(1,2)}(\SigmaX ^{-1};(1,2))\}+O(p).\]
Similarly, we can write $f_{(0)}(\SigmaX ^{-1};(1,2))f_{(0)}(\SigmaX ^{-1};(1,3))$ as
\begin{align*}
    &f_{(0)}(\SigmaX ^{-1};(1,2))f_{(0)}(\SigmaX ^{-1};(1,3))\\
    &= \{f_{(1)}(\SigmaX ^{-1};(1,2)) - \gamma_{(1)}(1) f_{(1)}(0;(1,2))\}\{f_{(1)}(\SigmaX ^{-1};(1,3)) - \gamma_{(1)}(1) f_{(1)}(0;(1,3))\} \\
    &= f_{(1)}(\SigmaX ^{-1};(1,2))f_{(1)}(\SigmaX ^{-1};(1,3)) \\
    &\quad - \gamma_{(1)}(1)\big\{f_{(1)}(\SigmaX ^{-1};(1,2))f_{(1)}(0;(1,3)) + f_{(1)}(0;(1,2))f_{(1)}(\SigmaX ^{-1};(1,3))\big\} \\
    & \quad + \gamma^2_{(1)}(1) f_{(1)}(0;(1,2))f_{(1)}(0;(1,3)).
\end{align*}
Applying Lemma~\ref{lem:Sherman-Morrison} multiple times, we can show the following results by Lemma~\ref{lem: sample covariance1_plus}: 
\begin{align}\label{eqlem10_1}
\Big| E \Big[\prod_{i=1}^3\prod_{j=1}^3\prod_{k_i,k_j} r_I(i)^{k_i}\{1-r_I(j)\}^{k_j} &f_{(1,2,3)}(\Omega_1;(1,2))f_{(1,2,3)}(\Omega_2;(1,3))\Big] \Big| = o (1),
\end{align}
where $k_i$ and $k_j$ can be $0,1,$ or $2$, and $\Omega_1$ and $\Omega_2$ can be $0$ or $\Sigma^{-1}$. We can then obtain the desired results by showing that
\begin{align*}
    E\{f_{(1,2,3)}(\SigmaX ^{-1};(1,2))f_{(1,2,3)}(0;(1,3))\}= o(1),\\
    E\{f_{(1,2,3)}(0;(1,2))f_{(1,2,3)}(0;(1,3))\} = O(1).
\end{align*}
We have
\begin{align*}
    & E\{f_{(1,2,3)}(0;(1,2))f_{(1,2,3)}(0;(1,3))\} \notag\\
    &\leq E\{f_{(1,2,3)}(\SigmaX ^{-1};(1,2)) + X_1^{\top} \SigmaX ^{-1}X_2\}\{f_{(1,2,3)}(\SigmaX ^{-1};(1,3)) + X_1^{\top} \SigmaX ^{-1}X_3\}\notag\\
    &= E \{X_2^{\top}(\HSigmaX_{-(1,2,3)}^{-1} - \SigmaX^{-1})X_1X^{\top}_1(\HSigmaX_{-(1,2,3)}^{-1} - \SigmaX^{-1})X_3+2X_2^{\top}(\HSigmaX_{-(1,2,3)}^{-1} - \SigmaX^{-1})X_1X^{\top}_1\SigmaX^{-1}X_3 \\
    & \quad + X_2^{\top} \SigmaX^{-1}X_1X^{\top}_1\SigmaX^{-1}X_3\} \notag\\
    &\leq \norm{E\{X\}}{2}^2 \left( E \{\norm{\HSigmaX_{-(1,2,3)}^{-1} - \SigmaX^{-1}}{\op}^2 + 2\norm{\HSigmaX_{-(1,2,3)}^{-1} - \SigmaX^{-1}}{\op}\} + \frac{1}{\lambda_{\min}} \right) \\
    &=O(1).
\end{align*}
Similarly, we have
\begin{align*}
    &|E\{f_{(1,2,3)}(\SigmaX ^{-1};(1,2))f_{(1,2,3)}(0;(1,3))\}| \notag\\
    &=|E\{f_{(1,2,3)}(\SigmaX ^{-1};(1,2))\}\{f_{(1,2,3)}(\SigmaX ^{-1};(1,3)) + X_1^{\top} \SigmaX ^{-1}X_3\}|\notag\\
    &=o(1).
\end{align*}
Then we have 
\[E \{f_{(0)}(\SigmaX ^{-1};(1,2))f_{(0)}(\SigmaX ^{-1};(1,3))\} = E \{f_{(1,2,3)}(\SigmaX ^{-1};(1,2))f_{(1,2,3)}(\SigmaX ^{-1};(1,3))\}+o(1).\]
The strategy of analyzing $E \{f_{(0)}(\SigmaX ^{-1};(1,2))f_{(0)}(\SigmaX ^{-1};(3,4))\}$ is similar, as we need to show that there exists a universal constant $C > 0$ such that
\begin{align} \label{eqlem10_2}
\Big| E \Big[\prod_{i=1}^4\prod_{j=1}^4 \prod_{k_i,k_j} r^{k_i}_I(i)\{1-r_I(j)\}^{k_j} f_{(1,2,3,4)}(\Omega_1;(1,2))f_{(1,2,3,4)}(\Omega_2;(3,4))\Big] \Big| = o (1).
\end{align}
Then we can have
\[E \{f_{(0)}(\SigmaX ^{-1};(1,2))f_{(0)}(\SigmaX ^{-1};(3,4))\} =E \{f_{(1,2,3,4)}(\SigmaX ^{-1};(1,2))f_{(1,2,3,4)}(\SigmaX ^{-1};(3,4))\}+o(1).\]
\end{proof}

\begin{lemma}\label{lem: sample covariance1_plus}
Under the assumptions of Lemma~\ref{lem:leave-out-matrix}, let 
    \[r_I(i) = \frac{X_i^{\top}\hat\Sigma^{-1}_{-I}X_i}{n+\sum\limits_{j=i}^{|I|}X_j^{\top}\hat\Sigma^{-i}_{-I}X_j}.\]
    Then we have
    \begin{itemize}
        \item When $I = \{1,2,3\}$, we have 
        \begin{align*}      E\Big[\prod_{i=1}^3\prod_{j=1}^3\prod_{k_i,k_j} r_I(i)^{k_i}\{1-r_I(j)\}^{k_j} f_{(1,2,3)}(\Omega_1;(1,2))f_{(1,2,3)}(\Omega_2;(1,3))\Big] = o (1).
        \end{align*}
        \item When $I = \{1,2,3,4\}$, we have 
        \begin{align*}     E\Big[\prod_{i=1}^4\prod_{j=1}^4 \prod_{k_i,k_j} r^{k_i}_I(i)\{1-r_I(j)\}^{k_j} f_{(1,2,3,4)}(\Omega_1;(1,2))f_{(1,2,3,4)}(\Omega_2;(3,4))\Big] = o (1),
        \end{align*}
    \end{itemize}
    where $k_i$ and $k_j$ can be $0,1,$ or $2$. $\Omega_1$ and $\Omega_2$ can be $0$ or $\Sigma^{-1}$ and $|I|$ is the number of elements in $I$.
\end{lemma}

\begin{proof}
Let $g(X_i) = X_i^{\top}\hat\Sigma^{-1}_{-I}X_i/n$, for $i=1,2,3,4$. We only need to prove the statements when $k_{i} = k_{j} = 2$. Because $p=o(n)$, $g(X_i)<1$, we can write $r_I(i)$ and $r^2_I (i)$ as Taylor series:
\[r_I(i) = g(X_i)\sum\limits_{j=0}^{m}(-1)^j\Big\{\sum\limits_{l=i}^{|I|}g(X_l)\Big\}^j + O\Big(\Big(\frac{p}{n}\Big)^{m+1}\Big)\]
and
\[r^2_I(i) = g(X_i)^2\sum\limits_{j=1}^{m}(-1)^{j+1}j\Big\{\sum\limits_{l=i}^{|I|}g(X_l)\Big\}^j + O\Big((m+1)\Big(\frac{p}{n}\Big)^{m+1}\Big).\]
Meanwhile, we have
\[1-r_I(i) = \sum\limits_{j=0}^{m}(-1)^j\Big\{\sum\limits_{l=i}^{|I|}g(X_l)\Big\}^j + O\Big(\Big(\frac{p}{n}\Big)^{m+1}\Big)\]
and
\[\{1-r_I(i)\}^2 = \sum\limits_{j=1}^{m}(-1)^{j+1}j\Big\{\sum\limits_{l=i}^{|I|}g(X_l)\Big\}^j + O\Big((m+1)\Big(\frac{p}{n}\Big)^{m+1}\Big).\]
We can choose a sufficiently large $m$ so that all the above remainders can be made arbitrarily small. Then we can simplify $\Big\{\sum\limits_{l=i}^{|I|}g(X_l)\Big\}^j$ for $i=1,2$ using the Binomial theorem:
\begin{align*}
\Big\{\sum\limits_{l=1}^{|I|}g(X_l)\Big\}^{j} = \sum\limits_{j_1+j_2+j_3=j} \binom{j}{j_1, j_2, j_3} g(X_1)^{j_1}g(X_2)^{j_2}g(X_3)^{j_3}
\end{align*}
and
\begin{align*}
\Big\{\sum\limits_{l=2}^{|I|}g(X_l)\Big\}^j = \sum\limits_{j_2+j_3=j} \binom{j}{j_2} g(X_2)^{j_2}g(X_3)^{j_3}.
\end{align*}
We also denote $\zeta_{j} \leq \sqrt{2}$ as a sequence converging to 1 as $j \rightarrow \infty$ such that $\zeta_{j}^{j} = j$. As we choose $p = o (n)$, both $g (X)$ and $\zeta_{j} g (X)$ are of order $O (p / n) = o (1)$ under Assumptions~\ref{ap2}--\ref{ap3}.

We first consider the case $I = \{1, 2, 3\}$. Then we can represent the leading terms of $r_{I}^{2} (i)$ and $\{1 - r_{I} (i)\}^{2}$ as follows, where we abuse the notation by ignoring the remainder terms:
\begin{align*}
& r_{I}^{2} (1) = g (X_{1})^{2} \sum_{j = 1}^{m} (-1)^{j + 1} \sum_{j_1 + j_2 + j_3 = j} \binom{j}{j_1, j_2, j_3} \{\zeta_{j} g (X_1)\}^{j_1} \{\zeta_{j} g (X_2)\}^{j_2} \{\zeta_{j} g (X_3)\}^{j_3}, \\
& \{1 - r_{I} (1)\}^{2} = \sum_{j = 1}^{m} (-1)^{j + 1} \sum_{j_1 + j_2 + j_3 = j} \binom{j}{j_1, j_2, j_3} \{\zeta_{j} g (X_1)\}^{j_1} \{\zeta_{j} g (X_2)\}^{j_2} \{\zeta_{j} g (X_3)\}^{j_3}, \\
& r_{I}^{2} (2) = g (X_2)^{2} \sum_{j = 1}^{m} (-1)^{j + 1} \sum_{j_2 + j_3 = j} \binom{j}{j_2} \{\zeta_{j} g (X_2)\}^{j_2} \{\zeta_{j} g (X_3)\}^{j_3}, \\
& \{1 - r_{I} (2)\}^{2} = \sum_{j = 1}^{m} (-1)^{j + 1} \sum_{j_2 + j_3 = j} \binom{j}{j_2} \{\zeta_{j} g (X_2)\}^{j_2} \{\zeta_{j} g (X_3)\}^{j_3}, \\
& r_{I}^{2} (3) = g (X_{3})^{2} \sum_{j = 1}^{m} (-1)^{j + 1} \{\zeta_{j} g (X_{3})\}^{j}, \\
& \{1 - r_{I} (3)\}^{2} = \sum_{j = 1}^{m} (-1)^{j + 1} \{\zeta_{j} g (X_{3})\}^{j}.
\end{align*}

Our proof strategy is then to explicitly expand \eqref{eqlem10_1} as a summation of monomials of $\zeta_{j} g (X_{i})$ with different degrees and then bound each of the monomials of degree $k$ by $O ((p / n)^{k})$ using the same projection argument as in \eqref{key projection}. Finally, as in \eqref{key expansion}, we can turn the summation of monomials back to a form with which we can directly apply the multinomial theorem. 

In particular, we can write \eqref{eqlem10_1} as follows:
\begin{align*}
& |\eqref{eqlem10_1}| = \Big| E \Big[ \prod_{l = 1}^3 r_I (l)^{2} \{1 - r_I(l)\}^{2} f_{(1,2,3)} (\Omega_1; (1, 2)) f_{(1,2,3)} (\Omega_2; (1, 3)) \Big] \Big| \\ 
& = \Bigg| E \Big[ \prod_{l = 1}^{3} g (X_{l})^{2}  \Big( \sum_{j = 1}^{m} (-1)^{j + 1} \sum_{\sum\limits_{i = 1}^{l} j_{l} = j} \binom{j}{j_{1}, \cdots, j_{l}} \prod_{i = 1}^{l} \{\zeta_{j_{l}} g (X_{j_{i}})\}^{j_{i}} \Big)^{2} f_{(1,2,3)} (\Omega_1; (1, 2)) f_{(1,2,3)} (\Omega_2; (1, 3)) \Big] \Bigg| \\
& \lesssim \prod_{l = 1}^{3} \Big( \frac{p}{n} \Big)^{2} \Big\{ \sum_{j = 1}^{m} (-1)^{j + 1} \sum_{\sum\limits_{i = 1}^{l} j_{l} = j} \binom{j}{j_{1}, \cdots, j_{l}} \prod_{i = 1}^{l} \Big( \frac{\zeta_{j_{l}} p}{n} \Big)^{j_{i}} \Big\}^{2} \Vert\HSigmaX^{-1}_{-(1,2,3)}-\Omega_1\Vert_{\op} \Vert\HSigmaX^{-1}_{-(1,2,3)}-\Omega_2\Vert_{\op} \\
& = o (1),
\end{align*}
where the first inequality follows from expanding the product of sums as the summation of monomials of $\zeta_{j} g (X_{i})$ with different degrees together with the projection argument as in \eqref{key projection} and the last line follows from the assumption $p = o (n)$ and the multinomial theorem. 

The same argument can be applied to show the desired bound for \eqref{eqlem10_2}.
\end{proof}

\begin{lemma}\label{lem:variance}
Under Assumptions~\ref{ap1}--\ref{ap3}, when $p=o(n)$, we have 
\[E\{X^{\top} (\HSigmaX^{-1} - \SigmaX^{-1})X\} = o(p),\ E\{X^{\top} (\HSigmaX^{-1} - \SigmaX^{-1})X\}^2 = O\left(\frac{p^3}{n}\right).\]
\end{lemma}

\begin{proof}
By the Sherman--Morrison formula in Lemma \ref{lem:Sherman-Morrison}, we have
\[\HSigmaX^{-1} = \HSigmaX^{-1}_{-i} - \HSigmaX^{-1}_{-i}X_i(n + X_i^{\top}\HSigmaX^{-1}_{-i}X_i)^{-1}X_i^{\top}\HSigmaX^{-1}_{-i}.\]
In turn, we have 
\begin{equation*}
X_i^{\top} (\HSigmaX^{-1} - \SigmaX^{-1})X_i = X_i^{\top} \{\HSigmaX^{-1}_{-i} - \SigmaX^{-1}\}X_i - \frac{(X_i^{\top}\HSigmaX^{-1}_{-i}X_i)^2}{n + X_i^{\top}\HSigmaX^{-1}_{-i}X_i}.
\end{equation*}
The first term in the above display can be upper bounded by:
\[E[X_i^{\top} \{\HSigmaX^{-1}_{-i} - \SigmaX^{-1}\}X_i]\leq E [\norm{X}{2}^{2}] E [\norm{\HSigmaX^{-1}_{-i} - \SigmaX^{-1}}{\op}] = O\left(p\cdot \sqrt\frac{p}{n}\right),\]
and the second term is of the same order as $(X_i^{\top} \HSigmaX^{-1}_{-i} X_i)^2/n$ when $p=o(n)$ and has a mean of order $O (p^2/n)$. Therefore, we have
\[\frac{1}{p} E\{X_i^{\top} (\HSigmaX^{-1} - \SigmaX^{-1})X_i\} = O \left( \sqrt{\frac{p}{n}} \right) = o (1).\]
Next, we consider $(X_i^{\top} (\HSigmaX^{-1} - \SigmaX^{-1})X_i)^2$. Using a similar argument, we can conclude that
\begin{align*}
E [(X_i^{\top} (\HSigmaX^{-1} - \SigmaX^{-1})X_i)^2] \lesssim E [(X_i^{\top} (\HSigmaX^{-1}_{-i} - \SigmaX^{-1})X_i)^2] + O \left( \frac{p^{4}}{n^{2}} \right).
\end{align*}
As $X_i$ is uniformly bounded, we have
\begin{align*}
& \ E \big[(X_i^{\top} \{\HSigmaX^{-1}_{-i} - \SigmaX^{-1}\}X_i )^2\big] \\
\leq & \ E \big[ \Vert X_i \Vert_{2}^{4} \big] E \big [ \Vert \HSigmaX^{-1}_{-i} - \HSigmaX^{-1} \Vert_{\op}^{2} \big] \lesssim p^{2} \frac{p}{n} = O \left( \frac{p^{3}}{n} \right).
\end{align*}
\end{proof}

\begin{lemma}\label{lem:repeatSH2}
Under Assumptions \ref{ap1}--\ref{ap3}, $\{X_i\}_{i=1}^n$ consists of $n$ copies of $p$-dimensional i.i.d.\ random vectors, using notations in the proof of Lemma \ref{lem:leave-out-matrix}, when $p=o(n)$, we have
    \begin{align*}       
        E\{f_{(0)}&(\SigmaX^{-1};(1,2))f_{(0)}(-\SigmaX^{-1};(1,2))f_{(0)}(\SigmaX^{-1};(1,3))f_{(0)}(-\SigmaX^{-1};(1,3))\} \\
        &= E\big\{f_{(1,2,3)}(\SigmaX^{-1};(1,2))f_{(1,2,3)}(-\SigmaX^{-1};(1,2))f_{(1,2,3)}(\SigmaX^{-1};(1,3))f_{(1,2,3)}(-\SigmaX^{-1};(1,3))\big\} + O(p^2),\\
       E\{f_{(0)}&(\SigmaX^{-1};(1,2))f_{(0)}(-\SigmaX^{-1};(1,2))f_{(0)}(\SigmaX^{-1};(3,4))f_{(0)}(-\SigmaX^{-1};(3,4))\} \\
        &= E\big\{f_{(1,2,3,4)}(\SigmaX^{-1};(1,2))f_{(1,2,3,4)}(-\SigmaX^{-1};(1,2))f_{(1,2,3,4)}(\SigmaX^{-1};(1,3))f_{(1,2,3,4)}(-\SigmaX^{-1};(1,3))\big\} + O(p).
    \end{align*}
        
\end{lemma}
    \begin{proof}
    We first apply Lemma \ref{lem:Sherman-Morrison} to the first term in LHS. By \eqref{eq:repeat0} in the proof of Lemma \ref{lem:repeated Sherman-Morrison}, and then we have
    \begin{align*}
        & f_{(0)} (\SigmaX^{-1};(1,2))f_{(0)}(-\SigmaX^{-1};(1,2))f_{(0)}(\SigmaX^{-1};(1,3))f_{(0)}(-\SigmaX^{-1};(1,3))\} \\
        &= f_{(1)}(\SigmaX^{-1};(1,2))f_{(1)}(-\SigmaX^{-1};(1,2))f_{(1)}(\SigmaX^{-1};(1,3))f_{(1)}(-\SigmaX^{-1};(1,3))\}\\
        &\quad - \gamma^3_{(1)}(1) \Big[f^2_{(1)}(0;(1,2))f_{(1)}(0;(1,3))\{f_{(1)}(\SigmaX^{-1};(1,3))+f_{(1)}(-\SigmaX^{-1};(1,3))\} \\&\quad\quad+ f^2_{(1)}(0;(1,3))f_{(1)}(0;(1,2))\{f_{(1)}(\SigmaX^{-1};(1,2))+f_{(1)}(-\SigmaX^{-1};(1,2))\}\Big] \\
        &\quad + \gamma^2_{(1)}(1) \Big[f_{(1)}(0;(1,2))f_{(1)}(0;(1,3))\{f_{(1)}(\SigmaX^{-1};(1,2))f_{(1)}(\SigmaX^{-1};(1,3)) \\&\quad\quad\quad+ f_{(1)}(\SigmaX^{-1};(1,2))f_{(1)}(-\SigmaX^{-1};(1,3)) + f_{(1)}(-\SigmaX^{-1};(1,2))f_{(1)}(\SigmaX^{-1};(1,3))
        \\&\quad\quad\quad  + f_{(1)}(-\SigmaX^{-1};(1,2))f_{(1)}(-\SigmaX^{-1};(1,3))\} \\
        & \quad \quad +f^2_{(1)}(0;(1,3))f_{(1)}(-\SigmaX^{-1};(1,2))f_{(1)}(\SigmaX^{-1};(1,2)) +  f^2_{(1)}(0;(1,2))f_{(1)}(-\SigmaX^{-1};(1,3))f_{(1)}(\SigmaX^{-1};(1,3))\Big] \\
        &\quad - \gamma_{(1)}(1) \Big[f_{(1)}(0;(1,2))f_{(1)}(\SigmaX^{-1};(1,3))f_{(1)}(-\SigmaX^{-1};(1,3))\{f_{(1)}(\SigmaX^{-1};(1,2))+f_{(1)}(-\SigmaX^{-1};(1,2))\} \\
        &\quad\quad+f_{(1)}(0;(1,3))f_{(1)}(\SigmaX^{-1};(1,2))f_{(1)}(-\SigmaX^{-1};(1,2))\{f_{(1)}(\SigmaX^{-1};(1,3))+f_{(1)}(-\SigmaX^{-1};(1,3))\} \Big] \\
        &\quad +\gamma^4_{(1)}(1)f^2_{(1)}(0;(1,2))f^2_{(1)}(0;(1,3)).
    \end{align*}


We can control each term by analyzing 
\begin{align*}
& \gamma^l_{(1)}(1)\cdot f_{(1,2,3)}^{k_1}(\SigmaX^{-1};(1,2)) \cdot f_{(1,2,3)}^{k_2}(-\SigmaX^{-1};(1,2)) \cdot f_{(1,2,3)}^{k_3}(0;(1,2)) \\
& \cdot f_{(1,2,3)}^{q_1}(\SigmaX^{-1};(1,3)) \cdot f_{(1,2,3)}^{q_2}(-\Sigma^{-1};(1,3)) \cdot f_{(1,2,3)}^{q_3}(0;(1,3)),
\end{align*}
where $l\in\{1,2,3,4\},k_1+k_2+k_3=2,\ q_1+q_2+q_3=2,\ k_1,k_2,q_1,q_2\neq 2,\ \{k_1,k_2,q_1,q_2\} \in \mathbb{N}$. This term can be viewed as a quadratic form of $X_1$ except when $k_1=1,q_1=0$ or $k_1=1,q_1=0$, because the matrix associated with this term is positive semi-definite except these two scenarios. When $k_1=1,q_1=0$ or $k_1=1,q_1=0$, however, we can write
\[f_{(1)}(\Sigma^{-1};(1,j)) = f_{(1)}(0;(1,j)) - X_1^{\top}\Sigma^{-1}X_j,\]
for $j=2,3$, and all the terms can be bounded by 
\begin{align*}
      &\big\vert E\{f_{(1,2,3)}^{k_1}(\SigmaX^{-1};(1,2))f_{(1,2,3)}^{k_2}(-\SigmaX^{-1};(1,2))f_{(1,2,3)}^{k_3}(0;(1,2))\\
      &\quad\times f_{(1,2,3)}^{q_1}(\SigmaX^{-1};(1,3))f_{(1,2,3)}^{q_2}(-\Sigma^{-1};(1,3))f_{(1,2,3)}^{q_3}(0;(1,3))\}\big\vert,
      \end{align*}
      because $\gamma_{(1)}(1)<1$. Then we have
    \begin{align*}
      &E\{f_{(1,2,3)}^{k_1}(\SigmaX^{-1};(1,2))f_{(1,2,3)}^{k_2}(-\SigmaX^{-1};(1,2))f_{(1,2,3)}^{k_3}(0;(1,2))\\
      &\quad\times f_{(1,2,3)}^{q_1}(\SigmaX^{-1};(1,3))f_{(1,2,3)}^{q_2}(-\Sigma^{-1};(1,3))f_{(1,2,3)}^{q_3}(0;(1,3))\}\\
      &= E\big[X_1^{\top} \{\HSigmaX^{-1}_{-(1,2,3)} - \SigmaX^{-1}\}^{k_1} \{\HSigmaX^{-1}_{-(1,2,3)} + \SigmaX^{-1}\}^{k_2}\Sigma\{\HSigmaX^{-1}_{-(1,2,3)} \}^{k_3}X_1\\
      &\quad \times X_1^{\top}\{\HSigmaX^{-1}_{-(1,2,3)} - \SigmaX^{-1}\}^{q_1} \{\HSigmaX^{-1}_{-(1,2,3)} + \SigmaX^{-1}\}^{q_2}\Sigma\{\HSigmaX^{-1}_{-(1,2,3)} \}^{q_3}X_1\big]\\
      &\leq \Vert X_1 \Vert_2^4 \cdot \Vert \HSigmaX^{-1}_{-(1,2,3)} - \SigmaX^{-1}\Vert ^{k_1+q_1}_{\op}\cdot\Vert\HSigmaX^{-1}_{-(1,2,3)} + \SigmaX^{-1}\Vert^{k_2+q_2}_{\op}\cdot \Vert\HSigmaX^{-1}_{-(1,2,3)}\Vert^{k_3+q_3}_{\op}\\
      & \lesssim M^2 p^2 \left( \sqrt{\frac{p}{n}} \right)^{k_1+q_1} \left( \frac{2}{\lambda_{\min}} \right)^{k_3+q_3} \left( \frac{3}{\lambda_{\min}} \right)^{k_2+q_2}\\
      &= O(p^2).
    \end{align*}
    The first equality follows because $X_2$ and $X_3$ are independent of $\HSigmaX^{-1}_{-(1,2,3)}$ and $X_1$. Then we obtain
    \[E\{f^2_{(0)}(\SigmaX^{-1};(1,2))f^2_{(0)}(-\SigmaX^{-1};(1,3))\} = E\big\{f^2_{(1,2,3)}(\SigmaX^{-1};(1,2))f^2_{(1,2,3)}(-\SigmaX^{-1};(1,3))\big\} + O(p^2).\]
    Finally, we analyze $f^2_{(0)}(\SigmaX^{-1};(1,2))f^2_{(0)}(-\SigmaX^{-1};(3,4))$. We are only left to control 
     \begin{align*}
    f_{(1,2,3,4)}^{k_1}(\SigmaX^{-1};& \,(1,2))f_{(1,2,3,4)}^{k_2}(0;(1,2))f_{(1,2,3,4)}^{k_3}(-\SigmaX^{-1}; (3,4))\\
    &\times f_{(1,2,3,4)}^{q_1}(\SigmaX^{-1};(1,2))f_{(1,2,3,4)}^{q_2}(0;(1,2))f_{(1,2,3,4)}^{q_3}(-\SigmaX^{-1}; (3,4)),
    \end{align*}
    where $k_1+k_2+k_3=2,\ q_1+q_2+q_3=2,\ k_1,k_2,q_1,q_2\neq 2,\ \{k_1,k_2,q_1,q_2\} \in \mathbb{N}$. As $X_2$, $X_3$, $X_4$ are independent of $\HSigmaX^{-1}_{-(1,2,3,4)}$ and $X_1$, we have 
    \begin{align*}
        E\Big[    &f_{(1,2,3,4)}^{k_1}(\SigmaX^{-1};(1,2))f_{(1,2,3,4)}^{k_2}(0;(1,2))f_{(1,2,3,4)}^{k_3}(-\SigmaX^{-1}; (3,4))\\
    &\times f_{(1,2,3,4)}^{q_1}(\SigmaX^{-1};(1,2))f_{(1,2,3,4)}^{q_2}(0;(1,2))f_{(1,2,3,4)}^{q_3}(-\SigmaX^{-1}; (3,4))\Big]\\
        &= E\big[X_1^{\top} \{\HSigmaX^{-1}_{-(1,2,3,4)} - \SigmaX^{-1}\}^{k_1+q_1}\{\HSigmaX^{-1}_{-(1,2,3,4)} \}^{k_2+q_2}\{\HSigmaX^{-1}_{-(1,2,3,4)} + \SigmaX^{-1}\}^{k_3+q_3}X_1\big]\\
        &\lesssim p\cdot \left( \sqrt{\frac{p}{n}} \right)^{k_1+q_1}\left( \frac{2}{\lambda_{\min}} \right)^{k_2+q_2}\left( \frac{3}{\lambda_{\min}} \right)^{k_3+q_3} \\
        &= O(p).
    \end{align*} 
    \end{proof}

\putbib[reference]
\end{bibunit}

\end{document}